\documentclass[11pt]{article}
\usepackage[utf8]{inputenc}
\usepackage{siunitx}
\usepackage[pagewise]{lineno}

\usepackage[top=1in, bottom=1in, left=1in, right=1in]{geometry}

\usepackage{amsmath,amsthm}
\usepackage{amssymb}
\usepackage{amsfonts}
\usepackage{enumerate}
\usepackage{rotating}
\usepackage{caption}
\usepackage{setspace}
\usepackage{footmisc}
\usepackage{graphicx}
\graphicspath{ {graphics/} }
\usepackage{subcaption}
\usepackage{ulem}
\usepackage{nameref}

\usepackage[colorinlistoftodos]{todonotes}

\usepackage{float}
\usepackage{natbib}
\usepackage{bm}
\usepackage{color}
\usepackage{tabularx}
\usepackage{multirow}
\usepackage{placeins}
\usepackage{dcolumn}
\newcolumntype{d}[1]{D{.}{.}{#1}}
\usepackage{hyperref}
\usepackage{array}
\usepackage{appendix}
\usepackage{verbatim}

\makeatletter
\renewcommand\@makefntext[1]{%
  \noindent\makebox[1em][r]{\@makefnmark}#1}
\makeatother

\allowdisplaybreaks[1]

\makeatletter
\renewcommand*{\@fnsymbol}[1]{\ifcase#1\or*\else\@alph{\numexpr#1-1\relax}\fi}
\makeatother

\def\citetpos#1{\citeauthor{#1}'s (\citeyear{#1})}

\defcitealias{FRWZ2015}{FRWZ}  %\citetalias{FRWZ2015} %\citepalias{FRWZ2015}
\hypersetup{colorlinks=true,citecolor=blue,linkcolor=red,filecolor=green,urlcolor=red}

\newcommand\fnote[1]{\smallskip\captionsetup{font=footnotesize}\caption*{#1}}

\newtheorem*{theorem}{Theorem}
\newtheorem{assumption}{Assumption}
\newtheorem{proposition}{Proposition}

\theoremstyle{remark}

\newtheorem*{algorithm}{Algorithm}

\begin{document}
% \linenumbers
% Declare new geometry for the title page only.
%\newgeometry{top=0.4in}
%---------------------------------------------
\newgeometry{top=0.3in, bottom=0.8in}

%\begin{titlepage}
\title{The unbearable lightness of equilibria in a low interest rate environment}
% \shortTitle{}
\author{Guido Ascari and Sophocles Mavroeidis\thanks{Ascari: Department of Economics and Management, University of Pavia, Via San Felice 5, 27100 Pavia, Italy and De Nederlandsche Bank, guido.ascari@unipv.it. Mavroeidis: Department of Economics, University of Oxford, Manor Road, OX1 3UQ, sophocles.mavroeidis@economics.ox.ac.uk. We would like to thank Boragan Arouba, Mikkel Plagborg-M{\o}ller, Anton Nakov, Frank Schorfheide, Sebastian Schmidt, Nathaniel Throckmorton and the participants of the NBER-EFSF meeting on methods and applications for DSGE models at the Federal Reserve Bank of Philadelphia in October 2019, the North American Summer meeting of the Econometric Society, the 27th International Conference on Computing in Economics and Finance, the 1st Sailing Macro Workshop in Ventotene, and seminar participants at the ECB, University of Oxford, for useful comments and discussion. We also thank Julian Ashwin, Angus Groom, David Murakami and Sriram Tolety for research assistance. This research is funded by the European Research Council via Consolidator grant
number 647152. Views expressed are those of the authors and do
not necessarily reflect official positions of De Nederlandsche Bank.}}
\date{\today}
% \pubMonth{}
% \pubYear{}
% \pubVolume{}
% \pubIssue{}
% \JEL{C62, E4, E52}
% \Keywords{incoherency, incompleteness, rational expectations, zero lower bound, DSGE}
\maketitle

\begin{abstract}
Structural models with no solution are incoherent, and those with multiple solutions are incomplete. We show that models with occasionally binding constraints are not generically coherent. Coherency requires restrictions on the parameters or on the support of the distribution of the shocks. In presence of multiple shocks, the support restrictions cannot be independent from each other, so the assumption of orthogonality of structural shocks is incompatible with coherency. Models whose coherency is based on support restrictions are generically incomplete, admitting a very large number of minimum state variable solutions. 



\end{abstract}

\setcounter{page}{0}
\thispagestyle{empty}
%\end{titlepage}
\pagebreak \newpage
\restoregeometry

\section{Introduction}

It is well-known that in structural models with
occasionally binding constraints, equilibria may not exist (incoherency) or
there may be multiple equilibria (incompleteness). \cite{GourierouxLaffontMonfort1980} (henceforth GLM) studied this problem in the context of simultaneous equations models with endogenous regime switching, and derived conditions for existence and uniqueness of solutions, which are known as `coherency and completeness' (CC) conditions. \cite{AruobaMlikotaSchorfheideVillalvazo2021} and \cite{Mavroeidis2019} derived these conditions for structural vector autoregressions with occasionally binding constraints. However, to the best of our knowledge, there are no general results about the conditions for existence and uniqueness of equilibria in dynamic forward-looking models with rational expectations when some variables are subject to occasionally binding constraints. This is despite the fact that there is a large and expanding literature on solution algorithms for such models \citep[see][]{FernandezRubioSchorheide2016hdk} applied for example to models with a zero lower bound (ZLB) constraint on the interest rate \citep[see e.g.,][]{FernandezGordonGuerronRubio2015,GuerrieriIacoviello2015,  AruobaCubaBordaSchorfheide2018, gustetal2017AER, AruobaCubaBordaHigaFloresSchorfheideVillalvazo2021,eggertsson2021toolkit}.

In this paper, we attempt to fill that gap in the literature. We show that the question of existence of equilibria (coherency) is a nontrivial problem in models with a ZLB constraint on the nominal interest rate. Our main finding is that, under rational expectations, coherency requires restrictions on the support of the distribution of the exogenous shocks, and these restrictions are difficult to interpret. 

The intuition for this result can be gauged from a standard New Keynesian (NK) model. Coherency of the model requires that the aggregate demand (AD) and supply (AS) curves intersect for all possible values of the shocks. If the curves are straight lines, then the model is coherent if and only if the curves are not parallel. Therefore, linear models are generically coherent. However, models with a ZLB constraint are at most piecewise linear even if the Euler equations of the agents are linearized. In those models coherency is no longer generic, because the curves may not intersect. This depends on the slope of the curves and their intercept. The former depends on structural parameters, while the latter depends on the shocks. 
%The set of values over which incoherency occurs has positive Lebesgue measure. 
In fact, many applications in the literature feature parameters and distribution of shocks that place them in the incoherency region (e.g., a monetary policy rule that satisfies the Taylor principle, structural shocks with unbounded support). Given the parameters, coherency can only be restored by restricting the support of the distribution of the shocks, so the AD and AS curves never fail to intersect. In other words, we need to exclude the possibility of sufficiently adverse shocks causing rational expectations to diverge.
%That is, we need to exclude the possibility of any sufficiently adverse shocks that would cause rational expectations to diverge. 
% However, when we restore coherency by placing support restrictions, we end up with an incomplete model that has many more solutions than typically reported in the literature. This raises questions about the properties of numerical solution algorithms.

We derive our main result first in a simple model that consists of an active Taylor rule with a ZLB constraint and a nonlinear Fisher equation with a single discount factor (AD) shock that can take two values. This setup has been used, amongst others, by \cite{EggertssonWoodford2003} and \cite{AruobaCubaBordaSchorfheide2018}, and it suffices to study the problem analytically and convey the main intuition. The main takeaway from this example is that when the Taylor rule is active, there exist no bounded fundamental or sunspot equilibria unless negative AD shocks are sufficiently small. Because this restriction on the support of the distribution of the shock is asymmetric, this finding is not equivalent to restricting the variance of the shock. 

We then turn to (piecewise) linear models, and focus on the question of existence of minimum state variable (MSV) solutions, which are the solutions that most of the literature typically focuses on. A key insight of the paper is that when the support of the distribution of the exogenous variables is discrete, these models can be cast into the class of piecewise linear simultaneous equations models with endogenous regime switching analysed by GLM. We can therefore use the main existence theorem of GLM to study their coherency properties. 
%This involves computing the determinants of the system in each of the possible $2^k$ constellations of regimes, where $k$ is the number of values the exogenous variables can take. If all determinants have the same sign, then the model is coherent and complete. 
Applying this methodology to a prototypical three-equation NK model, we find that the model is not generically coherent both when the Taylor rule is active and when monetary policy is optimal under discretion. The restrictions on the support that are needed to restore an equilibrium depend on the structural parameters as well as the past values of the state variables. When there are multiple shocks, the support restrictions are such that the shocks cannot have `rectangular' support, meaning that they cannot be independent from each other. For example, the range of values that the monetary policy shock can be allowed to take depends on the realizations of the other shocks. So, the assumption of orthogonality of structural shocks is incompatible with coherency.

When the CC condition is violated, imposing the necessary support restrictions to guarantee existence of a solution causes incompleteness, i.e., multiplicity of MSV solutions. We show that there may be up to $2^k$ MSV equilibria, where $k$ is the number of states that the exogenous variables can take.
%In the case of the NK model, the number of equilibria depends on the calibration of the model. We looked at the calibrations that define the two different cases in \cite{MertensRavn2014}, that is the ``fundamental driven'' and the ``confidence driven'' cases. In the former case there are two MSV while in the latter one there are four of them, but the number can change with different calibrations.
The literature on the ZLB stressed from the outset the possibility of multiple steady states and/or multiple equilibria, and of sunspots solutions due either to indeterminacy or to belief-driven fluctuations between the two steady states \citep[e.g.,][]{AruobaCubaBordaSchorfheide2018,MertensRavn2014}. Here, we stress a novel source of multiplicity: the multiplicity of MSV solutions.
%All the equilibria we look at are MSV solutions, thus they are not sunspots solutions due either to indeterminacy or to belief-driven fluctuations between the two steady states \citep[e.g.,][]{AruobaCubaBordaSchorfheide2018,MertensRavn2014}. 
% This is a different source of multiplicity which has heretofore not been acknowledged in the literature. 
% This novel (and subtle) source of multiplicity in the MSV solutions arises from the presence of dynamics and stochastic shocks.
%It is essential that researchers running solution algorithms that rely on MSV solutions are aware of it.

Finally, we identify possible ways out of the conundrum of incoherency and incompleteness of the NK model. These call for a different modelling of monetary policy. A first possibility would be to assume that monetary policy steps in with a different policy reaction, e.g., unconventional monetary policy (UMP), to catastrophic shocks that cause the economy to collapse. However, this policy response would need to be incorporated in the model, affecting the behavior of the economy also in normal times (i.e., when shocks are small). A more straightforward approach is to assume that UMP can relax the ZLB constraint sufficiently to restore the generic coherency of the model without support restrictions. 
%The intuition is simple. If we model the effectiveness of UMP via a shadow rate \citep[e.g.,][]{WuZhang2019,IkedaLiMavroeidisZanetti2020,Mavroeidis2019}, then we can represent the model with a kink such that the Taylor rule remains upward sloping rather than horizontal during the ZLB regime. If the slope of the Taylor rule during the ZLB regime remains greater than 1 (i.e., the rule satisfies the Taylor principle in both regimes), then coherency is restored and the MSV solution is unique and determinate. 
This underscores another potentially important role of UMP not emphasized in the literature so far: UMP does not only help take the economy out of a liquidity trap, but it is also useful in ensuring the economy does not collapse in the sense that there is no bounded equilibrium.

A number of theoretical papers provide sufficient conditions for existence of MSV equilibria in NK models \cite[see][]{egg2011nberma, Bonevaetal2016JME, Armenter2018, Christianoetal2018wpuniqueness,Nakata2018,Nakataschmidt2019JME}. Our contribution relative to this literature is to provide both necessary and sufficient conditions that can be applied more generally. \cite{Holden2021} analyses existence under perfect foresight, so his methodology is complementary to ours. \cite{Mendes2011} provides existence conditions on the variance of exogenous shocks in models without endogenous or exogenous dynamics, while %\cite{BasuBund2015} and 
\cite{richthrock2015bej} report similar findings based on simulations. Our analysis provides a theoretical underpinning of these findings and highlights that existence generally requires restrictions on the \textit{support} of the distribution of the shocks rather than their variance.%\footnote{The link of existence of equilibria to the variance of shocks requires both symmetry of the shock distribution and no persistence of the shocks. Thus, the intuition from the above studies does not carry over to more general settings, because restrictions on variance are generally neither necessary nor sufficient for a solution to exist.} 

The structure of the paper is as follows. %After a review of the related literature,
Section \ref{s: coherency problem} presents the main findings of the paper regarding the problem of incoherency (i.e., non-existence of equilibria). Section \ref{s: incompleteness} looks at the problem of incompleteness (i.e., multiplicity of MSV solutions). 
%extends the results to a prototypical three equation New Keynesian model. Section \ref{s: support restrictions} provides further analytical results on the requisite support restrictions in incomplete models using a modified version of the toy model in Section \ref{s: ACS simple example}. 
Section \ref{s: conclusions} concludes. All proofs are given in the Appendix available online.

\section{The incoherency problem}\label{s: coherency problem}

This Section illustrates the main results of the paper that concern coherency, i.e., existence of a solution, in models with a ZLB constraint. Subsection \ref{s: ACS simple example} presents the simplest nonlinear example. Subsection \ref{s: piecewise linear} turns to piecewise (log)linear models, including the three-equation NK model, and  introduces a general method for analysing their coherency properties. Subsection \ref{s: support restrictions} highlights the nature of the support restrictions needed for coherency allowing for continuous stochastic shocks, using a convenient forward-looking Taylor rule example. Subsection \ref{s: cc conditions k} derives the conditions on the Taylor rule coefficient for coherency and completeness in the simple NK model. Subsection \ref{s: UMP} shows how unconventional monetary policy can restore coherency in the NK model with an active Taylor rule. Finally, Subsection \ref{s: endog} examines the implications of endogenous dynamics.

\subsection{The incoherency problem in a simple example\label{s: ACS simple example}}
We illustrate the main results of the paper using the simplest possible model that is analytically tractable and suffices to illustrate our point in a straightforward way. It should be clear that the problem that we point out is generic and not confined to this simple setup. 
% Later on we show that similar results apply to the basic three equations New Keynesian model. 

The model is taken from Section 2 in \cite{AruobaCubaBordaSchorfheide2018} (henceforth ACS). It consists of two equations: a consumption Euler equation
\begin{equation}
1=E_{t}\left(  M_{t+1}\frac{R_{t}}{\pi_{t+1}}\right)  \label{eq: EE}%
\end{equation}
and a simple Taylor rule subject to a ZLB constraint
\begin{equation}
R_{t}=\max\left\{  1,r\pi_{\ast}\left(  \frac{\pi_{t}}{\pi_{\ast}}\right)
^{\psi}\right\}  ,\quad\psi> 1,\label{eq: Taylor}%
\end{equation}
where $R_{t}$ is the gross nominal interest rate, $\pi_{t}$ is the gross inflation rate, $\pi_{\ast}$ is the target of the central bank for the gross inflation rate, $M_{t+1}$ is the stochastic discount factor, and $r$ is the steady-state value of $1/M_{t+1}$, which is also the steady-state value of the gross real interest rate $R_t/E_t(\pi_{t+1})$. To complete the specification of the model, we need to specify the law of motion of $M_t$. 
\begin{assumption}\label{ass: M nonlinear absorbing}
$M_{t}$ is a 2-state Markov-Chain process with an absorbing state
$r^{-1}$, and a transitory state $r^{-1}e^{-r^{L}}>r^{-1}$ that persists with probability
$p>0$. 
\end{assumption}
This is a common assumption in the theoretical literature \citep[see, e.g.,][]{EggertssonWoodford2003,CEE2011JPE, egg2011nberma}. $r^{L}<0$
can be interpreted as negative real interest rate shock, which captures the
possibility of a temporary liquidity trap.

Substituting for $R_t$ in (\ref{eq: EE}) using (\ref{eq: Taylor}), we obtain
\begin{equation}
1=\max\left\{  1,r\pi_{\ast}\left(  \frac{\pi_{t}}{\pi_{\ast}}\right)
^{\psi}\right\}  E_{t}\left(  \frac{M_{t+1}}{\pi_{t+1}}\right), \quad \psi>1.  \label{eq: ACS nonlinear}%
\end{equation}
Let $\Omega_t$ denote the information set at time $t$, such that $E_t(\cdot):=E(\cdot|\Omega_t)$. In the words of \cite{Blan80}, a solution $\pi_t$ of the model is a sequence of functions of variables in $\Omega_t$ that satisfies (\ref{eq: ACS nonlinear}) for all possible realizations of these variables. Like \cite{Blan80}, we focus on bounded solutions. 
%Recall that the model is coherent if it has a solution. 

The following proposition provides, in the context of the present example, the main message of the paper, that coherency of the model (i.e., existence of a solution) requires restrictions on the support of the distribution of the state variable $M_t$. 
\begin{proposition}\label{prop: nonlinear ACS}
Under Assumption \ref{ass: M nonlinear absorbing} and $\psi>1$, a fundamental solution to (\ref{eq: ACS nonlinear}) exists if and only if the exogenous process $M_t$ satisfies the support restrictions
\begin{equation}\label{eq: support restr nonlin ACS}
r^{-1}\leq\pi_{\ast},\quad\text{and\quad}-r^{L}\leq\log\left(  \frac
{r\pi_{\ast}-1+p}{p}\right)-\frac{1}{\psi}\log\left(  r\pi_{\ast}\right).
\end{equation}
\end{proposition}
Here we sketch graphically the argument for the first of the support restrictions in (\ref{eq: support restr nonlin ACS}) in order to convey the main intuition for why a solution fails to exist when the shocks are sufficiently large. Note that the upper bound on $(-r^L)$ in (\ref{eq: support restr nonlin ACS}) is increasing in the Taylor rule coefficient $\psi$. So, for some values of the shock $(-r^L)$, the model may be coherent with a sufficiently active Taylor rule and incoherent with a less active one. Moreover, both support restrictions in (\ref{eq: support restr nonlin ACS}) become slacker as the inflation target $\pi_*$ increases. The proposition also shows that coherency does not depend on the variance of the exogenous process per se.\footnote{Raising $p$ reduces the variance, but it also reduces the upper bound for coherency on the shock $(-r^L)$ in (\ref{eq: support restr nonlin ACS}). Thus, a model with a higher variance of $M_t$ may be coherent, while a model with a lower variance of $M_t$ may be incoherent.}  

Suppose that $M_t$ is in the absorbing state $r^{-1}$. Then, there is no uncertainty in $\pi_{t+1}$ along a fundamental solution, so (\ref{eq: ACS nonlinear}) becomes a deterministic difference equation that can be represented in terms of $\hat{\pi}_t:=\log{(\pi_t/\pi_*)}$ (no approximation is involved) as
\[
\hat{\pi}_{t+1}=\max\left\{-\log{r\pi_*},\psi\hat{\pi}_{t}\right\}. 
\]
% \begin{figure}
%     \centering
%     \includegraphics[scale=0.5]{coherency2.jpg}
%     \caption{Illustration of coherency restriction $r\pi_*\geq 1$ under the absorbing state in Proposition \ref{prop: nonlinear ACS}. The red line plots $\hat{\pi}_{t+1}=\max\left\{-\log{r\pi_*},\psi\hat{\pi}_{t}\right\}$ with $\psi>1$ for two different values of $r\pi_*$. When $r\pi_*<1$, no bounded solution exists.}
%     \label{fig: coherency}
% \end{figure}
Figure \ref{fig: coherency} plots the right hand side of the above equation along with a $\ang{45}$ line. It is clear from the graph on the right that if $r\pi_*<1$, $\pi_{t+s}$ diverges for any initial value of $\pi_t$, i.e., there is no bounded solution. This is because the stable point to which $\pi_{t}$ would jump to in the absence of the constraint (i.e., the origin in the figure) violates the constraint, so it is infeasible. In contrast, when $r\pi_*\geq 1$, there exist many bounded solutions with $\pi_t \leq \pi_*$, which is a stable manifold in this case. In this simple example, $r$ corresponds to the steady-state value of the gross real interest rate, so it is fairly innocuous to assume $r\geq 1$ and the inflation target is typically nonnegative ($\pi_*\geq1$). But the same basic intuition applies in the transitory state: coherency of the model requires that the transitory shock is such that there exist stable paths which $\pi_t$ can jump to, or in other words, that the curve representing the transitory dynamics intersects with the $\ang{45}$ line, see Figure \ref{fig: ACS nonlinear} in  \ref{app: s: nonlinear ACS}. 
\begin{figure}[h!]
    \centering
    \includegraphics[scale=0.5]{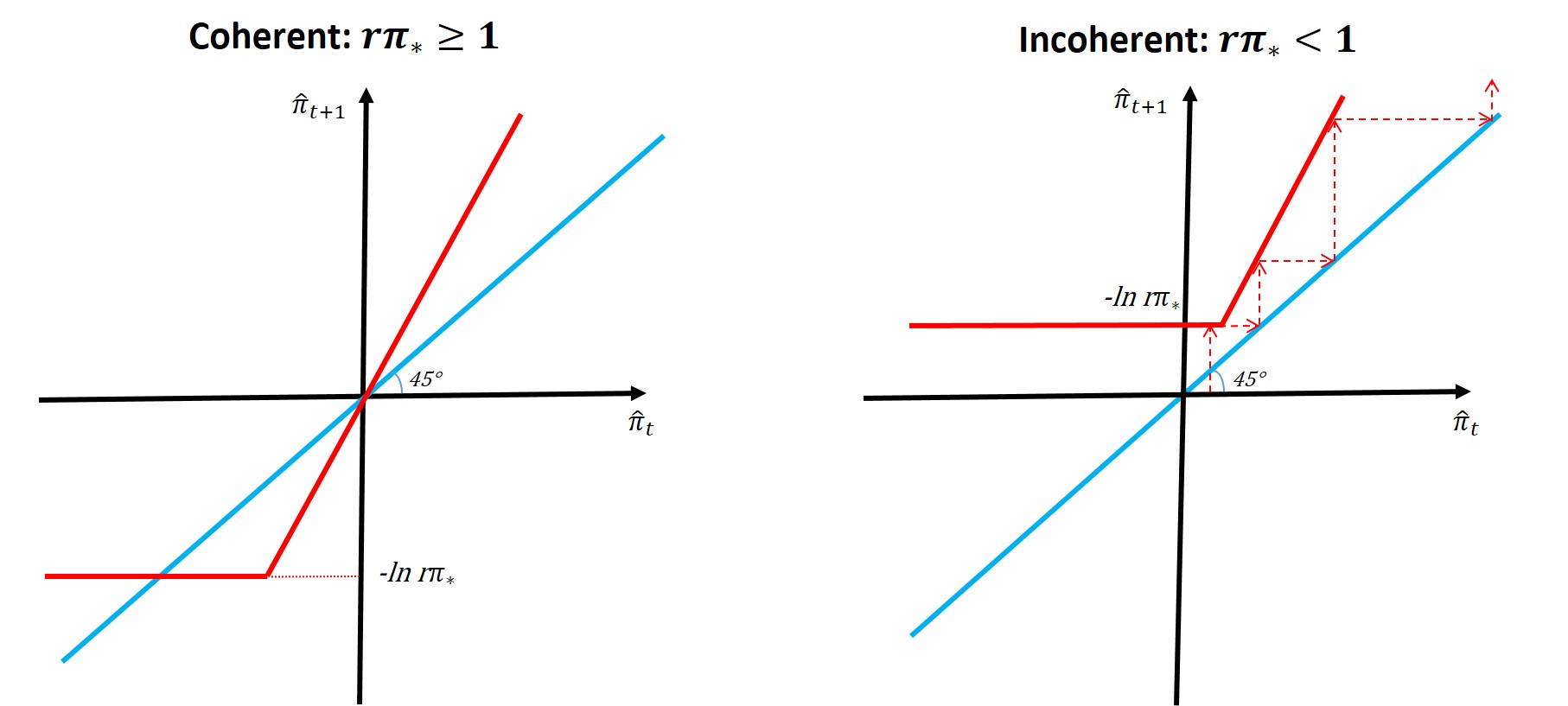}
    \caption{Illustration of coherency restriction $r\pi_*\geq 1$ under the absorbing state in Proposition \ref{prop: nonlinear ACS}. The red line plots $\hat{\pi}_{t+1}=\max\left\{-\log{r\pi_*},\psi\hat{\pi}_{t}\right\}$ with $\psi>1$ for two different values of $r\pi_*$. When $r\pi_*<1$, no bounded solution exists.}
    \label{fig: coherency}
\end{figure}

Proposition \ref{prop: nonlinear ACS} focused only on the case $\psi>1$, but it is easy to see from the proof, as well as from the argument in Figure \ref{fig: coherency}, that no support restrictions are needed when $\psi<1$: the model is always coherent when the Taylor rule is passive.

When the coherency condition in Proposition \ref{prop: nonlinear ACS} holds, the stationary solutions of the transition equations represent fundamental solutions at which $\pi_t$ depends only on $M_t$ and not on its lags. Such solutions are also known as minimum state variable (MSV) solutions in the literature, because they involve the smallest number of state variables (in this case, only one). So, for this model, the same coherency condition that is required for existence of bounded fundamental solutions is also necessary and sufficient for the existence of MSV solutions, which is a subset of all fundamental solutions. This is noteworthy because many of the solution methods in the literature focus on MSV solutions, e.g., \cite{FernandezGordonGuerronRubio2015}, \cite{richthrock2015bej}. 
%One may therefore wonder whether nonexistence of MSV solutions would be enough to conclude that the models are incoherent. In the present example, it would be.

We conclude our analysis of this simple example by considering sunspot solutions. For simplicity, we assume there are no fundamental shocks, as in \cite{MertensRavn2014}. 
%We conclude our analysis of this simple example by considering the possibility of sunspot solutions. The simplest way to analyse sunspot solutions is when there are no fundamental shocks, as in \cite{MertensRavn2014}. 
\begin{proposition}\label{prop: nonlinear ACS sunspots}
Suppose $M_t=r^{-1}$ with probability 1 and $\psi>1$, and let $\varsigma_t\in \{0,1\}$ be a first-order Markovian sunspot process that belongs to agents' information set $\Omega_t$. Sunspot solutions to (\ref{eq: ACS nonlinear}) exist if and only if $r^{-1}\leq\pi_{\ast}$.
\end{proposition}
Proposition \ref{prop: nonlinear ACS sunspots} shows that the support restriction for the existence of sunspot solutions is exactly the same as for the existence of fundamental solutions (see the condition corresponding to the absorbing state in Proposition \ref{prop: nonlinear ACS}). Thus allowing for sunspot equilibria does not alter the essence of the coherency problem, as we further show in the next subsection.

\subsection{Checking coherency of piecewise linear models} \label{s: piecewise linear}

Many of the solution methods in the literature apply to (log)linear models, whose only nonlinearity arises from the lower bound constraint on interest rates, e.g., \cite{EggertssonWoodford2003}, \cite{GuerrieriIacoviello2015}, \cite{KulishMorleyRobinson2017},  \cite{Holden2021}.\footnote{These models are often motivated as (log)linear approximations to some originally nonlinear model under the assumption that the equilibria of the linear model are close to the equilibria of the original nonlinear model \citep[see][]{Bonevaetal2016JME,eggsingh19jedc}. This assumption implicitly imposes conditions for the existence of these equilibria. The coherency of the approximating linear model is therefore a necessary precondition that needs to be checked.} Let $Y_{t}$ be a $n\times1$ vector of endogenous variables, $X_{t}$ be a
$n_x\times1$ vector of exogenous state variables, which could include a sunspot shock whose coefficients in the model are zero, $Y_{t+1|t}:=E\left(
Y_{t+1}|\Omega_{t}\right)  $, $X_{t+1|t}:=E\left(  X_{t+1}|\Omega_{t}\right)
,$ and $s_{t}\in\left\{  0,1\right\}  $ an indicator variable that takes the
value 1 when some inequality constraint is slack and zero otherwise. We
consider models that can be written in the canonical form%
\begin{equation}%
\begin{tabular}
[c]{l}%
$A_{s_{t}}Y_{t}+B_{s_{t}}Y_{t+1|t}+C_{s_{t}}X_{t}+D_{s_{t}}X_{t+1|t}=0$\\
$s_{t}=1_{\left\{  a^{\prime}Y_{t}+b^{\prime}Y_{t+1|t}+c^{\prime}%
X_{t}+d^{\prime}X_{t+1|t}>0\right\}  },$%
\end{tabular}
\label{eq: canon}
\end{equation}
where $A_{s},B_{s},C_{s},D_{s}$ are coefficient matrices, $a,b,c,d$ are
coefficient vectors and $1_A$ is the indicator function that takes the value 1 if $A$ holds and zero otherwise.\footnote{Although we focus on a single inequality constraint, the methodology we discuss here readily applies to more than one constraints. An example of a model with an additional ZLB on inflation expectations \citep{GorodnichenkoSergeyev21} is discussed in Appendix 
\ref{app: s: ZLB expectations}.} 

\paragraph{Example ACS}\label{ex: ACS}Taking a log-linear approximation of
(\ref{eq: EE}) around $M_{t}=r^{-1}$ and $\pi_{t}=\pi_{\ast}$ we obtain
$\hat{\pi}_{t+1|t}=\hat{R}_{t}+\hat{M}_{t+1|t},$ where $\hat{\pi}_{t}%
:=\log\left(  \pi_{t}/\pi_{\ast}\right)  ,$ $\hat{M}_{t}:=\log\left(
rM_{t}\right)  ,$ $\hat{R}_{t}:=\log R_{t}-\mu,$ $\mu:=\log\left(  r\pi_{\ast
}\right)  .$ Taking logs of (\ref{eq: Taylor}) (no approximation) yields
$\hat{R}_{t}=\max\left\{  -\mu,\psi\hat{\pi}_{t}\right\}  $ and combining the
two equations yields $\hat{\pi}_{t+1|t}-\hat{M}_{t+1|t}-\max\left\{  -\mu,\psi\hat{\pi}_{t}\right\}
=0.$
The regime indicator is $s_{t}=1_{\left\{  \psi\hat{\pi}_{t}+\mu>0\right\}
}.$ This model can be put in the canonical form
(\ref{eq: canon}) with $Y_{t}=\hat{\pi}_{t},$ $X_{t}=\left(  \hat{M}%
_{t},1\right)  ^{\prime},$ $A_{0}=0,$ $A_{1}=-\psi,$ $B_{0}=B_{1}=1,$
$C_{0}=\left(  0,\mu\right)  ,$ $C_{1}=\left(  0,0\right)  ,$ $D_{0}%
=D_{1}=\left(  -1,0\right)  $, $a=\psi,$ $b=0,$ $c=\left(  0,\mu\right)
^{\prime}$ and $d=\left(  0,0\right)  ^{\prime}$. 
\openbox

\paragraph{Example NK-TR}\label{ex: NK}The basic three-equation New Keynesian model, consisting of a Phillips curve, an Euler equation and a Taylor rule, is
\begin{subequations}\label{eq: NK}%
\begin{align}\label{eq: NK NKPC}
\hat{\pi}_{t} & =\beta\hat{\pi}_{t+1|t}+\lambda\hat{x}_{t}+u_{t}    \\
\hat{x}_{t} & =\hat{x}_{t+1|t}-\sigma\left(  \hat{R}_{t}-\hat{\pi}%
_{t+1|t}\right)  +\epsilon_{t} \label{eq: NK EE} \\
\hat{R}_{t} & =\max\left\{  -\mu,\psi\hat{\pi}_{t}+\psi_{x}\hat{x}_{t}%
+\nu_t\right\} \label{eq: NK TR}
\end{align}
\end{subequations}
where $\hat{\pi}_{t},\hat{R}_{t}$ were defined in the previous example and
$\hat{x}_{t}$ is the output gap. It can be put in the canonical form
(\ref{eq: canon}) with $Y_{t}=\left(  \hat{\pi}_{t},\hat{x}_{t}\right)
^{\prime},$ $X_{t}=\left(  u_{t},\epsilon_{t},\nu_t,1\right)
^{\prime}$, and coefficients given in \ref{app: s: coeffcanonical}.
\openbox

\paragraph{Example NK-OP}\label{ex: NK-OP}The NK model with optimal discretionary policy replaces (\ref{eq: NK TR}) with
%is given by equations (\ref{eq: NK NKPC}), (\ref{eq: NK EE}) and 
\begin{equation}
    \gamma \hat{x}_t + \lambda \hat{\pi}_t = 0, \quad \text{if} \quad \hat{R}_t > -\mu, \quad \text{or} \quad \gamma \hat{x}_t + \lambda \hat{\pi}_t < 0, \quad \text{if} \quad \hat{R}_t = -\mu, \label{eq: NK OP}
\end{equation}
where $\gamma\geq0$ is the weight the monetary authority attaches to output stabilization relative to inflation stabilization, see \cite{Armenter2018}, \cite{Nakata2018} or \cite{Nakataschmidt2019JME} for details. Substituting for $\hat{R}_t=-\mu$ in (\ref{eq: NK EE}) when the ZLB binds, the model can be written in terms of two equations: (\ref{eq: NK NKPC}) and $\hat{x}_t = (1-s_t)\left[\hat{x}_{t+1|t}-\sigma\left(-\mu-\hat{\pi}%
_{t+1|t}\right)  +\epsilon_{t} \right]-s_t \frac{\lambda}{\gamma}\hat{\pi}_t$, where $s_{t}=1_{\left\{\hat{\pi}_{t+1|t}+\frac{\hat{x}_{t+1|t} -\hat{x}_{t}  +\epsilon_{t}}{\sigma}+\mu>0\right\}  }$.
This can be put in the canonical representation (\ref{eq: canon}) with $Y_{t}=\left(  \hat{\pi}_{t},\hat{x}_{t}\right)
^{\prime}$, $X_{t}=\left(  u_{t},\epsilon_{t},1\right)$, and coefficients given in \ref{app: s: coeffcanonical}.
\openbox\medskip

A special case of (\ref{eq: canon}) without expectations of the endogenous
variables, i.e., $B_{0}=B_{1}=0$ and $b=0,$ is a piecewise linear simultaneous
equations model with endogenous regime switching, whose coherency was analysed
by GLM. We will now show how (\ref{eq: canon}) with expectations can be cast into the model analysed by GLM when the shocks are Markovian with discrete support. This is a key insight of the paper.

Without much loss of generality, we assume that the state variables
$X_{t}$ are first-order Markovian. We also focus on the existence of MSV solutions that can be represented as
$Y_{t}=f\left(  X_{t}\right)  $ for some function $f\left(  \cdot\right)  .$
Therefore, from now on, coherency of the model (\ref{eq: canon}) is understood to mean existence of some function $f\left(  \cdot\right)  $ such that
$Y_{t}=f\left(  X_{t}\right)  $ satisfies (\ref{eq: canon}). 

Assume that $X_{t}$ can be represented as a $k$-state stationary first-order
Markov chain process with transition matrix $K$, and collect all the possible
states $i=1,...,k$ of $X_{t}$ in a $n_x\times k$ matrix $\mathbf{X}$. Let
$e_{i}$ denote the $i$th column of $I_k$, the identity matrix of dimension $k$, so
that $\mathbf{X}e_{i}$ -- the $i$th column of $\mathbf{X}$ -- is the $i$th
state of $X_{t}$. Note that the elements of the transition kernel are $K_{ij} = \Pr\left(X_{t+1} = \mathbf{X}e_j|X_{t} = \mathbf{X}e_i\right)$ and hence, $E\left(  X_{t+1}|X_{t}=\mathbf{X}e_{i}\right)
=\mathbf{X}K^{\prime}e_{i}.$ Let $\mathbf{Y}$ denote the $n\times k$ matrix
whose $i$th column, $\mathbf{Y}e_{i}$, gives the value of $Y_{t}$ that
corresponds to $X_{t}=\mathbf{X}e_{i}$ along a MSV solution. Therefore,
along a MSV solution we have $E\left(  Y_{t+1}|Y_{t}=\mathbf{Y}e_{i}\right)
=E\left(  Y_{t+1}|X_{t}=\mathbf{X}e_{i}\right)  =\mathbf{Y}K^{\prime}e_{i}.$
Substituting into (\ref{eq: canon}), $\mathbf{Y}$ must satisfy the following
system of equations
\begin{align}
0  & =\left(  A_{s_{i}}\mathbf{Y}+B_{s_{i}}\mathbf{Y}K^{\prime}+  C_{s_{i}}\mathbf{X}+D_{s_{i}}\mathbf{X}K^{\prime}\right)
e_{i}\label{eq: canon i}\\
s_{i}  & =1_{\left\{  \left(  a^{\prime}\mathbf{Y}+b^{\prime}\mathbf{Y}%
K^{\prime}+  c^{\prime}\mathbf{X}+d^{\prime}%
\mathbf{X}K^{\prime}\right)  e_{i}>0\right\}  },\quad i=1,...,k. \nonumber
\end{align}
This system of equations can be expressed in the form $F\left(  \mathbf{Y}%
\right)  =\kappa\left(  \mathbf{X}\right)  $, where $\kappa\left(
\cdot\right)  $ is some function of $\mathbf{X,}$ and $F\left(  \cdot\right)
$ is a piecewise linear continuous function of $\mathbf{Y}$. Specifically, let
$J$ be a subset of $\left\{  1,...,k\right\}  .$ Then, we can write $F\left(
\cdot\right)  $ as
\begin{equation}
F\left(  \mathbf{Y}\right)  =\sum_{J}\mathcal{A}_{J}1_{\mathcal{C}_{J}%
}vec\left(  \mathbf{Y}\right)  ,\label{eq: F}%
\end{equation}
where $\mathcal{C}_{J}=\left\{  \mathbf{Y}:\mathbf{Y\in\Re}^{n\times k}%
,s_{i}=1_{\left\{  i\in J\right\}  }\right\}  $ is defined by a particular
configuration of regimes over the $k$ states given by $J$.
%The question of coherency is whether the piecewise linear function $F\left(  \cdot\right)$ in (\ref{eq: F}) is invertible.
If the piecewise linear function $F\left(  \cdot\right)$ in (\ref{eq: F}) is invertible, then the system is coherent. This can be checked using Theorem 1 from GLM reproduced below. 

\begin{theorem}[GLM]\label{th: GLM}
Suppose that the mapping $F\left(  \cdot\right)  $ defined in
(\ref{eq: F}) is continuous. A necessary and sufficient condition for
$F\left(  \cdot\right)  $ to be invertible is that all the determinants
$\det\mathcal{A}_{J},$ $J\subseteq\left\{  1,...,k\right\}  $ have the same
sign.\footnote{We only need to
check the determinants over all $2^{k}$ subsets of $\left\{  1,...,k\right\}
$ rather than $2^{nk}$ subsets of $\left\{  1,...,nk\right\}  ,$ because the
$A_{J}$ will be the same for all $n$-dimensional blocks of $vec\left(
\mathbf{Y}\right)  $ that belong to the same state $i=1,...,k.$ }
\end{theorem}

The above determinant condition is straightforward to check. If the condition
is satisfied, then the model has a unique MSV solution. If the condition
fails, the model is not generically coherent, meaning that there will be
values of $\mathbf{X}$ for which no MSV solution exists. Since $\mathbf{X}$
represents the support of the distribution of $X_{t}$, violation of the
coherency condition in the \nameref{th: GLM} Theorem means that a MSV solution can only be
found if we impose restrictions on the support of the distribution of the
exogenous variables $X_{t}$. 

\paragraph{\nameref{ex: ACS} continued}
Suppose $\hat{M}_{t}$ follows a two-state Markov Chain with transition kernel $K$, there are four possible subsets of $\left\{  1,2\right\}$.
%, and $K_{11}=p$ and $K_{22}= q$.
Let PIR refer to a positive interest rate state when the ZLB constraint is slack
and ZIR to a zero interest rate state when the ZLB constraint binds. Given $e_1:=(1,0)'$, $e_2:=(0,1)'$, the coefficients of (\ref{eq: F}) are
% $J=\left\{
% 1,2\right\}$: the constraint is slack in both states (PIR,PIR); $J=\left\{
% 2\right\}$: the constraint binds in the first state only (ZIR,PIR); $J=\left\{
% 1\right\}$: the constraint binds in the second state only (PIR,ZIR); and
% $J=\varnothing$: the constraint binds in both states (ZIR,ZIR). Since $Y_{t}$
% is a scalar ($n=1),$ the coefficient matrices in each case are%
\begin{equation}%
\begin{tabular}
[c]{lll}%
$\mathcal{A}_{J_{1}}=A_{1}I_{2}+B_{1}K,$ & $J_{1}=\left\{  1,2\right\}  $ &
$\text{(PIR,PIR)}$\\
$\mathcal{A}_{J_{2}}=%
e_1e_1' (A_0 I_2+B_0K) + e_2 e_2'(A_1 I_2 + B_1 K)
% \begin{pmatrix}
% A_{0} & 0\\
% 0 & A_{1}%
% \end{pmatrix}
% +%
% \begin{pmatrix}
% B_{0}p & B_{0}\left(  1-p\right)  \\
% B_{1}\left(  1-q\right)   & B_{1}q
% \end{pmatrix}
% 
$, 
& $J_{2}=\left\{  2\right\}  $ & $\text{(ZIR,PIR)}$\\
$\mathcal{A}_{J_{3}}=%
e_2e_2' (A_0 I_2+B_0K) + e_1 e_1'(A_1 I_2 + B_1 K)
% \begin{pmatrix}
% A_{1} & 0\\
% 0 & A_{0}%
% \end{pmatrix}
% +%
% \begin{pmatrix}
% B_{1}p & B_{1}\left(  1-p\right)  \\
% B_{0}\left(  1-q\right)   & B_{0}q
% \end{pmatrix}
,$ & $J_{2}=\left\{  1\right\}  $ & $\text{(PIR,ZIR)}$\\
$\mathcal{A}_{J_{4}}=A_{0}I_{2}+B_{0}K,$ & $J_{4}=\varnothing$ & (ZIR,ZIR)
\end{tabular}
\label{eq: A matrices ACS}%
\end{equation}
where, as we showed previously, $A_{0}=0,$ $A_{1}=-\psi,$ and $B_{0}=B_{1}=1.$
%\openbox
%\paragraph{\nameref{ex: ACS} continued} 
From (\ref{eq: A matrices ACS}), we obtain $\det\mathcal{A}%
_{J_{1}}=\left(  \psi-1\right)  (  1-p\allowbreak-q+\psi)  ,$ $\det
\mathcal{A}_{J_{2}}=p\left(  1-\psi\right)  +q-1,$ $\det\mathcal{A}_{J_{3}%
}=p-1+q\left(  1-\psi\right)  ,$ $\det\mathcal{A}_{J_{4}}\allowbreak=p+q-1.$ We focus on
the case $\psi>1.$ Since $0\leq p,q\leq1,$ it follows immediately that
$\det\mathcal{A}_{J_{1}}$ is positive while $\det\mathcal{A}_{J_{2}}$ and
$\det\mathcal{A}_{J_{3}}$ are both negative, so the coherency condition in the
\nameref{th: GLM} Theorem is violated. 
\openbox\bigskip

The next proposition states that the conclusion that an active Taylor rule leads to a model that is not generically coherent generalizes to the basic three-equation NK model. The following one states that the same conclusion applies to a NK model with optimal policy.

\begin{proposition} \label{prop: NK-TR CC}
The NK-TR model given by equations (\ref{eq: NK NKPC}) with $u_t=0$, (\ref{eq: NK EE}) with $\epsilon_{t}$ following a two-state Markov chain process, and the active Taylor rule (\ref{eq: NK TR}) with $\psi>1$ and $\psi_{x}=\nu
_t=0$, is not generically coherent.\footnote{The assumption $\psi _{x}=0$
in the Taylor rule is imposed to simplify the exposition. The conclusion that the model is not generically coherent when it satisfies the Taylor
principle can be extended to the case $\psi _{x}\neq 0
$, when the Taylor principle becomes $\psi +\frac{\beta -1}{\lambda }\psi
_{x}>1$, see the proof of the Proposition for further discussion.} 
\end{proposition}

\begin{proposition} \label{prop: NK-OP CC}
The NK-OP model given by equations (\ref{eq: NK NKPC}) with $u_t=0$, (\ref{eq: NK EE}) with $\epsilon_{t}$ following a two-state Markov chain process, and the optimal discretionary policy (\ref{eq: NK OP}) is not generically coherent. 
\end{proposition}
Proposition \ref{prop: NK-OP CC} proves that there are values of the shocks for which no MSV equilibrium exists, thus formally corroborating the numerical findings in \cite{Armenter2018} about non-existence of Markov-perfect equilibria (which we call MSV solutions) in the NK-OP model.

% \begin{proposition} \label{prop: NK CC}
% The NK model given by equations (\ref{eq: NK NKPC}) with $u_t=0$, (\ref{eq: NK EE}) with $\epsilon_{t}$ following a two-state Markov chain process, and either (i) an active Taylor rule (\ref{eq: NK TR}) with $\psi>1$ and $\psi_{x}=\nu
% _t=0$, or (ii) optimal discretionary policy (\ref{eq: NK OP}) is not generically coherent. 
% \end{proposition}

Analogously to Proposition \ref{prop: nonlinear ACS} in the previous subsection, we can characterize the support restrictions for existence of a solution in the special case given by Assumption \ref{ass: M nonlinear absorbing}, such that $p<1$ (transitory state) and $q=1$ (absorbing state), with support of
$\hat{M}_{t}$ equal to $(-r^{L})$ and $0$, respectively.
%(see %Appendix \ref{app: s: prop NK-TR sup res} and \ref{app: s: prop NK-OP sup res} for the analytical derivations). 

\begin{proposition} \label{prop: NK-TR sup res}
Consider the NK-TR model of Proposition \ref{prop: NK-TR CC}. Suppose further that $\epsilon_{t}=-\sigma \hat{M}_{t+1|t}$, where $M_t$ satisfies Assumption \ref{ass: M nonlinear absorbing}, and define $\theta:=\frac{\left(  1-p\right)  \left(1-p\beta\right)}{p\sigma\lambda}$. A MSV solution exists if and only if
\begin{subequations}
\label{eq: supp restr NK}
\begin{align}\label{eq: supp restr NK E}
    \text{either}\quad & \theta > 1  \text{ and } r^{-1}\leq\pi_{\ast}, \\
    \text{or} \qquad & \theta \leq 1, \text{ } r^{-1}\leq\pi_{\ast} \text{ and } -r^L \leq \log(r\pi_*)\left(\frac{\psi-p}{\psi p}+\frac{\theta}{\psi}\right). \label{eq: supp restr NK B}
\end{align}
\end{subequations}
\end{proposition}

Figure \ref{fig: nk_psilargerthanp} helps to grasp the economic intuition. The $AD$ curve is piecewise linear depending on whether the economy is at the ZLB ($AD^{ZLB}$) or monetary policy follows the Taylor rule ($AD^{TR}$). The negative shock shifts the $AD$ curve to the left. In the transitory state, there are four possibilities
depending on the value of $\theta$, and on the equilibrium in the absorbing
state, which can be either a PIR one or a ZIR one (see Appendix 
\ref{app: s: prop NK-TR sup res}). 
%(see %Appendix 
%\ref{app: s: prop NK-TR sup res} for the analytical derivations).
\begin{figure}[h]
\centering
\includegraphics[scale = 0.45]{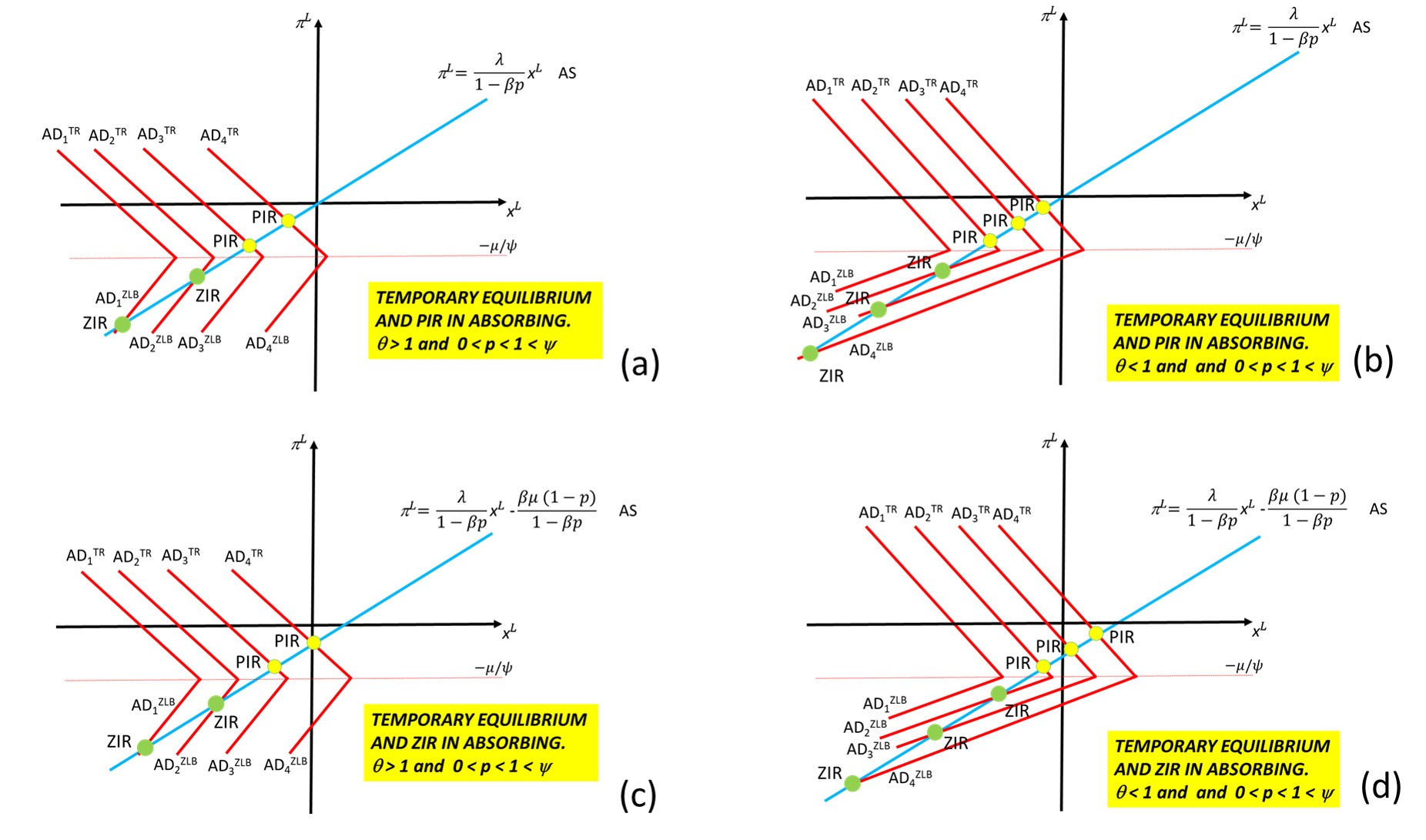}
\caption{The temporary state in the NK model when $\psi>1$.} 
\label{fig: nk_psilargerthanp}
\end{figure}

When $\theta >1$, the $AS$ is flatter than $AD^{ZLB}$, and the $AS-AD$ system is described by the curves plotted in the left column of Figure \ref{fig: nk_psilargerthanp} for the two cases when the absorbing state is PIR on the top, i.e., panel (a), and  when the absorbing state is ZIR on the bottom, i.e., panel (c). Inspection of these two graphs shows there is always a solution in both cases. Hence, when $\theta >1,$ the only necessary support
restriction is $\left( r\pi _{\ast }\right) ^{-1}\leq 1$, which guarantees the existence of an equilibrium in the absorbing state, as stated in (\ref{eq: supp restr NK E}).\footnote{$\theta>1$ exactly corresponds to condition C2 in Proposition 1 of \cite{egg2011nberma}. Figure \ref{fig: nk_psilargerthanp} provides a visual and intuitive interpretation of the coherency condition in these two sub-cases related to the analysis presented in \cite{egg2011nberma} and \cite{Bilbiie2019neofisher} for the NK-TR model.} 
Next, turn to the case $\theta \leq 1.$ The $AS$ is steeper than $AD^{ZLB},$ and the $AS-AD$ system is described by the curves plotted in the right column of Figure \ref{fig: nk_psilargerthanp} for the two cases when the absorbing state is PIR on the top, i.e., panel (b), and  when the absorbing state is ZIR on the bottom, i.e., panel (d). Clealry, a further support restriction is needed on the value of the shock in the transitory state to avoid the $AD$ curve being completely above the $AS$ curve. Intuitively, the negative shock cannot be too large (in absolute value) for an equilibrium (actually two equilibria in this case) to exist. This is what the second condition (\ref{eq: supp restr NK B}) guarantees.

\begin{proposition} \label{prop: NK-OP sup res}
Consider the NK-OP models of Proposition \ref{prop: NK-OP CC}. Suppose further that $\epsilon_{t}=-\sigma \hat{M}_{t+1|t}$, where $M_t$ satisfies Assumption \ref{ass: M nonlinear absorbing}, and define $\theta:=\frac{\left(  1-p\right)  \left(1-p\beta\right)}{p\sigma\lambda}$. A MSV solution exists if and only if
\begin{subequations}
\label{eq: supp restr NK-OP}
\begin{align}\label{eq: supp restr NK-OP E}
    \text{either}\quad & \theta > 1  \text{ and } r^{-1}\leq\pi_{\ast}, \\
    \text{or} \qquad & \theta \leq 1, \text{ } r^{-1}\leq\pi_{\ast} \text{ and } -r^L \leq \frac{\log(r\pi_*)}{p}. \label{eq: supp restr NK-OP B}
\end{align}
\end{subequations}
\end{proposition}

Figure \ref{fig: nk-OP_temp} helps to grasp the economic intuition. The $AD$ curve is piecewise linear depending on whether the economy is at the ZLB ($AD^{ZLB}$) or monetary policy follows the optimal rule ($AD^{OP}$). The negative shock shifts the $AD^{ZLB}$ curve upward. In the transitory state, there are four possibilities.
When $\theta >1$, the $AS$ is flatter than $AD^{ZLB}$ and the relevant plots are panel (a) and (c) on the left column, depending whether agents expect the absorbing state to be a PIR or a ZIR. There is always a solution in both cases, so again we conclude that when $\theta >1,$ the only necessary support
restriction is $\left( r\pi _{\ast }\right) ^{-1}\leq 1$, as stated in (\ref{eq: supp restr NK-OP E}). Instead, when $\theta \leq 1$, the $AS$ is steeper than $AD^{ZLB},$ and the relevant plots are the one on the right column. There exists an equilibrium (actually two equilibria) if and only if $(-r^{L})$ is below a threshold level, which is given by $\frac{\mu}{p},$ as stated by (\ref{eq: supp restr NK-OP B}) (see Appendix 
\ref{app: s: prop NK-OP sup res}).
\begin{figure}[htb]
\centering\includegraphics[scale = 0.5]{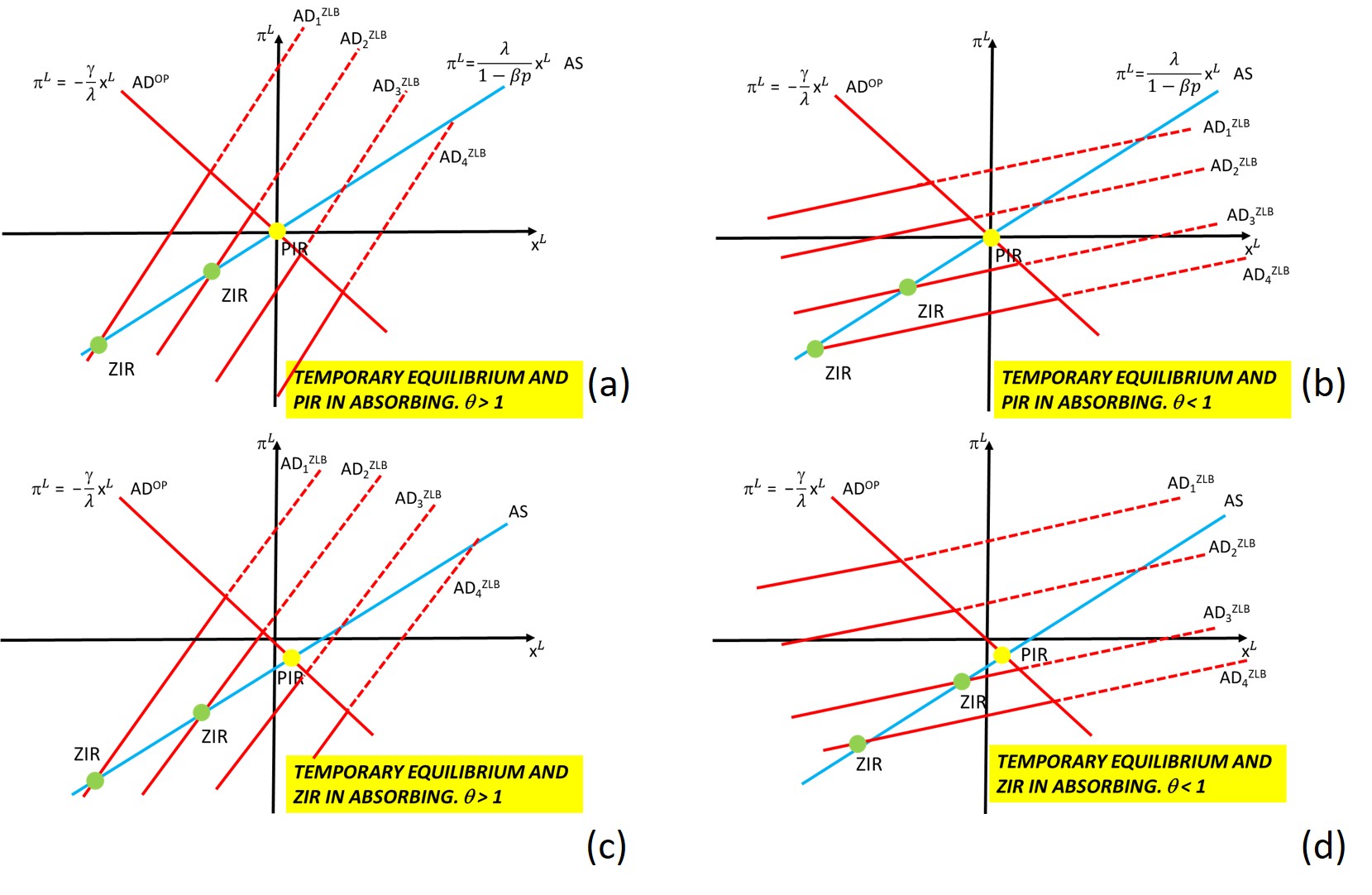}
\caption{The temporary state in the NK-OP model.} 
\label{fig: nk-OP_temp}
\end{figure}

The literature on confidence-driven equilibria \citep[e.g.,][]{MertensRavn2014} emphasised the possibility of sunspots when $\theta\leq1$. 
Propositions  \ref{prop: NK-TR sup res} and \ref{prop: NK-OP sup res}  do not consider sunspot equilibria. Existence of sunspot equilibria can be examined by including a sunspot shock in the exogenous state variables $X_t$. For example, analogously to Proposition \ref{prop: nonlinear ACS sunspots}, it can be shown that if the discount factor shock $M_t$ takes a single value $1/r$ and we allow for a binary sunspot process $\varsigma_t$ with transition matrix $K_\varsigma$, the NK-TR model with $\psi>1$ is not generically coherent, and an equilibrium exists if and only if $r^{-1}\leq \pi_*$.\footnote{This is true for any transition matrix $K_\varsigma$, i.e., not confined to the case when one of the states of the sunspot shock is absorbing as in \cite{MertensRavn2014}, see \ref{app: s: sunspots} for details.}

In the case of the NK-OP model, \cite{Nakata2018} and \cite{Nakataschmidt2019JME} consider only the case when the ZLB always binds in the `low' state, corresponding to $-pr^L>\log (r\pi_*)$. Because this excludes (\ref{eq: supp restr NK-OP B}), the condition for existence of a Markov-perfect (MSV) equilibrium given in \citet[Prop.~1]{Nakataschmidt2019JME} corresponds to (\ref{eq: supp restr NK-OP E}), which they express as a restriction on the transition probabilities, equivalent to $\theta>1$ in (\ref{eq: supp restr NK-OP E}), see Appendix 
\ref{app: NS} for details. Therefore, Proposition \ref{prop: NK-OP sup res} corroborates and extends the existence results in \cite{Nakataschmidt2019JME} by highlighting that existence requires restrictions on the values the shocks can take (support) rather than on the transition probabilities.
%\footnote{Like most of the related literature, \cite{Nakataschmidt2019JME} focus on the condition for the existence of one particular type of equilibrium, when the constraint is binding (ZIR) in the bad state and not binding (PIR) in the good state. However, the model admits, under the conditions (\ref{eq: supp restr NK-OP}) in Proposition \ref{prop: NK-OP sup res} other equilibria, see Section \ref{s: incompleteness}.} 

Finally, note that as $\sigma$ gets large, $\theta$ goes to zero and condition (\ref{eq: supp restr NK}) reduces to 
\begin{equation}\label{eq: support restr lin ACS}
r^{-1}\leq\pi_{\ast},\quad\text{and\quad}-r^{L}\leq\log\left(  r\pi_{\ast}\right)  \frac{\psi-p}{\psi p},
\end{equation}
which is the support restriction for \nameref{ex: ACS}.\footnote{The first inequality in (\ref{eq: support restr lin ACS}) is identical to the corresponding condition in (\ref{eq: support restr nonlin ACS})
in Proposition \ref{prop: nonlinear ACS} for existence of a fundamental solution in the nonlinear ACS model. This is
not surprising because in the absorbing state the two models are identical --
no approximation is involved. The second condition is approximately the same as the
corresponding second inequality in (\ref{eq: support restr nonlin ACS}) when $r\pi_*$ is close to 1. The right hand sides of the two inequalities differ by $\log\left(\frac{r\pi_*-1}{p}+1\right)-\frac{\log(r\pi_*)}{p}$, which is zero to a first-order approximation around $r\pi_*=1$.}

\subsection{More about the nature of the support restrictions\label{s: support restrictions}}
%We have seen above that it is difficult to derive all the equilibria and characterize the requisite support restrictions in incomplete models when the distribution of the shocks is continuous. 
To shed some further light on the nature of the support restrictions, we consider a modification of \nameref{ex: ACS} that allows us to characterize the support restrictions analytically even when there are multiple shocks and the distribution of the shocks is continuous. Specifically, we replace the contemporaneous Taylor rule (\ref{eq: Taylor}) with a purely forward-looking one that also includes a monetary policy shock $\nu_t$. In log-deviations from steady state, the forward-looking Taylor rule is $\hat{R}_t=\max(-\mu,\psi\hat{\pi}_{t+1|t}+\nu_t)$. Substituting for $\hat{\pi}_{t+1|t}=\hat{R}_t+\hat{M}_{t+1|t}$ from the log-linear Fisher equation, we obtain the univariate equation
\begin{equation}
\hat{R}_{t}   =\max\left\{  -\mu  ,\psi\hat{R}_t + \psi\hat{M}_{t+1|t}+\nu_t\right\}.\label{eq: fig5}
\end{equation}
% This model fits into the canonical form (\ref{eq: canon}) with $Y_t=\hat{R}_t$, $X_t=(\hat{M}_{t},\nu_t,1)'$, $A_0 =1$, $A_1=1-\psi$, $B_0=B_1=0$, $C_0=(0,0,\mu)$, $C_1=(0,-1,0)$, $D_0=0_{1\times3}$, $D_1=(-\psi,0,0)$, $a=\psi$, $b=0$,$c=(0,1,\mu)^{\prime}$ and $d=(\psi,0,0)^{\prime}$.

% \begin{figure}
% \centering
% \includegraphics[scale=0.5]{fig5_new.jpg}
% \caption{Illustration of the restriction on the support of
% $\hat{M}_{t+1|t},\nu_t$ in the model given by the intersection between the LHS of (\ref{eq: fig5}), blue line, and the RHS of (\ref{eq: fig5}), red line.}
% \label{fig: incompleteness}
% \end{figure}

The advantage of a forward-looking Taylor rule in this very simple model is that it allows us to substitute out the expectations of the endogenous variable $\hat{\pi}_{t+1|t}$, and therefore obtain an equation that is immediately piecewise linear in the remaining endogenous variable $\hat{R}_t$. Application of the \nameref{th: GLM} Theorem then shows that the model is not generically coherent when $\psi>1$. This is shown graphically in Figure \ref{fig: incompleteness}, where the left-hand side (LHS) and right-hand side (RHS) of (\ref{eq: fig5}) are shown in blue and red, respectively. When $\psi>1$, (\ref{eq: fig5}) may have no solution, an example of which is shown on the left graph in Figure \ref{fig: incompleteness}; or it may have two solutions, which is shown in the graph on the right in Figure \ref{fig: incompleteness}. The latter graph also highlights the range of values of the shocks corresponding to incoherency -- when the positively sloped part of the red curve lies in the grey area, and the ones for which two solutions exist -- when the positively sloped part of the red curve lies to the right of the grey area.
The support restrictions required for existence of a solution are $\psi\hat{M}_{t+1|t}+\nu_{t} \leq (\psi-1)\mu$.
\begin{figure}[htb]
\centering
\includegraphics[scale=0.5]{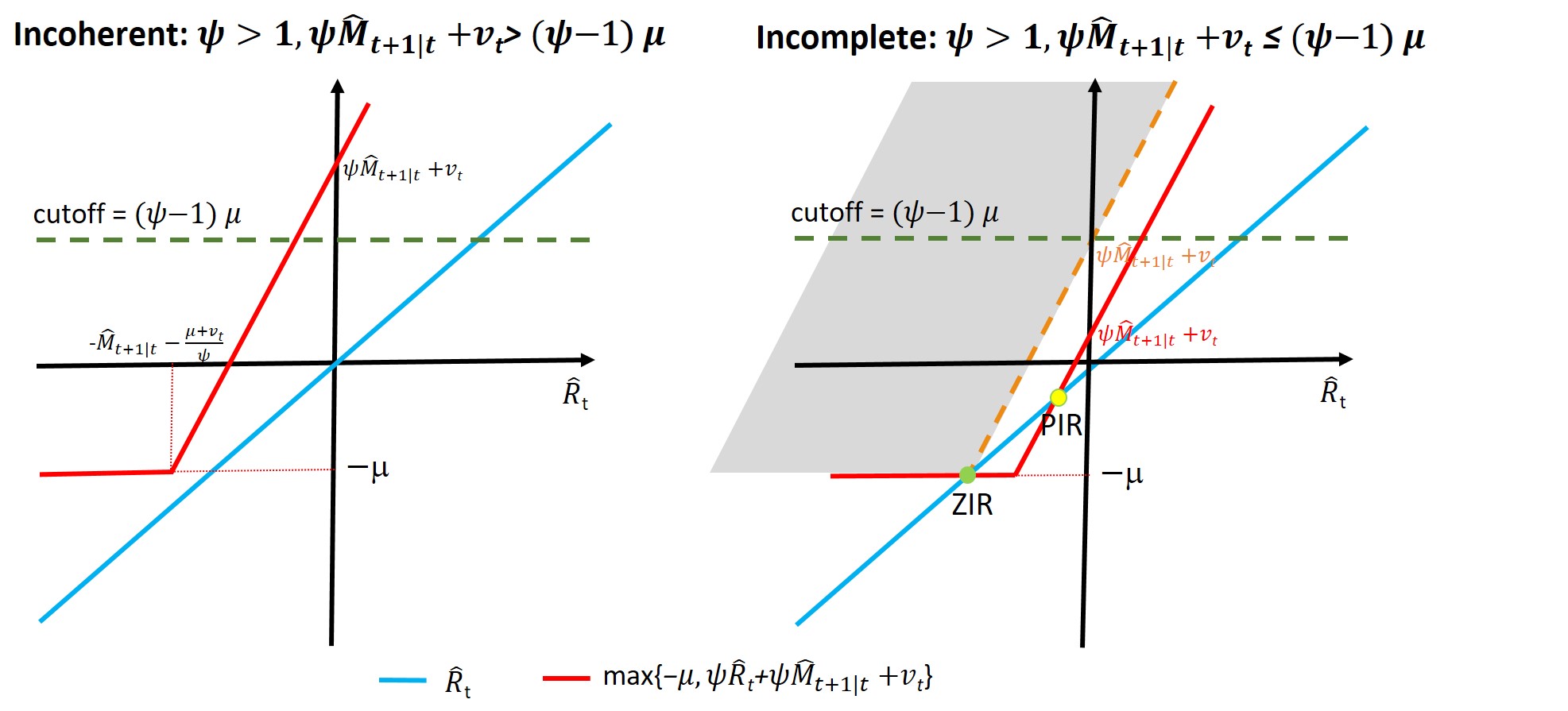}
\caption{Illustration of the restriction on the support of
$\hat{M}_{t+1|t},\nu_t$ in the model given by the intersection between the LHS of (\ref{eq: fig5}), blue line, and the RHS of (\ref{eq: fig5}), red line.}
\label{fig: incompleteness}
\end{figure}

Suppose further that $\hat{M}_t$ follows the AR(1) process $\hat{M}_{t}=\rho\hat{M}_{t-1}+\sigma \epsilon_{t}$ with $E_{t-1}\epsilon_t = 0$, which is the continuous counterpart to the Markov Chain representation we used previously.
The support restrictions can then be equivalently rewritten as 
\begin{equation}
\nu_{t}\leq-\psi\rho\sigma\epsilon_{t}-\psi\rho^2\hat{M}_{t-1}-\left(
1-\psi\right)  \log\left(  r\pi_{\ast}\right)  ,\text{ when }\psi>1.
\label{eq: sup restr}%
\end{equation}
Condition (\ref{eq: sup restr}) has important implications that have been overlooked in the literature on estimation of DSGE models with a ZLB constraint: the shocks $\epsilon_t$ and $\nu_t$ cannot be independently distributed over time, nor can they be independent of each other.
%\todo{\footnotesize{check you're happy with this}}

First, suppose $\nu_t=0$, so the only shock driving the model is $\epsilon_t$.  Condition (\ref{eq: sup restr}) says that $\epsilon_t$ cannot be independently and identically distributed over time, since the support of its distribution depends on past $\hat{M}_{t}$, and hence, past $\epsilon_{t}$. The presence of state variables ($\hat{M}_{t}$) will generally cause support restrictions to depend on the past values of the state variables.
% \footnote{While we cannot solve analytically this problem in the presence of endogenous state variables, again, we do not see how the fact that state variables are endogenous could change this fundamental property of the model, except perhaps under very specific modelling assumptions.}

Second, condition (\ref{eq: sup restr}) states that 
%. In the presence of the monetary policy shock, (\ref{eq: sup restr}) 
the monetary policy shock cannot be independent of the real shock since the support of their distribution cannot be rectangular. Specifically, the monetary policy shock cannot be too big relative to current and past shocks to the discount factor if we are to rule out incoherency. 

If these shocks are structural shocks in a DSGE model, such \textit{necessary} support restrictions are difficult to justify. Structural shocks are generally assumed to be orthogonal. In our opinion, it is hard to make sense of structural shocks whose supports depend on the value of the other shocks in a time-dependent way, as well as on the past values of the state variables.
We believe this is a substantive problem for any DSGE model with a ZLB constraint, and, possibly, more generally, with any kinked constraint. %\sout{A working model to analyze the functioning of an economy cannot be an incoherent model that does not admit a solution.}

A possible solution to this problem is to interpret condition (\ref{eq: sup restr}) as a constraint on the monetary policy shock $\nu_{t}$. When a very adverse shock hits the economy, monetary policy has to step in to guarantee the existence of an equilibrium, that is, to avoid the collapse of the economy. In a sense, this can represent what we witnessed after the Great Financial Crisis or after the COVID-19 pandemic: central banks engaged in massive operations through unconventional monetary policy measures (beyond the standard interest rate policy) in response to these large negative shocks. Hence, $\nu_{t}$ could represent what is currently missing in the simple Taylor rule or in the optimal policy problem to describe what monetary policy needs to do to guarantee the existence of an equilibrium facing very negative shocks and a ZLB constraint. This positive interpretation of condition (\ref{eq: sup restr}) calls for going beyond these descriptions of monetary policy behavior to  model explicitly monetary policy conduct such that incoherency disappears. An alternative way, presented in Subsection \ref{s: UMP} below, relates to the use of unconventional monetary policy modelled via a shadow rate.  

% Another solution, discussed in Subsection \ref{s: UMP} below, relates to the use of unconventional monetary policy modelled via a shadow rate.  

\subsection{The Taylor coefficient and the coherency and completeness conditions}\label{s: cc conditions k}

In the examples above, we saw that active Taylor rules ($\psi>1$) lead to incoherency, so restrictions on the support of the shocks are required for equilibria to exist. More generally, we can use the \nameref{th: GLM} Theorem to find the range of parameters of the models that guarantee coherency without support restrictions. In this subsection, we investigate this question in piecewise linear models with discrete shocks that follow a generic $k$-state Markov Chain.

The main result of this subsection is that the Taylor rule needs to be passive for the coherency and completeness condition in the \nameref{th: GLM} Theorem to be satisfied in the NK model. More generally, there is an upper bound, $\bar{\psi}_k$, on the Taylor rule coefficient $\psi$, which depends on parameters and on the number of states $k$, and it is always less than one. 

We start with an analytical result for the special case with two states $k=2$.
\begin{proposition} \label{prop: NK cutoff}
Consider the NK model given by (\ref{eq: NK}) with $\psi_{x}=u_{t}=\nu
_t=0$ and suppose $\epsilon_{t}$ follows a Markov Chain with two states $\epsilon^1,\epsilon^2$ and transition probabilities $p=\Pr(\epsilon_{t+1}=\epsilon^1\allowbreak|\epsilon_{t}=\epsilon^1)$ and $q=\Pr(\epsilon_{t+1}=\epsilon^2\allowbreak|\epsilon_{t}=\epsilon^2)$ and define 
\begin{equation}\label{eq: NK cutoff}
\psi_{p,q,\beta,\sigma\lambda}:=p+q-1-\frac{\left(  2-p-q\right)  \left(  1-p\beta - q\beta +\beta\right)  }{\sigma\lambda}.
\end{equation}
The coherency condition in the \nameref{th: GLM} Theorem holds if and only if
\begin{subequations}\label{eq: CC-NK-2s}
\begin{align}\label{eq: CC-NK-2s E}
\text{either}\quad & \psi_{p,q,\beta,\sigma\lambda}<0 \text{ and } \psi_{p,q,\beta,\sigma\lambda}<\psi<1, \\ 
\text{or}\qquad &  \psi_{p,q,\beta,\sigma\lambda}>0 \text{ and } \psi<\psi_{p,q,\beta,\sigma\lambda}\leq 1. \label{eq: CC-NK-2s B}
\end{align}
\end{subequations}
\end{proposition}

Again, the coherency condition depends on the slopes of the AS (\ref{eq: NK NKPC}) and AD (\ref{eq: NK EE}) curves. However, in all cases, it rules out $\psi>1$, generalizing Proposition \ref{prop: NK-TR CC}.
If one of the states is absorbing, $q=1$, then $\psi_{p,q,\beta,\sigma\lambda} = p-\frac{\left(  1-p\right) \left(1-p\beta\right)}{\sigma\lambda},$ and the condition in (\ref{eq: CC-NK-2s E}) $\psi_{p,q,\beta,\sigma\lambda}<0$ is equivalent to $\theta>1$, as in (\ref{eq: supp restr NK E}) in Proposition \ref{prop: NK-TR sup res}, implying that the slope of the AS curve is flatter than the one of the AD curve under ZLB in the temporary state.

Another important special case is $p=q=(1+\rho)/2$, where $\rho \in (-1,1)$ is the autocorrelation coefficient of the shock $\epsilon_t$. In that case, we obtain $\psi_{p,q,\beta,\sigma\lambda} = \rho-\frac{\left(  1-\rho\right) \left(1-\rho\beta\right)}{\sigma\lambda}$. This can be thought of as a two-state approximation of a continuous AR(1) process for $\epsilon_t$. We can evaluate the coherency condition numerically for a $k$-state \cite{Rouwenhorst1995} approximation of an AR(1) process with $k>2$. Table \ref{t: calibrations} reports the coherency condition for various calibrations of the model found in \cite{MertensRavn2014},  \cite{eggsingh19jedc} and \cite{Bilbiie2019neofisher}.\footnote{Note that in some of the calibrations, the dynamics are driven by a sunspot shock, e.g., the confidence-driven model listed as MR2014 CD. However, the derivation of the coherency condition remains exactly the same when the transition matrix $K$ corresponds to a sunspot shock instead of the fundamental shock $\epsilon_t$.}
%Note that $p$ is replaced by $\rho$ in the table in anticipation of the more general results below.
For example, when $\rho\sigma\lambda<\left(  1-\rho\right)  \left(
1-\rho\beta\right)$, we verified numerically to 6 decimal digit precision that the coherency condition remains $\psi < 1$, that is, (\ref{eq: CC-NK-2s E}) holds for all $k\leq30$. In the opposite case, the coherency condition is (\ref{eq: CC-NK-2s B}), and $\bar{\psi}_k$ 
%remains (\ref{eq: CC-NK-2s B}) when $\rho$ and $\sigma\lambda$ are small, but 
can get considerably smaller for large values of $\rho$ and $\sigma\lambda$. For any given values of $\rho$ and $\sigma\lambda$, $\bar{\psi}_k$ seems to converge to some value that is bounded away from zero (see the last column of Table \ref{t: calibrations}). As discussed previously, \nameref{ex: ACS} is a special case that obtains when $\sigma$ is large. In that case, $\bar{\psi}_k$ tends to zero with $k$, which suggests that the coherency condition is not satisfied for any $\psi>0$ in the ACS model with a continuously distributed AR(1) shock.

\begin{table}[h!]

%\hspace{-0.75in}
\centering
\caption{Coherency condition $\psi<\bar{\psi}$ for different calibrations of the NK model\label{t: calibrations}} 
%\footnotesize
\begin{tabular}{|l|c|c|c|c|c|c|c|}
\hline
Paper & $\beta $ & $\sigma $ & $\lambda $ & $\mu $ & $\rho$
& $\bar{\psi}$  \\ \hline 
MR2014 FD & 0.99 & 1 & 0.4479 & 0.01 & 0.4 & 1 \\ 
MR2014 CD & 0.99 & 1 & 0.4479 & 0.01 & 0.7 & 0.494 \\ \hline
\cite{Bilbiie2019neofisher} & 0.99 & 1 & 0.02 & 0.01 & 0.8 & 1  \\ 
 & 0.99 & 1 & 0.2 & 0.01 & 0.8 & 0.592 \\ 
\hline
ES2019 GD & 0.9969 & 0.6868 & 0.0091 & 0.0031 & 0.9035 & 1 \\ 
ES2019 GR & 0.997 & 0.6202 & 0.0079 & 0.003 & 0.86 & 1  \\ \hline
%\cite{eggsingh19jedc} GD & 0.9969 & 1.4561 & 0.0091 & 0.0031 & $0.901$ & 0.140 \\ \hline
%\cite{eggsingh19jedc} GR & 0.997 & 1.6125 & 0.0079 & 0.003 & $0.86$ & 1  \\ \hline
\end{tabular}
\fnote{Notes: MR2014: \cite{MertensRavn2014}, FD: Fundamental-driven, CD: Confidence-driven; ES2019: \cite{eggsingh19jedc}, GD: Great Depression, GR: Great Recession. These papers assume an absorbing state and $\rho$ corresponds to the persistence probability of the transitory state.}
\end{table}

\subsection{Coherency with unconventional monetary policy}\label{s: UMP}

In this subsection we show that UMP can relax the restrictions for coherency
in the NK model. An UMP channel can be added to the model in \nameref{ex: NK}
using a `shadow rate' $\hat{R}_{t}^{\ast }$ that represents the desired UMP
stance when it is below the ZLB. 

Consider a model of bond market
segmentation \citep{ChenCurdiaFerrero2012}, where a fraction of households can only invest in long-term bonds. In such a model, the amount of long-term assets held by the private sector affects the term premium and provides an UMP channel via long-term asset purchases by the central bank. If we assume that asset purchases (quantitative easing) follow a similar policy rule to the Taylor rule, i.e., react to inflation deviation from target, then the IS curve (\ref{eq: NK EE}) can be written as (see Appendix 
\ref{s: QE})
\begin{eqnarray}
\hat{x}_{t} &=&\hat{x}_{t+1|t}-\sigma \left((1-\xi) \hat{R}_{t}+\xi \hat{R}_{t}^{\ast }-\hat{\pi}_{t+1|t}\right) +\epsilon _{t},
\label{eq: NK EE UMP} \\
\hat{R}_{t} &=&\max \left\{ -\mu ,\hat{R}_{t}^{\ast }\right\} ,  \quad \hat{R}_{t}^*=\psi \hat{\pi}_{t}+\psi _{x}\hat{x}_{t}+\nu_{t}, \label{eq: NK TR UMP}
\end{eqnarray}%
where $\xi$ is a function of the fraction of households constrained to invest
in long-term bonds, the elasticity of the term premium with respect to the stock of long term bonds and the intensity of UMP. 
The standard NK model (\ref{eq: NK}) arises as a special case with $\xi
=0$.

The conditions for coherency can be derived analytically in the case of a
single AD shock with a two-state support, analogously to Proposition \ref%
{prop: NK cutoff}.

\begin{proposition}
\label{prop: NK cutoff UMP} Consider the NK model given by (\ref{eq: NK NKPC}%
), (\ref{eq: NK EE UMP}) and (\ref{eq: NK TR UMP}) with $\psi _{x}=u_{t}=\nu
_{t}=0$ and suppose $\epsilon _{t}$ follows a Markov Chain with one
absorbing state and one transitory state that persists with probability $p$. Then, the coherency condition
in the \nameref{th: GLM} Theorem holds if and only if 
\begin{subequations}
\label{eq: CC-NK-2s-UMP}
\begin{align}
\qquad & \psi >\max \left( 1,\frac{1}{\xi }\right) ,
\label{eq: CC NK UMP high} \\
\text{or}\quad \ & \max \left( \psi _{p,1,\beta ,\sigma \lambda },\frac{\psi
_{p,1,\beta ,\sigma \lambda }}{\xi }\right) <\psi <\min \left( 1,\frac{1}{%
\xi }\right) ,  \label{eq: CC NK UMP mid} \\
\text{or}\quad \psi & <\min \left( \psi _{p,1,\beta ,\sigma \lambda },\frac{%
\psi _{p,1,\beta ,\sigma \lambda }}{\xi }\right) ,  \label{eq: CC NK UMP low}
\end{align}%
\end{subequations}
where $\psi _{p,1,\beta ,\sigma \lambda }\leq 1$ is defined in (\ref{eq: NK
cutoff}).
\end{proposition}

As $\xi $ goes to zero, the model reduces to the standard NK\ model (\ref%
{eq: NK}), and the coherency condition (\ref{eq: CC-NK-2s-UMP}) reduces to (%
\ref{eq: CC-NK-2s}). We already established that in that case there are no values of $\psi
>1$ that lead to coherency, i.e., an active Taylor rule violates the
coherency condition. However, when UMP\ is present and effective, i.e., $\xi>0$, condition (\ref{eq: CC NK UMP high}) shows that an active Taylor
rule can still lead to coherency, i.e., a MSV solution exists without
support restrictions. For example, the value $\psi =1.5$ for the Taylor rule
coefficient used in typical calibrations leads to coherency if $\xi >2/3$. This is consistent
with the estimation results reported in \cite{IkedaLiMavroeidisZanetti2020},
who find the identified set of $\xi$ to be $\left[ 0.74,0.76\right]$ using postwar U.S. data.

\subsection{Endogenous state variables}\label{s: endog}

Up to this point, we analysed models without endogenous dynamics. In this subsection, we describe the challenges posed by endogenous dynamics and discuss various solutions. 

We can add endogenous dynamics to the canonical model (\ref{eq: canon}) as follows:
\begin{equation}%
\begin{tabular}
[c]{l}%
$A_{s_{t}}Y_{t}+B_{s_{t}}Y_{t+1|t}+C_{s_{t}}X_{t}+D_{s_{t}}X_{t+1|t}+H_{s_{t}%
}Y_{t-1}=0$\\
$s_{t}=1_{\left\{  a^{\prime}Y_{t}+b^{\prime}Y_{t+1|t}+c^{\prime}%
X_{t}+d^{\prime}X_{t+1|t}+h^{\prime}Y_{t-1}>0\right\}  },$%
\end{tabular}
\label{eq: canon endog}%
\end{equation}
where $Y_{t}$ is a $n\times1$ vector of (nonpredetermined) endogenous variables, $X_{t}$ is a
$n_{x}\times1$ vector of exogenous state variables, $H_{s}$ are $n\times n$ coefficient matrices, $h$ is an $n\times1$ coefficient vector, and the
remaining coefficients were already defined above in (\ref{eq: canon}).

\paragraph{Example NK-ITR}\label{ex: NK inert}(NK model with Inertial Taylor rule)
A generalization of the basic three-equation NK model in \nameref{ex: NK} is obtained by replacing (\ref{eq: NK TR}) with
$\hat{R}_{t}=\max\{  -\mu,\allowbreak\phi\hat{R}_{t-1}+\psi\hat{\pi}_{t}%
+\psi_{x}\hat{x}_{t}+\nu_{t}\}  $. It can be put in the canonical form
(\ref{eq: canon endog}) with $Y_{t}=\left(  \hat{\pi}_{t},\hat{x}_{t},\hat
{R}_{t}\right)  ^{\prime},$ $X_{t}=\left(  u_{t},\epsilon_{t},\nu
_{t},1\right)  ^{\prime}$, and coefficient matrices given in Appendix 
\ref{app: s: coeffNKITR}. \hfill \openbox\medskip

As before, we study the existence of MSV solutions, which are of the form
$Y_{t}=f\left(  Y_{t-1},X_{t}\right)  $. We also assume as before that $X_{t}$
follows a $k$-state Markov chain with transition kernel $K$ and support
$\mathbf{X\in\Re}^{n_{x}\times k}$, so that the $i$th column of $\mathbf{X,}$ i.e., $\mathbf{X}e_{i},$ gives the value of $X_{t}$ in state $i=1,\ldots,k$.
In models without endogenous state variables (\ref{eq: canon}), we can represent the
MSV solution $Y_{t}=f\left(  X_{t}\right)  $ exactly by a constant $n\times k$ matrix
$\mathbf{Y}$, since $Y_{t}$ has exactly $k$ points of support
corresponding to the $k$ states of $X_{t}$ that are constant over time. 
% Because the support of $Y_t$ was time-invariant, we could express $E(Y_{t+1}|X_t=\mathbf{X}e_i)$ as $\mathbf{Y}K'e_i$, and represent the model as a piecewise linear system of equations in $\mathbf{Y}$, given by (\ref{eq: canon i}), and then apply the \nameref{th: GLM} Theorem.
However, with endogenous states the support of $Y_{t}$ will vary endogenously
over time along a MSV solution $Y_{t}=f\left(  Y_{t-1},X_{t}\right)$ through the evolution of $Y_{t-1}$, and thus, it cannot be characterized by a constant matrix $\mathbf{Y}$.
That is, the matrix of support points of $Y_t$ -- which we previously defined by the time-invariant matrix $\mathbf{Y}$ because the support of $X_t$ is time-invariant -- must now be a function of $Y_{t-1}$, too. Hence while without endogenous state variables along a MSV solution we have $E\left(  Y_{t+1}|Y_{t}=\mathbf{Y}e_{i}\right) =E\left(  Y_{t+1}|X_{t}=\mathbf{X}e_{i}\right)  =\mathbf{Y}K^{\prime}e_{i}$, when there are endogenous state variables we have $E(Y_{t+1}|Y_{t}\allowbreak=\mathbf{Y}_{t} e_{i},X_{t}\allowbreak=\mathbf{X}e_{i})\allowbreak=\mathbf{Y}_{t+1}^{i}K^{\prime}e_{i},$ where $\mathbf{Y}_{t+1}^{i}$ gives the support of $Y_{t+1}$ when $Y_{t}$ is in the $i$th state. The problem is that the support of $Y_{t}$ is rising exponentially for any given initial condition $Y_{0},$ so the MSV solution cannot be represented by any finite-dimensional system of piecewise linear equations. This makes the analysis of Subsection \ref{s: piecewise linear} generally inapplicable.

To make progress, we will consider solving the model backwards from some
terminal condition in a way that nests the case of no endogenous dynamics that
we studied earlier. We will assume, for simplicity of notation, that the
endogenous state variable is a scalar, i.e., $H_{s_{t}}Y_{t-1}=h_{s_{t}%
}y_{t-1}$ in (\ref{eq: canon endog}), where $h_{s_{t}}$ is $n\times1$ and
$y_{t}:=g^{\prime}Y_{t}$ is a linear combination of $Y_{t}$, for some known
$n\times1$ vector $g$, as in \nameref{ex: NK inert}, where $g=\left(
0,0,1\right)  ^{\prime}$.\footnote{Having multiple endogenous state variables
will increase the number of coefficients we need to solve, but will not
increase the complexity of the problem: the solution $\mathbf{Y}%
_{t}=\mathbf{G}y_{t-1}+\mathbf{Z}$ will need to be replaced by $\mathbf{Y}%
_{t}=\sum_{l=1}^{n}\mathbf{G}_{l}Y_{l,t-1}+\mathbf{Z}$.} Next, suppose there
is a date $T$ such that for all $t\geq T$, the MSV solution $f\left(
y_{t-1},X_{t}\right) $ can be represented in the form
$\mathbf{Y}_{t}=\mathbf{G}y_{t-1}+\mathbf{Z}$, with $n\times k$ matrices
$\mathbf{G}$ and $\mathbf{Z}$. This is a matrix representation of a general possibly nonlinear function $f(y_{t-1},X_t)$. The matrix $\mathbf{Z}$ represents the part of $Y_t$ that depends only on the exogenous variables $X_t$. When there are no endogenous states, we have $\mathbf{G}=0$, so $\mathbf{Y}_{t}=\mathbf{Z}$, which we denoted by $\mathbf{Y}$ in (\ref{eq: canon i}) above. Each column of $\mathbf{G}$ gives the coefficients on $y_{t-1}$ in the MSV solution that correspond to each different state of $X_t$. For example, if $k=2$, then the endogenous dynamics in the low state $i=1$, say, can be different from the high state $i=2$. If $i=1$ were a ZIR and $i=2$ were a PIR, the endogenous dynamics could differ across regimes, as shown analytically in the example below. 
% This allows the endogenous dynamics to vary
% across regimes (the columns of $\mathbf{G}$ need not be the same), and nests
% the case of no endogenous dynamics with $\mathbf{G}=0$. 

In the case of no endogenous dynamics, we analysed the solution of $\mathbf{Y}(=\mathbf{Z}$ here) using the method of undetermined coefficients, see equation (\ref{eq: canon i}) above. The corresponding equation for the model with endogenous dynamics is 
% We will attempt to solve for
% $\mathbf{G}$ and $\mathbf{Z}$ using the method of undetermined coefficients.
% Substituting for $\mathbf{Y}_{t}$ and $\mathbf{Y}_{t+1}^{i}\allowbreak
% =\mathbf{G}\left(  g^{\prime}\mathbf{Y}_{t}e_{i}\right)  +\mathbf{Z,}$  into
% (\ref{eq: canon endog i}) and rearranging yields
\begin{align}\label{eq: equations for G and Z}
0 &  =\left(  A_{s_{t,i}}\mathbf{G}e_{i}+h_{s_{t,i}}+B_{s_{t,i}}%
\mathbf{G}K^{\prime}e_{i}g^{\prime}\mathbf{G}e_{i}\right)  y_{t-1}\\ 
&  +\left(  A_{s_{t,i}}\mathbf{Z}+B_{s_{t,i}}\mathbf{G}K^{\prime}%
e_{i}g^{\prime}\mathbf{Z}+B_{s_{t,i}}\mathbf{Z}K^{\prime}+C_{s_{t,i}%
}\mathbf{X}+D_{s_{t,i}}\mathbf{X}K^{\prime}\right)  e_{i}, \nonumber
\end{align}
for all $i=1,\ldots,k$, see \ref{app: s: bruteforcealg}. 
%Note that if you set $h_{s_{t,i}}=0$, $G=0$ and $\mathbf{Z}=\mathbf{Y}$, you arrive at (\ref{eq: canon i}). 
Note that (\ref{eq: equations for G and Z}) exactly nests (\ref{eq: canon i}) when $h_{s_{t,i}}=0$, $G=0$ and $\mathbf{Z}=\mathbf{Y}$. Given a particular regime configuration $J\subseteq
\left\{  1,\ldots,k\right\}$, which determines in which of the states the constraint is slack, $s_{t,i}=1$, (\ref{eq: equations for G and Z}) gives a system of $2nk$ polynomial equations in the $2nk$ unknowns $\mathbf{G}$ and $\mathbf{Z}$ by equating the coefficients on $y_{t-1}$ and the constant terms to zero, respectively. Unfortunately, because this system of equations is not piecewise linear in $\mathbf{G}$ and $\mathbf{Z}$, we cannot use the \nameref{th: GLM} Theorem to check coherency.
Therefore, one would need to resort to a brute force method of going through
all possible $2^{k}$ regime configurations $J$ and checking if there are any
solutions %$\mathbf{Y}_{t}=\mathbf{G}_{J}y_{t-1}+\mathbf{Z}_{J}$ that satisfies
that satisfy the inequality constraints. Since $y_{t-1}$ is endogenous, one would need
to solve the system backwards subject to some initial condition $y_{0}$, and
then check coherency for all possible values of $y_{0}$. An algorithm for
doing this is given in Appendix 
\ref{app: s: bruteforcealg}
% with an example of a NK model with forward guidance
.\footnote{It is worth noting that solving the model
backwards from $T$ to $1$ requires (up to)\ $2^{k\left(  T-1\right)  }$
calculations, in order to consider all possible regime paths from $1$ to
$T-1$. This is an NP hard problem even for fixed $k$. One could drastically
reduce number of calculations by limiting the possible regime transitions, e.g., as
in \cite{eggertsson2021toolkit}, at the cost of making the conditions for
coherency even stricter.}

\paragraph{\nameref{ex: NK inert} continued}
Again, it is possible to obtain some analytical results in the special case of the model in \nameref{ex: NK inert} if we assume, as in Proposition \ref{prop: NK-TR sup res}, that $\psi_{x}=u_t=\nu_t=0$ and that $\epsilon_{t}=-\sigma \hat{M}_{t+1|t}$, where $M_t$ satisfies Assumption \ref{ass: M nonlinear absorbing}. 
% that is, it follows a two-state Markov chain process such that $p<1$ (transitory state) and $q=1$ (absorbing state), with support of $\hat{M}_{t}$ equal to $-r^{L}$ and $0$, respectively.
Here, we report results on the existence of a MSV solution such that the economy is in a ZIR in the temporary state where $\hat{M}_{t} =-r^{L}$ and then converges to a PIR in the absorbing state where $\hat{M}_{t} =0$. In other words, we report the support restrictions needed for a ZIR to exist in the temporary state, given that the agents expect to move to the stable manifold of the PIR system as soon as the shock vanishes. Appendix 
\ref{app: s: NKITR_analytics} shows that such a solution exists if and only if
\begin{subequations}
\label{eq: supp restr NK inert}
\begin{align}
\text{either}\quad &  \theta>1\text{ and }r^{-1}\leq\pi_{\ast} \text{ and }%
-r^{L}\geq-\bar{r}^{L}
,\label{eq: supp restr NK inert E}\\
\text{or}\qquad &  \theta\leq1,\text{ }r^{-1}\leq\pi_{\ast}\text{ and }
-r^{L}\leq-\bar{r}^{L},\label{eq: supp restr NK inert B}
\end{align}
\end{subequations}
% \begin{equation}
%     \theta\leq1,\text{ }r^{-1}\leq\pi_{\ast}\text{ and } -r^{L}\leq-\bar{r}^{L},\label{eq: supp restr NK inert B} \qquad \text{where}
% \end{equation}
% where $\theta:=\frac{\left(  1-p\right)  \left(1-p\beta\right)}{p\sigma\lambda}$ and
where\vspace{-36pt}
\begin{equation}
-\bar{r}^{L} =\mu\left(  \frac{\psi
-p}{\psi p}+\frac{\theta}{\psi}+\frac{\phi}{\psi}\left(  1-\theta\right)  -(1-p)\frac
{\lambda\gamma_{x}+\gamma_{\pi}\left[  \beta(1-p)+\lambda\sigma\right]
}{\lambda\sigma p}\right), \label{eq: rl NK-ITR}
\end{equation}
$\theta:=\frac{\left(  1-p\right)  \left(1-p\beta\right)}{p\sigma\lambda}$ and $\gamma_{x},\gamma_{\pi}$ are functions of the model's parameters. $\gamma_{x},\gamma_{\pi}$ define the slope of the stable manifold of the PIR system such that in the absorbing state $\hat{\pi}_t$ and $\hat{x}_t$  will travel to the PIR steady state along the paths $\hat{\pi}_t = \gamma_{\pi} \hat{R}_{t-1}$ and $\hat{x}_t=\gamma_{x}\hat{R}_{t-1}$. 
% As noted, this result is more limited in scope than the one in  Proposition \ref{prop: NK-TR sup res} because here we are not checking for all the possible MSV solutions, but just for the ZIR-PIR one. However, t
This result nests the corresponding analysis in Proposition \ref{prop: NK-TR sup res}, because if $\phi=0$ then $\gamma_{\pi}=\gamma_{x}=0$ and (\ref{eq: supp restr NK inert B}) collapses to (\ref{eq: supp restr NK B}), where  $-r^{L}\leq\log(r\pi_{\ast})\allowbreak\left(  \frac{\psi-p}{\psi p}+\frac{\theta
}{\psi}\right)$. (\ref{eq: supp restr NK inert}) shows that, for a ZIR-PIR to exist, the negative shock should be large enough in absolute value when $\theta>1$, while it should be small enough in absolute value when $\theta \leq 1.$ Numerical results (not reported) show that, if that condition is not satisfied, then there is a PIR-PIR solution if $\theta>1$, and there is no solution if $\theta\leq 1$, as implied by Proposition \ref{prop: NK-TR sup res} for the case of no inertia.

The effects of inertia on the support restrictions needed for the existence of an equilibrium when $\theta\leq 1$ can be evaluated numerically, since we do not have analytic expressions for $\gamma_\pi$ and $\gamma_x$. Figure \ref{fig: bound} shows the lower bound $\bar{r}^L$ as a function of $\phi$ for the two calibrations given in Table \ref{t: calibrations} where $\theta\leq 1$.
Given that the shock needs to be above the lower bound (i.e., smaller in absolute value) for the equilibrium to exists, the graph reveals that inertia relax the support restrictions.
\begin{figure}[h!]
\centering
\includegraphics[width =6in]{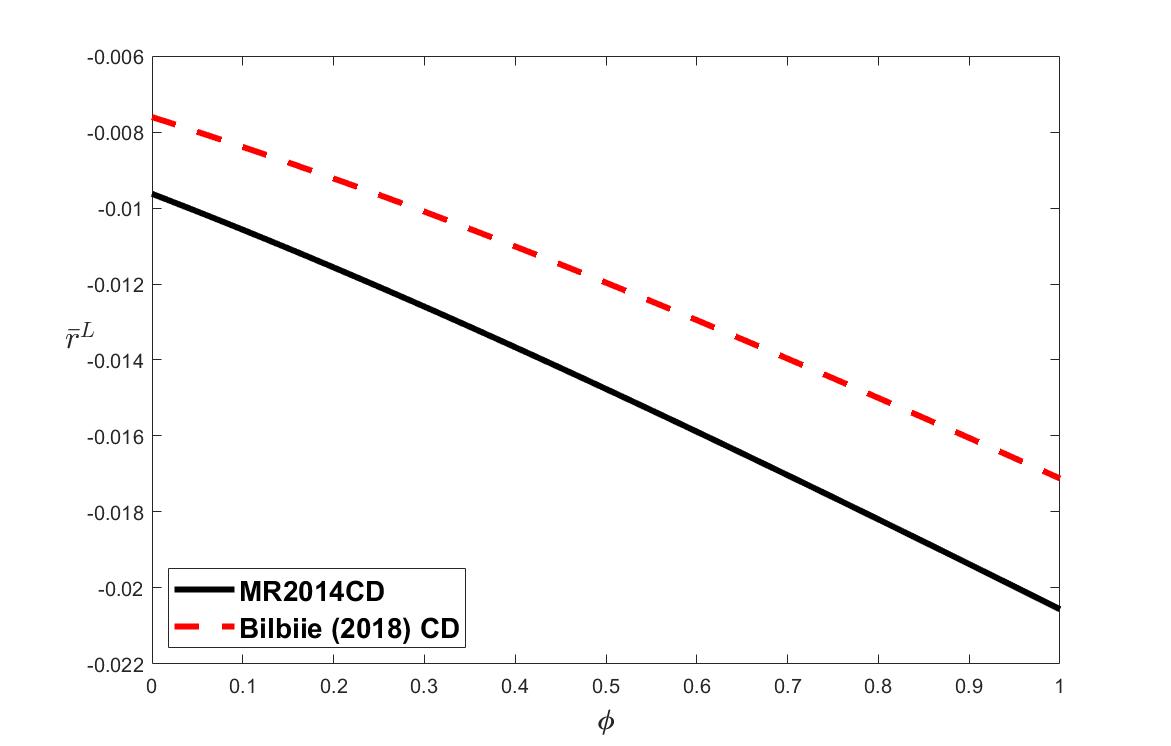}
%[
%natheight=7.077600in,
%natwidth=9.676400in,
%height=4.7997in,
%width=6.5518in
%]
\caption{$\bar{r}^{L}$ as a function of $\phi$, using parameters from \cite{MertensRavn2014} and \cite{Bilbiie2019neofisher} calibrations shown in Table \ref{t: calibrations} where $\theta<1$.}
\label{fig: bound}
\end{figure}

% Compared with the non-inertial case in Proposition \ref{prop: NK-TR sup res},
% where $-r^{L}\leq\log(r\pi_{\ast})\allowbreak\left(  \frac{\psi-p}{\psi p}+\frac{\theta
% }{\psi}\right)  $ when $\theta\leq1$, we see that the impact of inertia on the
% lower bound $\bar{r}^{L}$ depends on the sign of $\frac{\phi}{\psi}\left(
% 1-\theta\right)  -\left(  1-p\right)  \frac{\gamma_{x}+\lambda\sigma
% +\beta(1-p)\gamma_{\pi}}{\sigma p}$.

\subsubsection{Quasi differencing}\label{s: quasi diff}
In a very special case, we can analyse the coherency of the model using the \nameref{th: GLM} Theorem.
%as in Subsection \ref{s: piecewise linear}. 

\begin{assumption}\label{ass: ST} Assume
the first $n_{1}$ elements $Y_{1t}$ of $Y_{t}=\left(  Y_{1t}^{\prime},Y_{2t}^{\prime
}\right)  ^{\prime}$ in model (\ref{eq: canon}) are predetermined,
$B_{0},B_{1}$ are invertible, and there exists a $n\times n$ invertible matrix
$Q$ such that $Q^{-1}B_{s}^{-1}A_{s}Q=\Lambda_{s}$ is upper triangular for
both $s=0,1.$ Let $Q_{1}$ denote the first $n_{1}$ elements of $Q,$ and assume also
that $a^{\prime}Q_{1}=b^{\prime}Q_{1}=0.$ 
\end{assumption}

The first part of Assumption \ref{ass: ST} is satisfied for commuting pairs of
matrices, but commutation is not necessary.\footnote{For example, a pair of
non-commuting triangular matrices is trivially simultaneously triangularized.
One can check Assumption \ref{ass: ST} using the algorithm of \cite{Dubi2009}.} 
This assumption is clearly restrictive, but not empty, see \nameref{ex: ACS inert} below.

If Assumption \ref{ass: ST} holds, then we can remove the predetermined variables $Y_{1t}$ from (\ref{eq: canon}) by premultiplying the equation by
some $n_{2}\times n$ matrix $(-\Gamma,I_{n_2})$, where $\Gamma$ is $n_2\times n_1$
% $\left(  Q^{-1}\right)  _{22}^{-1}\left(Q^{-1}\right)  _{2\cdot},$ where $\left(  Q^{-1}\right)  _{22}$ is the bottom right $n_{2}\times n_{2}$\ submatrix of $Q^{-1}$ and $\left(  Q^{-1}\right)_{2\cdot}$ consists of the bottom $n_{2}$ rows of $Q^{-1}$, 
(see Appendix 
\ref{app: s: simutri}). Intuitively, we transform the original system to one that does
not involve any predetermined variables by taking a `quasi difference' of
$Y_{2t}$ from the predetermined variables $Y_{1t},$ say $\widetilde{Y}_{t} := Y_{2t}-\Gamma Y_{1t}$.
% where $\Gamma=-\left(  Q^{-1}\right)  _{22}^{-1}\allowbreak\left(  Q^{-1}\right)  _{21}$.
This is typically possible in linear models, but it is not, in general,
possible in piecewise linear models because the `quasi difference'
coefficient $\Gamma$ needs to be the same across all regimes. Assumption
\ref{ass: ST} ensures this is possible. Since the model in $\widetilde{Y}_{t}$ does not have any predetermined variables and is still of the form (\ref{eq: canon}),
we can analyze its coherency using the \nameref{th: GLM} Theorem as in
Subsection \ref{s: piecewise linear}.

\paragraph{Example ACS-STR}

\label{ex: ACS inert}Consider a generalization of \nameref{ex: ACS} where
the Taylor rule $\hat{R}_t=\max\left\{-\mu, \psi\hat{\pi}_t\right\}$ is replaced with $\hat{R}_{t}=\max\left\{  -\mu,\phi\hat
{R}_{t-1}+\psi\hat{\pi}_{t}\right\}  $. Appendix
\ref{app: s: ACS-STR} shows that this
model satisfies Assumption \ref{ass: ST}, applies the \nameref{th: GLM} Theorem and analyzes the solutions. When the exogenous shock satisfies
Assumption \ref{ass: M nonlinear absorbing}, the coherency condition is
$\psi<p-\phi$, which generalizes the one found earlier without
inertia: $\psi<p$. When $\psi>1$, the requisite support restriction on $r^{L}$
is $-r^{L}<\mu\frac{\psi  +\phi-p}{p\psi}$. For $\phi>0$, this support restriction is weaker than the one for the noninertial model given in (\ref{eq: support restr lin ACS}). \openbox

\section{The incompleteness problem}\label{s: incompleteness}

%\paragraph{Multiple MSV equilibria.}
This Section explores the multiplicity of the MSV solutions in piecewise linear models of the form (\ref{eq: canon}). The main message is that when the CC condition of the \nameref{th: GLM} Theorem is not satisfied, but the support of the distribution of the shocks is restricted appropriately, there are many more MSV solutions than are typically considered in the literature. This is distinct from the usual issue of indeterminacy in models without occasionally binding constraints. 

As we discussed in Subsection \ref{s: piecewise linear}, when the state variables $X_t$ follow a $k$-state Markov chain, these models can be written as $F\left(  \mathbf{Y}%
\right)  =\kappa\left(  \mathbf{X}\right)$, where $F(\cdot)$ is the piecewise linear function (\ref{eq: F}), and all possible solutions correspond to
$\mathbf{Y=}\mathcal{A}_{J}^{-1}\kappa\left(  \mathbf{X}\right)$ for each
$J\subseteq\left\{  1,\ldots,k\right\}$. Thus, there are up to $2^k$ possible MSV solutions. 

% In passing, note that the cardinality of the MSV solutions is not equal to the cardinality of steady states. Thus, it is important to distinguish multiplicity of steady states and multiplicity of equilibria. In \nameref{ex: ACS}, for instance, when $\psi>1$ and Assumption \ref{ass: M nonlinear absorbing} holds, there are two steady states, but up to four MSV equilibria, while in the case $p+q-1<\psi<1$, there is a unique steady state but up to two MSV equilibria. 

\begin{figure}
[ptb]
\begin{center}
\centering\includegraphics[width =6in]{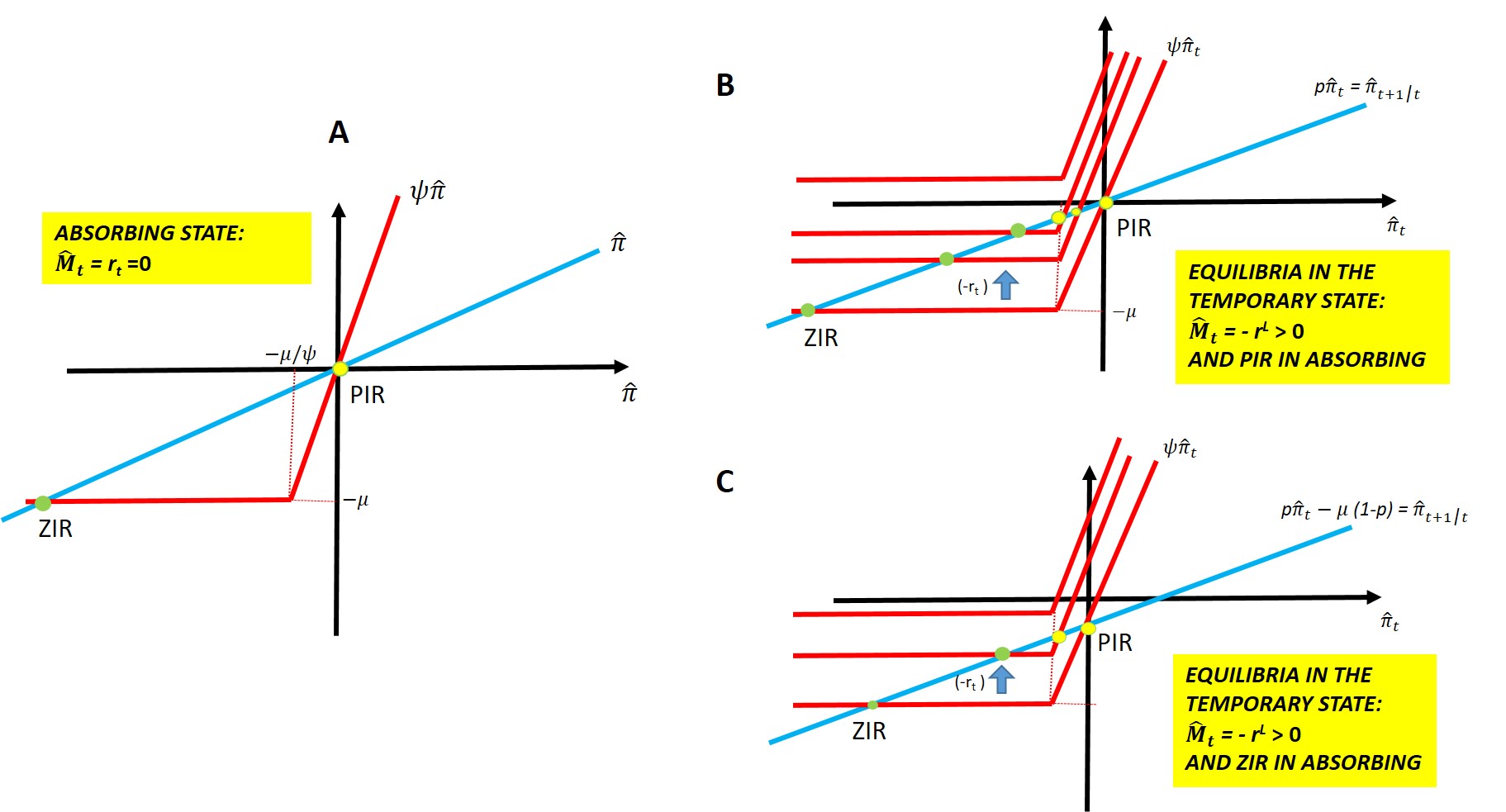}
%[
%natheight=7.077600in,
%natwidth=9.676400in,
%height=4.7997in,
%width=6.5518in
%]
\caption{The possible equilibrium outcomes in \nameref{ex: ACS}}
\label{fig: 2states}
\end{center}
\end{figure}

\paragraph{\nameref{ex: ACS} continued} 
To start with, consider the special case where $M_t$ satisfies Assumption \ref{ass: M nonlinear absorbing}. A corollary of Proposition \ref{prop: NK cutoff} (with $q=1$ and $\sigma=\infty$) shows that the CC condition is satisfied if and only if $\psi<p$. Hence, for $\psi\geq p$ there will be up to 4 MSV solutions if the support of $M_t$ allows it. These are shown in Figure \ref{fig: 2states} and Table \ref{tab: equilibria simple} for $\psi>1$, see Appendix
\ref{app:simple ex} for the derivation. The equilibrium typically used in the literature is the second one in Table  \ref{tab: equilibria simple}: ZIR in the transitory and PIR in the absorbing state. This is a fairly intuitive choice in this simple case, but there is no clear choice in more general scenarios with no absorbing state and $k>2$.
\setlength{\tabcolsep}{0.7em} % for the horizontal padding
%{\renewcommand{\arraystretch}{1.2}  \setlength{\extrarowheight}{4pt}
\begin{table}[h!]
   \centering
    \caption{The four possible equilibria in \nameref{ex: ACS} when $\psi>1$}
    \label{tab: equilibria simple}
    \begin{tabular}{ccc}
  \hline \hline 
        Analytical Solution & & Type of Equilibrium \\ \hline \hline 
         & & \\
       
          \begin{math}
\hat{\pi}_{t}=\left\{
\begin{array}
[c]{ll}
r^{L}\frac{p}{\psi-p} & \text{if }\hat{M}_{t}=-r^{L}\in\left(  0,\mu
\frac{\psi-p}{\psi p}\right) \\
0 & \text{if }\hat{M}_{t}=0
\end{array}
\right. 
         \end{math} & & (PIR, PIR)  \\ & & \\\hline 
         & & \\
         
         \begin{math}
\hat{\pi}_{t}   =\left\{
\begin{array}
[c]{ll}
-r^{L}-\frac{\mu}{p}, & \text{if }\hat{M}_{t}=-r^{L}\in\left(  0,\mu\frac
{\psi-p}{\psi p}\right)  \\
0 & \text{if }\hat{M}_{t}=0
\end{array}
\right.
         \end{math} & & (ZIR, PIR) \\ & & \\\hline 
         & & \\

        \begin{math}
\hat{\pi}_{t}   =\left\{
\begin{array}
[c]{ll}
\frac{pr^{L}-\left(  1-p\right)  \mu}{\psi-p}, & \text{if }\hat{M}_{t}
=-r^{L}\in\left(  0,\mu\frac{\psi-1}{\psi}\right) \\
-\mu & \text{if }\hat{M}_{t}=0
\end{array}
\right.
        \end{math} & & (PIR, ZIR) \\ & & \\\hline 
         & & \\

        \begin{math}
\hat{\pi}_{t}   =\left\{
\begin{array}
[c]{ll}%
-r^{L}-\mu, & \text{if }\hat{M}_{t}=-r^{L}\in\left(  0,\mu\frac{\psi-1}{\psi
}\right) \\
-\mu, & \text{if }\hat{M}_{t}=0.
\end{array}
\right.
        \end{math} & & (ZIR, ZIR)  \\ & & \\\hline 
         & & \\
    \end{tabular}
\end{table}

To demonstrate the problem, we consider the case where $\hat{M}_t$ is described by a $k$-state \cite{Rouwenhorst1995} approximation of an AR(1) process $\hat{M}_t = \rho\hat{M}_{t-1} + \sigma_\varepsilon \varepsilon_t$, with parameter values $\rho=0.9$ and $\sigma_\varepsilon=0.0007$, and we set $\psi = 1.5$ and $\mu=2\log(1.005)$ following the calibration in ACS, so that the CC condition fails. Figure \ref{fig: mult sol k=3} reports the 8 MSV solutions corresponding to $k=3$. We notice that the first solution is at the ZIR for all values of the shock, while the last solution is the opposite, always at PIR. Unsurprisingly, those two solutions are linear in $\hat{M}_{t}.$ The remaining 6 solutions are non-linear and half of them are non-monotonic in $\hat{M}_{t}$. In Appendix
\ref{app: s: kfigures}, we present results for $k>3$, showing that the number of MSV solutions increases with $k$. In all cases, two solutions correspond to ZIR-only and PIR-only equilibria. For any $k$, it is possible to impose restrictions on the support of the distribution of the shocks such that we are always at ZIR or always at PIR. \openbox

\begin{figure}[htb]
\centering\includegraphics{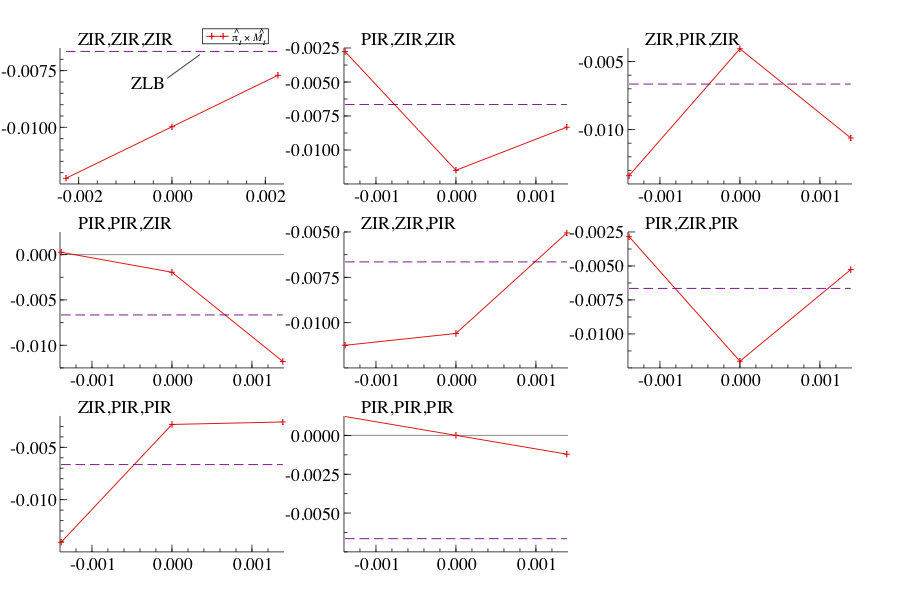}
\caption{MSV solutions of model:  $\hat{\pi}_{t|t+1} = max(-\mu,\psi \hat{\pi}_t)+\hat{M}_{t+1|t}$, when $\mu=0.01$, $\psi=1.5$ and $\hat{M}_t$ follows a 3-state Markov Chain with mean 0, conditional st.~dev.~$\sigma = 0.0007$, and autocorrelation $\rho=0.9$.\label{fig: mult sol k=3}} 
\end{figure}

\paragraph{\nameref{ex: NK} continued}

Consider again the NK model with $u_t=\nu_t=\psi_x=0$ and a Rouwenhorst approximation to an AR(1) process for the AD shock $\epsilon_t$. 
%Figure \ref{fig: MR-NK} plots the decision rules associated with various MSV equilibria of the model for $k=20$ using the parameter values from \cite{MertensRavn2014}, i.e., the left panel uses ``MR2014 CD'' and the right panel  ``MR2014 FD''.
Figure \ref{fig: MR-NK} plots the decision rules for $\hat{\pi}_t$, $\hat{x}_t$ and $\hat{R}_t$, as functions of $\epsilon_t$,
associated with various MSV equilibria of the model for $k=20,$ using the parameter values from \cite{MertensRavn2014}, i.e., the left panel uses ``MR2014 CD'' and the right panel  ``MR2014 FD''.\footnote{The support of the distribution of the shock $\epsilon_t$ has been carefully chosen to avoid incoherency. In this case, because the distribution of the shock is symmetric, the necessary support restrictions can be imposed by manipulating the standard deviation of the shock, denoted by $\sigma_\epsilon$. Larger values yield more dispersion, so when $\sigma_\epsilon$ gets sufficiently large, there are no MSV equilibria.}
%two different calibrations of the model. 
The graphs on the left report the four MSV equilibria arising from the calibration in which  $\rho\sigma\lambda>\left(  1-\rho\right)  \left(
1-\rho\beta\right)$. We notice that two of those equilibria have $\hat{\pi}_t, \hat{x}_t$ respond positively to the AD shock, while the other two equilibria are exactly the opposite. The graphs on the right report the case  $\rho\sigma\lambda<\left(  1-\rho\right)  \left(
1-\rho\beta\right)$, where now only two MSV equilibria have been found, and they both have the property that the policy functions are increasing in the AD shock. Moreover, changing the parameters of the structural model or of the shocks yields a different number of solutions. For example, with a low variance of the shock and $\sigma=4$, the ``MR2014 CD'' case in Table \ref{t: calibrations} delivers 8 solutions.
\hfill\openbox

\begin{figure}%
\centering
    \begin{subfigure}{0.495\textwidth}
        \includegraphics[width=\textwidth]{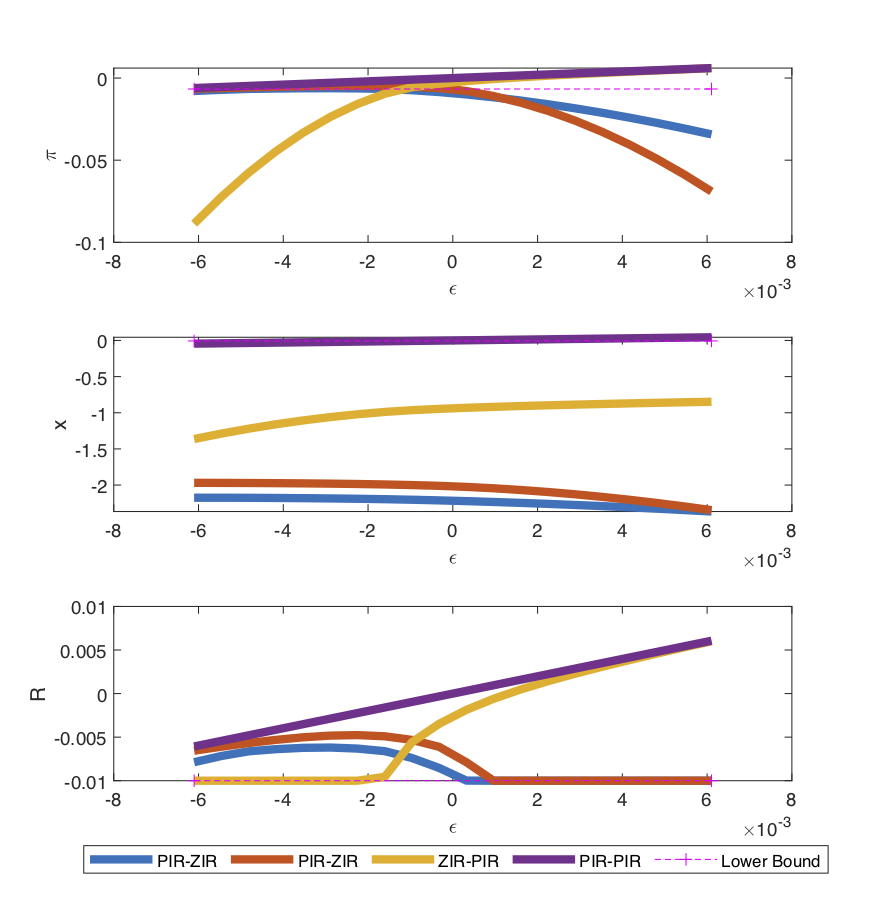}
        \caption{Confidence-driven}
        \label{fig: MR-Ben}
    \end{subfigure}
        \begin{subfigure}{0.495\textwidth}
        \includegraphics[width=\textwidth]{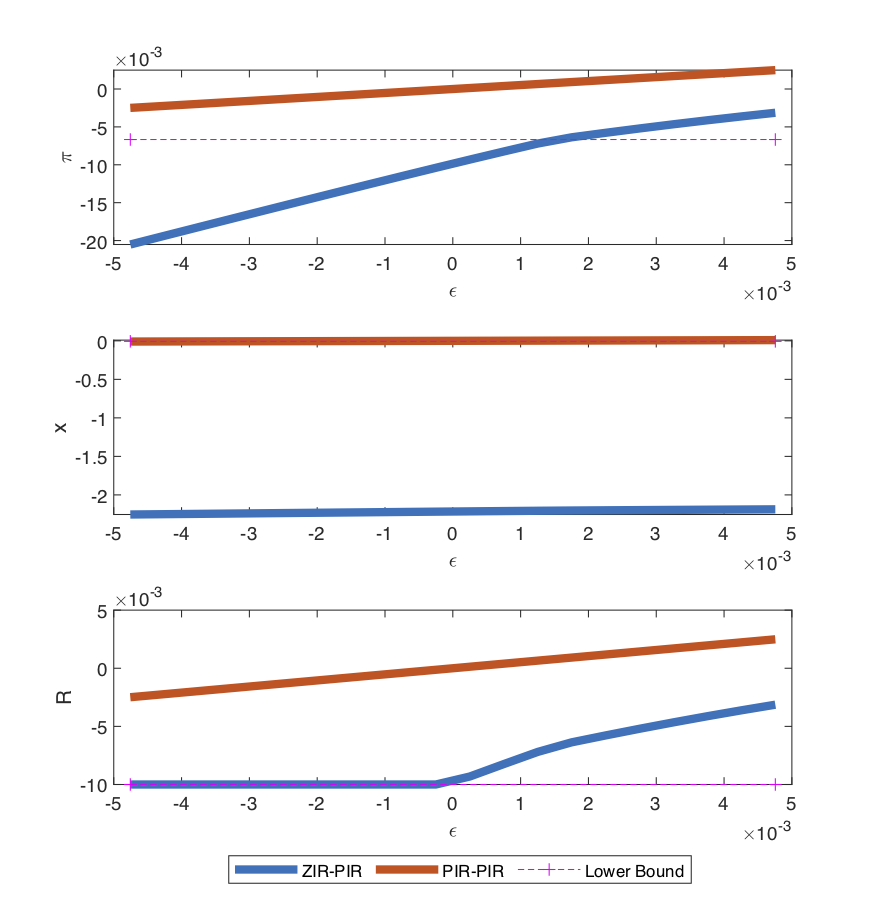}
        \caption{Fundamental-driven}
        \label{fig: MR-Egg}
    \end{subfigure}
\caption{Decision Rules associated with different MSV solutions (equilibria) of the NK model, using parameters from \citetpos{MertensRavn2014} calibration shown in Table \ref{t: calibrations} and $k=20$. The figures on the left correspond to $\rho=0.7$ with $\sigma_\epsilon = 0.0011$, while on the right $\rho=0.4$ and $\sigma_\epsilon=0.0014$.}
\label{fig: MR-NK}
\end{figure}

\section{Conclusions\label{s: conclusions}}
This paper highlights a seemingly overlooked problem in rational expectation models with an occasionally binding constraint. The constraint might make the model incoherent or incomplete. 

We propose a method for checking the coherency and completeness (CC) condition, that is, the existence and uniqueness of equilibria in piecewise linear DGSE models with a ZLB constraint based on \cite{GourierouxLaffontMonfort1980}. When applied to the typical NK model, this method shows that  the CC condition generally violates the Taylor principle. Hence, the case typically analysed in the literature is either incoherent or incomplete. This raises two main issues future research should focus on when solving or estimating these models.
%, regarding the estimation and the solution of these models, that should be the focus of future research. 

First, we have shown that there must be restrictions on the distribution of the shocks to ensure the existence of equilibria. These support restrictions are time-varying and, in the case of multiple shocks, their support is not rectangular, i.e., the shocks cannot be independent of each other. This raises a first question regarding the interpretation of these shocks: in what sense are they structural if they cannot be independent? A second related question regards the estimation of these models: what are the implications of these restrictions for the correct form of the likelihood? 

Second, we have shown there are typically (many) more equilibria than currently reported in the literature. These findings raise questions about the properties of existing numerical solution algorithms, for example, which solutions among the many possible ones do they find and why. 

We have not found a computationally feasible way to analyse coherency and completeness in forward-looking models in which the variables are continuously distributed. This problem is hard because of the infinite dimensionality induced by the rational expectation operator, and the fact that the computations required for discrete approximations are NP hard. 
%It is still possible to make considerable progress using the existing methodology of \cite{GourierouxLaffontMonfort1980} if we are willing to deviate from rational expectations, e.g., by assuming some form of learning or bounded rationality, or, relatedly, if we focus on models with piecewise linear decision rules, as in \cite{Mavroeidis2019} and \cite{AruobaCubaBordaHigaFloresSchorfheideVillalvazo2021}. 
This is an important challenge for future research.

Finally, our results highlight the role of unconventional monetary policy in ensuring coherency. An incoherent model cannot be an operational model of the economy. Hence, the need for support restrictions can be positively interpreted as an implicit need for a different policy reaction to catastrophic shocks to ensure the economy does not collapse. This suggests a direction for amending the basic NK model, by modelling monetary policy in such a way that, conditional on bad shocks hitting the economy and conventional interest rate policy being constrained by the ZLB, the use of unconventional monetary policies offers a route to solving the incoherency problem. This route is not only promising, but, even more importantly, realistic: central banks engaged in massive operations through unconventional monetary policy measures (beyond the standard interest rate policy) in response to the large negative shocks causing the Great Financial Crisis and the COVID-19 pandemic. Future work should also study whether it is possible to design fiscal policy to ensure equilibrium existence \citep[e.g.,][]{NakataSchmidt2020}. A takeaway from of our paper, therefore, is to warn that considering a ZLB constraint on monetary policy requires an explicit modelling of unconventional monetary policies (or some other mechanisms) to avoid incoherency.

%Finally, our results highlight the role of unconventional monetary policy in ensuring coherency. An incoherent model cannot be an operational model of the economy. Hence, the need for support restrictions can be positively interpreted as an implicit need for a different policy reaction to catastrophic shocks to ensure the economy does not collapse. While we do not consider the possibility of using fiscal policy to ensure equilibrium existence \citep[e.g.,][]{NakataSchmidt2020}, we explore the possibility of modelling monetary policy in such a way that, conditional on bad shocks hitting the economy and conventional interest rate policy being constrained by the ZLB, the use of unconventional monetary policies offers a route to solving the incoherency problem. This route is not only promising, but, even more importantly, realistic: central banks engaged in massive operations through unconventional monetary policy measures (beyond the standard interest rate policy) in response to the large negative shocks causing the Great Financial Crisis and the COVID-19 pandemic. A takeaway from of our paper, therefore, is to warn that considering a ZLB constraint on monetary policy requires an explicit modelling of unconventional monetary policies (or some other mechanisms) to avoid incoherency.

% \pagebreak

\bibliographystyle{elsarticle-harv}
\bibliography{UL}

% \begin{center}
%     \textsc{\Large Appendix - ONLINE, NOT FOR PUBLICATION}
% \end{center}

% \vspace{4pt}
%\setcounter{secnumdepth}{0}\newpage
% \singlespacing

\newpage
\setcounter{equation}{0} \renewcommand{\theequation}{A\arabic{equation}}
\appendix
\section{Appendix}

\subsection{Derivation of results in Subsection \ref{s: ACS simple example}}\label{app: s: nonlinear ACS}
\begin{proof}[Proof of Proposition \ref{prop: nonlinear ACS}]

Let $M^{t}=\left(  M_{t},...,M_{0}\right)  $ denote the history of $M_{t}.$
We consider fundamental solutions $f_{\pi_{t}}\left(  M^{t}\right)$. Let $M_{L}%
^{t}=\left(  \frac{e^{-r^{L}}}{r},...,\frac{e^{-r^{L}}}{r}\right)  $ denote a
path along which $M_{t}$ is in the transitory state. It follows that (with
slight abuse of notation)
\begin{equation}
E\left(  \left.  \frac{M_{t+1}}{\pi_{t+1}}\right\vert M_{L}^{t}\right)
=p\frac{\frac{e^{-r^{L}}}{r}}{f_{\pi_{t+1}}\left(  \frac{e^{-r^{L}}}{r},M_{L}%
^{t}\right)  }+\left(  1-p\right)  \frac{r^{-1}}{f_{\pi_{t+1}}\left(  r^{-1}%
,M_{L}^{t}\right)  }. \label{eq: exp trans}%
\end{equation}
Next, if $M_{t+1}=r^{-1}$, we have $
E\left(  \left.  \frac{M_{t+2}}{\pi_{t+2}}\right\vert M_{t+1}=r^{-1},M_{L}%
^{t}\right) \allowbreak =\frac{r^{-1}}{f_{\pi_{t+2}}\left(  r^{-1},r^{-1},M_{L}^{t}\right)
}$, so that after taking logs and re-arranging, (\ref{eq: ACS nonlinear}) becomes
\begin{equation}
f_{\hat{\pi}_{t+2}}\left(  0,0,\hat{M}_{L}^{t}\right)  =\max\left\{  -\mu,\psi
f_{\hat{\pi}_{t+1}}\left(  0,\hat{M}_{L}^{t}\right)  \right\}  ,
\label{eq: abs in logs}%
\end{equation}
where $f_{\hat{\pi}_{t}}\left(  \cdot\right)  :=\log f_{\pi_{t}}\left(  \cdot\right)-\log(\pi_*)$,
$\mu:=\log (r\pi_*)$, and $\hat{M}_{t}:=\log M_{t}+\log r$, the latter being in
log-deviation from its absorbing state. We have already established the support restriction $r\pi_*\geq1$ in the main text after Proposition \ref{prop: nonlinear ACS}, which means $\mu\geq0$. Because $\psi>1$, the
difference equation (\ref{eq: abs in logs}) has two steady states, $-\mu$ and $0$, corresponding to ZIR and PIR, respectively. Moreover, the ZIR steady state is stable, while the PIR is unstable. Therefore, for stable equilibria we must have that $f_{\hat{\pi}_{t+1}}\left(  0,\hat{M}_{L}^{t}\right)  \leq 0$, for if $f_{\hat{\pi}_{t+1}}\left(  0,\hat{M}_{L}%
^{t}\right)  >0$, $f_{\hat{\pi}_{t+s}}\left(  0_{s\times1},\hat{M}_{L}^{t}\right)
$ will grow exponentially without bound. So, a stable fundamental solution
must have $f_{\hat{\pi}_{t+1}}\left(  0,\hat{M}_{L}^{t}\right)  \leq0,$ or
equivalently $f_{\pi_{t+1}}\left(  r^{-1},M_{L}^{t}\right)  \leq \pi_*.$

Setting $f_{\pi_{t+1}}\left(  r^{-1},M_{L}^{t}\right)  =\bar{\pi}\leq \pi_*$ in
(\ref{eq: exp trans}), substituting for $E\left(  \left.
\frac{M_{t+1}}{\pi_{t+1}}\right\vert M_{L}^{t}\right)  $ in
(\ref{eq: ACS nonlinear}) and rearranging yields
\begin{equation}
\pi_{t+1}^{L}=\frac{\bar{\pi}p\max\left\{ 1,r\pi_*\left(  \pi_{t}^{L}/\pi_*\right)
^{\psi}\right\}  e^{-r^{L}}}{r\bar{\pi}-\left(  1-p\right)  \max\left\{
1,r \pi_*\left(  \pi_{t}^{L}/\pi_*\right)  ^{\psi}\right\}  },\quad\pi_{t}^{L}<\pi_*\left(
\frac{\bar{\pi}/\pi_*}{1-p}\right)  ^{1/\psi},\label{eq: trans lev}%
\end{equation}
where $\pi_{t}^{L}:=f_{\pi_{t}}\left(  M_{L}^{t}\right)$, for compactness of notation, and the bound on
$\pi_{t}^{L}$ is required for $\pi_{t+1}^{L}$ to be positive. Take logs and define $\hat{\pi}_{t}^{L}:=\log \pi_{t}^{L}-\log\pi_*$, then (\ref{eq: trans lev}) can be written as
%Taking log deviations from $\pi_*$, this %
\begin{equation}
\hat{\pi}_{t+1}^{L}=\left\{
\begin{array}
[c]{ll}%
\log\frac{p\bar{\pi}/\pi_*}{r\bar{\pi}-1+p}-r^{L}, & \hat{\pi}_{t}^{L}\leq-\frac{\mu
}{\psi}\\
\log\frac{p\bar{\pi}/\pi_*}{\bar{\pi}/\pi_*-\left(  1-p\right)  e^{\psi\hat{\pi}_{t}^{L}}
}+\psi\hat{\pi}_{t}^{L}-r^{L}, & -\frac{\mu}{\psi}<\hat{\pi}_{t}^{L}< \bar{\hat{\pi}}_{t}^{L}=\frac
{\log\bar{\pi}/\pi_*-\log\left(  1-p\right)  }{\psi}.
\end{array}
\right.  \label{eq: trans}%
\end{equation}
%where $\hat{\pi}_{t}^{L}:=\log \pi_{t}^{L}-\log\pi_*$. 
\begin{figure}[ptb]%
\centering
\includegraphics[height=3.3503in,width=5.5in]
{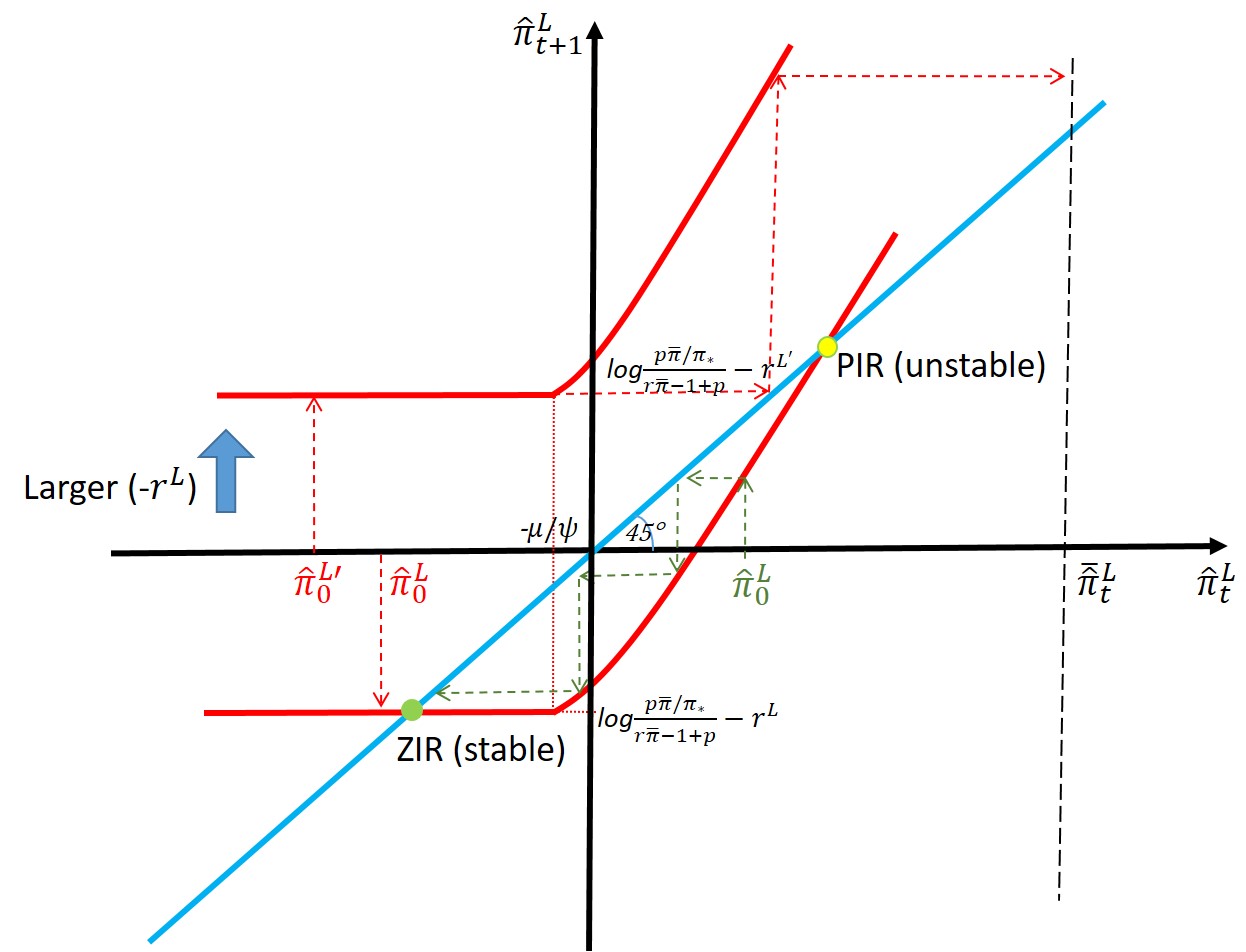}
\caption{Plot of (\ref{eq: trans}) for different values of $-r^{L^\prime}>-r^{L}>0.$}%
\label{fig: ACS nonlinear}%
\end{figure}

Figure \ref{fig: ACS nonlinear} plots (\ref{eq: trans}) 
%with $\bar{\pi}=\pi_*=1$ 
against $\hat{\pi}_t^L$ together with the $\ang{45}$ line. We distinguish two cases. The first case is when the curve intersects with the
$\ang{45}$ line, so that (\ref{eq: trans}) has two (generic) steady states. This
happens when the kink in (\ref{eq: trans}) is below the $\ang{45}$ line. Noting that the kink is given by $(\hat{\pi}_t^L=-\frac{\mu}{\psi};\allowbreak\hat{\pi}_{t+1}^L=\log\frac{p\bar{\pi}/\pi_*}{r\bar{\pi}-1+p}\allowbreak-r^{L})$, then the condition $\hat{\pi}_t^L>\hat{\pi}_{t+1}^L$ becomes 
\begin{equation}
-r^{L}\leq-\frac{\mu}{\psi}-\log\frac{p\bar{\pi}/\pi_*}{r\bar{\pi}-1+p}\leq-\frac
{\log(r\pi_*)}{\psi}-\log\frac{p}{r\pi_*-1+p},\label{eq: supprestr nonlinear ACS}%
\end{equation}
where the second inequality holds because $-\log\frac{p\bar{\pi}/\pi_*}{r\bar{\pi
}-1+p}$ is an increasing function of $\bar{\pi}\leq\pi_*$ and the
definition of $\mu=-\log(r\pi_*).$ When the
restriction (\ref{eq: supprestr nonlinear ACS}) on the support of the shock
holds, then there clearly exist stable solutions to the model for arbitrary
initial conditions $\hat{\pi}^{L}_{0}\leq\hat{\pi}^{PIR}$, where $\hat{\pi}^{PIR}$ is the high-inflation fixed point of the difference equation (\ref{eq: trans}).

Now consider the case when the support restriction
(\ref{eq: supprestr nonlinear ACS}) does not hold. In this case, for any
initial value $\hat{\pi}^{L}_{0}$ the solution of the difference equation
(\ref{eq: trans}) will move along an explosive path while $\hat{\pi}^{L}_{t}$ is
less than $\frac{\log\bar{\pi}/\pi_*-\log\left(  1-p\right)  }{\psi},$ and will
eventually break down after a finite number of periods.

Finally, note how the transitory state resembles the simple case of the absorbing state in the main text, and Figure \ref{fig: ACS nonlinear} parallels Figure \ref{fig: coherency}. At $\bar{\pi}=\pi _{\ast },$ the support restrictions simply implies that  $\hat{\pi}^{ZIR}=\ln \frac{p}{r\pi _{\ast }-1+p}-r^{L}<-\frac{\mu }{\psi }<0.$ So, for an equilibrium to exists the intercept of the ZIR part of the red line must be negative, as in Figure \ref{fig: coherency}.
\end{proof}

\begin{proof}[Proof of Proposition \ref{prop: nonlinear ACS sunspots}]

Sunspot solutions $\pi_t$ may depend on $\varsigma_t$ and its lags. It is assumed that $\varsigma_t$ follows a first-order Markov chain, and so we may denote by $\pi_t^{\varsigma}$ the two different values that $\pi_t$ can take depending on the outcome of the sunspot shock.\footnote{These values may also vary over $t$ if the solution is history dependent, which we do not rule out.} Letting $q_\varsigma := \Pr(\varsigma_{t+1}=1|\varsigma_t=\varsigma)$, (\ref{eq: ACS nonlinear}) becomes
\begin{equation}
1=\max\left\{  r^{-1},\pi_*\left(  \frac{\pi_{t}^{\varsigma_{t}}}{\pi_*}\right)  ^{\psi}\right\}  \left(
\frac{1-q_{\varsigma_{t}}}{\pi_{t+1}^{0}}+\frac{q_{\varsigma_{t}}}{\pi_{t+1}^{1}}\right)
\qquad \varsigma_t = 0,1.\label{eq: ACS sunspots}%
\end{equation}
This is a system of nonlinear difference equations in $\pi_t^{\varsigma}$.  

First, consider the case in which at least one of the initial values $\pi
_{t}^{\varsigma_{t}}$ corresponds to a ZIR, which, wlog, we can
set as $\left(  \pi_{t}^{0}/\pi_*\right)  ^{\psi}\leq (r\pi_*)^{-1}$, since the labelling
of $\varsigma_{t}$ is arbitrary. Under this assumption, (\ref{eq: ACS sunspots})
yields $r=\left(  \frac{1-q_{0}}{\pi_{t+1}^{0}}+\frac{q_{0}}{\pi_{t+1}^{1}%
}\right)  ,$ which we can solve for $\pi_{t+1}^{0}$ and substitute back into
(\ref{eq: ACS sunspots}) with $\varsigma_{t}=1$ to get%
\[
\pi_{t+1}^{1}=\frac{\max\left\{  r^{-1},\pi_*\left(  \pi_{t}^{1}/\pi_*\right)  ^{\psi
}\right\}  \left(  q_{1}-q_{0}\right)  }{1-q_{0}-r\max\left\{  r^{-1},\pi_*\left(
\pi_{t}^{1}/\pi_*\right)  ^{\psi}\right\}  \left(  1-q_{1}\right)  }.
\]
This has almost exactly the same shape as (\ref{eq: trans lev}) that is
plotted in Figure \ref{fig: ACS nonlinear}. Hence, the same argument as above
estabilishes the support restriction $r^{-1}\leq\pi_*$.

Second, suppose $\pi_{t}^{\varsigma_{t}}$ corresponds to a PIR for both $\varsigma_{t},$ i.e., $\left(  \pi_{t}^{\varsigma_{t}}/\pi_*\right)  ^{\psi}>(r\pi_*)^{-1}%
$. By the argument in the previous paragraph, if at any future date $\pi_{t+j}^{\varsigma_{t+j}}$ is a ZIR, then the support restriction for coherency $(r\pi_*)^{-1}\leq1$ applies. So, the only case to consider is when $\left(  \pi
_{t}^{\varsigma_{t}}/\pi_*\right)  ^{\psi}>(r\pi_*)^{-1}$ for all $t$, i.e., the
economy is always at a PIR. In this case, (\ref{eq: ACS sunspots}) becomes
$1=(\pi_{t}/\pi_*)^{\psi}E_{t}\left( \pi_*/\pi_{t+1}\right)  ,$ with the additional
restriction $\pi_{t}>\pi_*(r\pi_*)^{-1/\psi}$ for all $t.$ Because $\psi>1,$ this equation
has the unique stable solution $\pi_{t}=\pi_*$ for all $t$ if and only if
$r^{-1}\leq\pi_*$. 
\end{proof}

\subsection{Derivation of results in Subsection \ref{s: piecewise linear}}\label{app: s: piecewise linear}
\subsubsection{Coefficients of the canonical form} \label{app: s: coeffcanonical}

\paragraph{Coefficients in \nameref{ex: NK}} 
\begin{align*}
A_{0}  & =%
\begin{pmatrix}
1 & -\lambda\\
0 & 1
\end{pmatrix}
,\quad A_{1}=%
\begin{pmatrix}
1 & -\lambda\\
\sigma\psi & 1+\sigma\psi_{x}%
\end{pmatrix}
,\quad B_{0}=B_{1}=%
\begin{pmatrix}
-\beta & 0\\
-\sigma & -1
\end{pmatrix}
,\\
C_{0}  & =%
\begin{pmatrix}
-1 & 0 & 0 & 0\\
0 & -1 & 0 & -\sigma\mu
\end{pmatrix}
,\quad C_{1}=%
\begin{pmatrix}
-1 & 0 & 0 & 0\\
0 & -1 & \sigma & 0
\end{pmatrix}
,\quad D_{0}=D_{1}=0_{2\times4},
\end{align*}
$a=\left(  \psi,\psi_{x}\right)  ^{\prime},$ $b=\left(  0,0\right)
^{\prime}$, $c=\left(0,0,1,\mu\right)^{\prime}$ and $d=0_{4\times1}.$ 
\openbox

\paragraph{Coefficients in \nameref{ex: NK-OP}} 
 
\begin{align*}
A_{0}  & =%
\begin{pmatrix}
1 & -\lambda\\
0 & 1
\end{pmatrix}
,\quad A_{1}=%
\begin{pmatrix}
1 & -\lambda\\
\frac{\lambda}{\gamma} & 1%
\end{pmatrix}
,\quad B_{0}=%
\begin{pmatrix}
-\beta & 0\\
-\sigma & -1
\end{pmatrix}
, \quad B_{1}=%
\begin{pmatrix}
-\beta & 0\\
0 & 0
\end{pmatrix}
,\\
C_{0}  & =%
\begin{pmatrix}
-1 & 0 & 0\\
0 & -1  & -\sigma\mu
\end{pmatrix}
,\quad C_{1}=%
\begin{pmatrix}
-1 & 0 & 0\\
0 & 0 & 0
\end{pmatrix}
,\quad D_{0}=D_{1}=0_{2\times3},
\end{align*}
$a=\left(  0,-\sigma^{-1}\right)  ^{\prime},$ $b=\left(  1,\sigma^{-1}\right)
^{\prime}$, $c=\left(0,\sigma^{-1},\mu\right)^{\prime}$ and $d=0_{3\times1}.$ \hfill\openbox

\subsubsection{Proof of Proposition \ref{prop: NK-TR CC}\label{app: s: prop3}}
In preparation for the proof of Proposition \ref{prop: NK-TR CC}, we first establish a result that will be used in the proofs of both Propositions \ref{prop: NK-TR CC} and \ref{prop: NK cutoff}.

\begin{proposition}\label{prop: NK-TR canonical}
The NK-TR model given by (\ref{eq: NK}) with $u_{t}=\nu_t=0$
and $\epsilon_{t}$ a two-state Markov Chain with transition Kernel $K=%
\begin{pmatrix}
p & 1-p\\
1-q & q
\end{pmatrix}
$ can be
written in the form $F(\mathbf{Y})=\kappa(\mathbf{X})$, where $\mathbf{Y}$ is a $2\times1$ vector containing the values of $\hat{\pi}_t$ in each of the two states, and $F(\cdot)$ is the piecewise linear function (\ref{eq: F}) with 
\begin{equation}
\begin{tabular}{lll}
$\mathcal{A}_{J_{1}}=Q+\lambda \sigma \left( \psi I-\frac{\psi _{x}}{\lambda 
}\left( I-\beta K\right) \right) ,$ & $J_{1}=\left\{ 1,2\right\} $ & $\text{%
(PIR,PIR)}$ \\ 
$\mathcal{A}_{J_{2}}=Q+\lambda \sigma \left( \psi I-\frac{\psi _{x}}{\lambda 
}\left( I-\beta K\right) \right) e_{2}e_{2}^{\prime },$ & $J_{2}=\left\{
2\right\} $ & $\text{(ZIR,PIR)}$ \\ 
$\mathcal{A}_{J_{3}}=Q+\lambda \sigma \left( \psi I-\frac{\psi _{x}}{\lambda 
}\left( I-\beta K\right) \right) e_{1}e_{1}^{\prime },$ & $J_{2}=\left\{
1\right\} $ & $\text{(PIR,ZIR)}$ \\ 
$\mathcal{A}_{J_{4}}=Q,$ & $J_{4}=\varnothing $ & (ZIR,ZIR).%
\end{tabular}%
\   \label{eq: A matrices NK}
\end{equation}%
where $e_{i}$ is the unit vector with $1$ in position $i$,
\begin{equation}\label{eq: Q}
    Q:=I-K-\beta \left( I-K\right) K-\lambda \sigma K
\end{equation}
and 
\begin{equation}
\begin{tabular}{l}
$\det \mathcal{A}_{J_{1}}=\sigma ^{2}\lambda ^{2}\left( \psi -1+\frac{\beta
-1}{\lambda }\psi _{x}\right) \left( \psi -\psi _{p,q,\beta ,\sigma \lambda
}-\psi _{x}\frac{\left( 1+\beta \left( 1-p-q\right) \right) }{\lambda }%
\right) ,$ \\ 
$\det \mathcal{A}_{J_{2}}=-\sigma ^{2}\lambda ^{2}\psi _{p,q,\beta ,\sigma
\lambda }\left( \psi -1+\frac{\beta -1}{\lambda }\psi _{x}\right) +\sigma
\left( 1-q\right) \left[ \rule[0pt]{0pt}{10pt}\sigma \lambda \beta \psi
_{x}\right. $ \\ 
$\qquad \qquad \left. +\lambda \left( \beta \left( p+q-1\right) -1-\sigma
\lambda \right) \left( \frac{\beta -1}{\lambda }\psi _{x}+\psi \right) %
\right] ,$ \\ 
$\det \mathcal{A}_{J_{3}}=-\sigma ^{2}\lambda ^{2}\left( \psi -\psi
_{p,q,\beta ,\sigma \lambda }-\psi _{x}\frac{\left( 1+\beta \left(
1-p-q\right) \right) }{\lambda }\right) $ \\ 
$\qquad \qquad -\sigma \left( 1-q\right) \left[ \left( 1-\left( p+q\right)
\beta +\sigma \lambda -\beta ^{2}\left( 1-p-q\right) \right) \psi
_{x}\right. $ \\ 
$\qquad \qquad \qquad \qquad \left. -\lambda \psi \left( 1+\sigma \lambda
+\beta \left( 1-p-q\right) \right) \right] ,$ \\ 
$\det \mathcal{A}_{J_{4}}=\sigma ^{2}\lambda ^{2}\psi _{p,q,\beta ,\sigma
\lambda }.$%
\end{tabular}
\label{eq: determinants}
\end{equation}
where $\psi _{p,q,\beta ,\sigma \lambda }$ is given in (\ref{eq:
NK cutoff}).
\end{proposition}

\begin{proof} Collect the $k=2$ states of $\epsilon _{t}$ in the vector $\epsilon
=\left( \epsilon ^{1},\epsilon ^{2}\right) ^{\prime }$ and denote the
corresponding states of $\hat{\pi}_{t},\hat{x}_{t},\hat{R}_{t}$ along a MSV
solution by 2-dimensional vectors $\hat{\pi},\hat{x}$ and $\hat{R},$
respectively, where $y=f\left( \epsilon \right) $ for some
function $f\left( \cdot \right) ,$ and for each $y\in \left\{ \hat{\pi}%
,\hat{x},\hat{R}\right\} .$ Because the dynamics are exogenous and
determined completely by $K$, we have $E\left( y_{t+1}|\epsilon
_{t}=\epsilon ^{i}\right) =e_i'K y$. Stacking the two conditioning
states, we can write, with slight abuse of notation, $y%
_{t+1|t}=K\epsilon $. Substituting into (\ref{eq: NK NKPC}) with $u_{t}=0,$
we obtain 
\begin{equation}\label{eq: NK PC 2s}
\hat{\pi}=\beta \overbrace{K\hat{\pi}}^{\hat{\pi}_{t+1|t}}+\lambda \hat{x}.
\end{equation}%

Similarly, from (\ref{eq: NK EE}) we obtain
\begin{equation}\label{eq: NK EE 2s}
\hat{x}=\overbrace{K\hat{x}}^{\hat{x}_{t+1|t}}-\sigma \left( \hat{R}-K\hat{\pi}\right) +\epsilon .
\end{equation}%
Combining the above two equations, we obtain%
\begin{equation*}
\left( I-K\right) \hat{\pi}=\beta \left( I-K\right) K\hat{\pi}-\lambda
\sigma \left( \hat{R}-K\hat{\pi}\right) +\lambda \epsilon .
\end{equation*}%
Substituting for $\hat{R}=\max \left\{ -\mu \iota _{2},\psi \hat{\pi}+\psi
_{x}\hat{x}\right\} ,$ obtained from (\ref{eq: NK TR}) with $\nu _{t}=0,$
and for $\hat{x}=\lambda ^{-1}\left( I-\beta K\right) \hat{\pi},$ and
rearranging we get:%
\begin{equation}
Q\hat{\pi}=-\lambda \sigma \max \left\{ -\mu \iota_2,\left( \psi I-\frac{\psi
_{x}}{\lambda }\left( I-\beta K\right) \right) \hat{\pi}\right\} +\lambda
\epsilon.
\label{eq: NK single eq}
\end{equation}%
This yields (\ref{eq: A matrices NK}). The determinants (\ref{eq: determinants}) were derived using straightforward algebraic calculations (performed using Scientific Workplace). 
% Next, recalling that $K_{11}=p$ and $%
% K_{22}=q$, the determinants of each $\mathcal{A}_{J_{i}}$ in (%
% \ref{eq: A matrices NK}) are given by%
% \begin{align*}
% \det \mathcal{A}_{J_{1}}& =(1-\psi )\left( \det Q-\sigma ^{2}\lambda
% ^{2}\psi \right)  \\
% \det \mathcal{A}_{J_{2}}& =\det Q+\sigma \lambda \psi \left( \left(
% 1-p\right) \left( 1-\left( p+q-1\right) \beta \right) -p\sigma \lambda
% \right)  \\
% \det \mathcal{A}_{J_{3}}& =\det Q+\sigma \lambda \psi \left( \left(
% 1-q\right) \left( 1-\left( p+q-1\right) \beta \right) -q\sigma \lambda
% \right)  \\
% \det \mathcal{A}_{J_{4}}& =\det Q.
% \end{align*}%
% It can also be shown that $\det Q=\sigma ^{2}\lambda ^{2}\psi _{p,q,\beta ,\sigma
% \lambda}$, where $\psi _{p,q,\beta ,\sigma
% \lambda}$ is given in (\ref{eq: NK cutoff}). Substituting for $\det Q$ in the previous display then yields (\ref{eq: determinants}). 
\end{proof}

\begin{proof}[Proof of Proposition \ref{prop: NK-TR CC}]
Setting $\psi_{x}=0$ in (\ref{eq: determinants}), we obtain%
\begin{equation*}
\det \mathcal{A}_{J_{1}} =\sigma ^{2}\lambda ^{2}\left( \psi -1\right)
\left( \psi -\psi _{p,q,\beta ,\sigma \lambda }\right) >0. 
\end{equation*}%
Since $\psi _{p,q,\beta ,\sigma \lambda }\leq 1$, $\det \mathcal{A}%
_{J_{1}}>0$, so coherency requires $\psi _{p,q,\beta ,\sigma \lambda }>0$
for $\det \mathcal{A}_{J_{4}}>0$ from (\ref{eq: determinants}). However, in that case, from (\ref{eq: determinants}) we get 
\begin{equation*}
\det \mathcal{A}_{J_{2}}=-\sigma ^{2}\lambda ^{2}\left( \psi _{p,q,\beta
,\sigma \lambda }\left( \psi -1\right) +\psi \left( 1-q\right) \left( 1+%
\frac{1-\beta \left( p+q-1\right) }{\sigma \lambda }\right) \right) <0
\end{equation*}%
because $\beta \left( p+q-1\right) <1$, violating the CC condition in the \nameref{th: GLM} Theorem.
% Setting $q=1$ in (\ref{eq: NK cutoff}), we get $\psi _{p,1,\beta ,\sigma
% \lambda }=p-\frac{\left( 1-p\beta \right) \left( 1-p\right) }{\sigma \lambda 
% }=p\left( 1-\theta \right) .$ Substituting in (\ref{eq: determinants}) and
% using the definition of $\theta =\frac{\left( 1-p\beta \right) \left(
% 1-p\right) }{p\sigma \lambda }$ yields%
% \begin{equation*}
% \det \mathcal{A}_{J_{2}}=\sigma ^{2}\lambda ^{2}p\left( 1-\theta \right)
% \left( 1-\psi \right) ,\quad \text{and \quad }\det \mathcal{A}%
% _{J_{4}}=\sigma ^{2}\lambda ^{2}p\left( 1-\theta \right) .
% \end{equation*}%
% Hence, for all $\psi >1,$ $\theta <1$ yields $\det \mathcal{A}_{J_{4}}>0$
% and $\det \mathcal{A}_{J_{2}}<0$, while $\theta <1$ yields $\det \mathcal{A}%
% _{J_{4}}<0$ and $\det \mathcal{A}_{J_{2}}>0.$ So the CC condition in the \nameref{th: GLM} Theorem fails.
\end{proof}
\paragraph{Extension to $\psi_x\neq0$} In this case, the Taylor principle becomes
\begin{equation}\label{eq: Taylor Principle}
\psi+\frac{\beta-1}{\lambda}\psi_x > 1.
\end{equation}
It is straightforward to show that the CC condition fails when (\ref{eq: Taylor Principle}) holds for the absorbing case $q=1$. Using this constraint in (\ref{eq: determinants}), we obtain 
\[
\det \mathcal{A}_{J_{2}} = -\det \mathcal{A}_{J_{4}}\left(\psi+\frac{\beta-1}{\lambda}\psi_x-1\right).
\]
Thus, the two determinants must have opposite sign, violating the CC condition in the \nameref{th: GLM} Theorem. 

It seems too complicated to prove this result analytically for $q<1$, but we have verified it numerically for all the parametrizations we considered (see the replication code provided).\openbox
% \paragraph{Part (ii)}

% First, observe that $\det\mathcal{A}_{J_{1}}>0$ holds for all admissible values of the parameters
% $\beta,p,q\in\left[  0,1\right]$, and $\gamma,\lambda>0$, since
% $\gamma\left(  1-\beta\right)  +\lambda^{2}>0$ and $\left(  1+\left(
% 1-p-q\right)  \beta\right)  \geq0$. Therefore, when $\theta>1$ ($\psi_{p,1,\beta,\sigma\lambda}<0$), the CC condition cannot hold because $\det\mathcal{A}%
% _{J_{4}}<0$. Turning to the case $\theta<1$ ($\psi_{p,1,\beta,\sigma\lambda}>0$) we
% immediately notice that both $\det\mathcal{A}_{J_{2}}$ and $\det
% \mathcal{A}_{J_{3}}$ are negative, since the terms in the numerator of the
% fractions are all positive.
%\openbox

\subsubsection{Proof of Proposition \ref{prop: NK-OP CC}} \label{app: s: prop4}
Next, we establish a result that will be used in the proof of Propositions \ref{prop: NK-OP CC}.

\begin{proposition}\label{prop: NK-OP canonical}
The NK-OP model given by (\ref{eq: NK}) with  (\ref{eq: NK TR}) replaced by (\ref{eq: NK OP}) with $u_{t}=\nu_t=\psi_x=0$ and $\epsilon_{t}$ a two-state Markov Chain with transition Kernel $K=
\begin{pmatrix}
p & 1-p\\
1-q & q
\end{pmatrix}
$ can be written in the form (\ref{eq: F}) with 
\begin{equation}%
\begin{tabular}
[c]{ll}%
$\mathcal{A}_{J_{1}}=\left(  1+\frac{\lambda^{2}}{\gamma}\right)  I-\beta K,$
& $J_{1}=\left\{  1,2\right\}  $ \\
$\mathcal{A}_{J_{2}}=I-\beta K-e_{1}e_{1}^{\prime}\left(  K\left(  I-\beta
K\right)  +\lambda\sigma K\right)  +\frac{\lambda^{2}}{\gamma}e_{2}%
e_{2}^{\prime},$ & $J_{2}=\left\{  2\right\}  $ \\
$\mathcal{A}_{J_{3}}=I-\beta K-e_{2}e_{2}^{\prime}\left(  K\left(  I-\beta
K\right)  +\lambda\sigma K\right)  +\frac{\lambda^{2}}{\gamma}e_{1}%
e_{1}^{\prime},$ & $J_{2}=\left\{  1\right\}  $ \\
$\mathcal{A}_{J_{4}}=Q$ &
$J_{4}=\varnothing$ 
\end{tabular}
\ \ \ \label{eq: A matrices NK-OP}%
\end{equation}
and%
\begin{equation}%
\begin{tabular}
[c]{l}%
$\det\mathcal{A}_{J_{1}}=\frac{\left(  \gamma\left(  1-\beta\right)
+\lambda^{2}\right)  \left(  \gamma\left(  1+\left(  1-p-q\right)
\beta\right)  +\lambda^{2}\right)  }{\gamma^{2}},$\\
$\det\mathcal{A}_{J_{2}}=-\frac{\left(  \gamma\left(  1-\beta\right)
+\lambda^{2}\right)  \left(  \sigma\lambda\psi_{p,q,\beta,\sigma\lambda
}+\left(  1-q\right)  \left(  1+\left(  1-p-q\right)  \beta\right)  \right)
+\sigma\lambda\left(  1-q\right)  \left(  \gamma+\lambda^{2}\right)  }{\gamma
},$\\
$\det\mathcal{A}_{J_{3}}=-\frac{\left(  \gamma\left(  1-\beta\right)
+\lambda^{2}\right)  \left(  \sigma\lambda\psi_{p,q,\beta,\sigma\lambda
}+\left(  1-p\right)  \left(  1+\left(  1-p-q\right)  \beta\right)  \right)
+\sigma\lambda\left(  1-p\right)  \left(  \gamma+\lambda^{2}\right)  }{\gamma
},$\\
$\det\mathcal{A}_{J_{4}}=\sigma^{2}\lambda^{2}\psi_{p,q,\beta,\sigma\lambda}.$%
\end{tabular}
\ \ \label{eq: determinants OP}%
\end{equation}
\end{proposition}

\begin{proof} 
From (\ref{eq: NK EE}) and (\ref{eq: NK OP}) we obtain
\begin{equation}
\hat{x}=\left\{
\begin{array}
[c]{ccc}%
K\hat{x}-\sigma\left(  -\mu-K\hat{\pi}\right)  +\epsilon, & \text{if }%
K\hat{\pi}+\frac{1}{\sigma}\left(  K\hat{x}-\hat{x}+\epsilon\right)  \leq-\mu
& \text{(ZIR)}\\
-\frac{\lambda}{\gamma}\hat{\pi}, & \text{if }K\hat{\pi}+\frac{1}{\sigma
}\left(  K\hat{x}-\hat{x}+\epsilon\right)  >-\mu & \text{(PIR)}%
\end{array}
\right.  \label{eq: NK EE 2s_OP}
\end{equation}
where the inequalities are element-wise. Substituting for $\hat{x}$ using
(\ref{eq: NK PC 2s}) yields%
\[
\left(  I-\beta K\right)  \hat{\pi}=\left\{
\begin{array}
[c]{cc}%
K\left(  I-\beta K\right)  \hat{\pi}-\lambda\sigma\left(  -\mu-K\hat{\pi
}\right)  +\lambda\epsilon, & \text{(ZIR)}\\
-\frac{\lambda^{2}}{\gamma}\hat{\pi}, & \text{(PIR)}%
\end{array}
\right.
\]
where ZIR occurs if and only if $K\hat{\pi}+\frac{1}{\lambda\sigma}\left(
\left(  K-I\right)  \left(  I-\beta K\right)  \hat{\pi}+\lambda\epsilon
\right)  \leq-\mu$ (element-wise). Thus, for PIR,PIR we have%
\[
\mathcal{A}_{J_{1}}=\left(  1+\frac{\lambda^{2}}{\gamma}\right)  I-\beta K
\]
For ZIR,PIR, we have%
\[
\mathcal{A}_{J_{2}}=I-\beta K-e_{1}e_{1}^{\prime}\left(  K\left(  I-\beta
K\right)  +\lambda\sigma K\right)  +\frac{\lambda^{2}}{\gamma}e_{2}%
e_{2}^{\prime},
\]
and PIR,ZIR can be obtained symmetrically. For ZIR,ZIR, we have%
\[
\mathcal{A}_{J_{4}}=I-\beta K-\left(  K\left(  I-\beta K\right)
+\lambda\sigma K\right)  =Q.
\]
This yields (\ref{eq: A matrices NK}). Finally, it is straightforward to
verify (\ref{eq: determinants OP}).
\end{proof}

\begin{proof}[Proof of Proposition \ref{prop: NK-OP CC}]
First, observe that $\det\mathcal{A}_{J_{1}}>0$ holds for all admissible values of the parameters
$\beta,p,q\in\left[  0,1\right]$, and $\gamma,\lambda>0$, since
$\gamma\left(  1-\beta\right)  +\lambda^{2}>0$ and $(  1\allowbreak+(
1-\allowbreak p-q)  \beta)  \geq0$. Therefore, when $\theta>1$ ($\psi_{p,1,\beta,\sigma\lambda}<0$), the CC condition cannot hold because $\det\mathcal{A}%
_{J_{4}}<0$. Turning to the case $\theta<1$ ($\psi_{p,1,\beta,\sigma\lambda}>0$) we
immediately notice that both $\det\mathcal{A}_{J_{2}}$ and $\det
\mathcal{A}_{J_{3}}$ are negative, since the terms in the numerator of the
fractions are all positive.
\end{proof}

\subsubsection{Proof of Proposition \ref{prop: NK-TR sup res}} \label{app: s: prop NK-TR sup res}
%\begin{proof}[Proof of Proposition \ref{prop: NK-TR sup res}]\label{app: s: prop NK-TR sup res}

We first look at the absorbing (or steady) state, where $\epsilon _{t}=0.$
Then, we need to solve
\begin{equation}\label{eq: AS/AD abs NK}
\ \hat{\pi}=\frac{\lambda }{1-\beta }\hat{x}\quad AS\quad ;\quad \hat{\pi}
=\max \left\{ -\mu ,\psi \hat{\pi}\right\}= \max \left\{
\begin{array}
[c]{c}%
\psi\frac{\lambda}{1-\beta}\hat{x}\quad AD^{TR}\\
-\mu \qquad AD^{ZLB}
\end{array}.
\right.
\end{equation}
This is depicted in Figure \ref{app: fig: nk_abs}. It is immediately obvious
that the necessary support restriction for existence of a solution is $\mu
\geq 0,$ i.e., $\left( r\pi _{\ast }\right) ^{-1}\leq 1.$ When this holds,
there are two possible solutions: 1) PIR: $(\hat{\pi},\hat{x},\hat{R}%
)=(0,0,0)$; and 2) ZIR: $(\hat{\pi},\hat{x},\hat{R})=(-\mu ,-\mu \frac{%
(1-\beta )}{\lambda },-\mu )$. 
\begin{figure}
\centering\includegraphics[scale = 0.5]{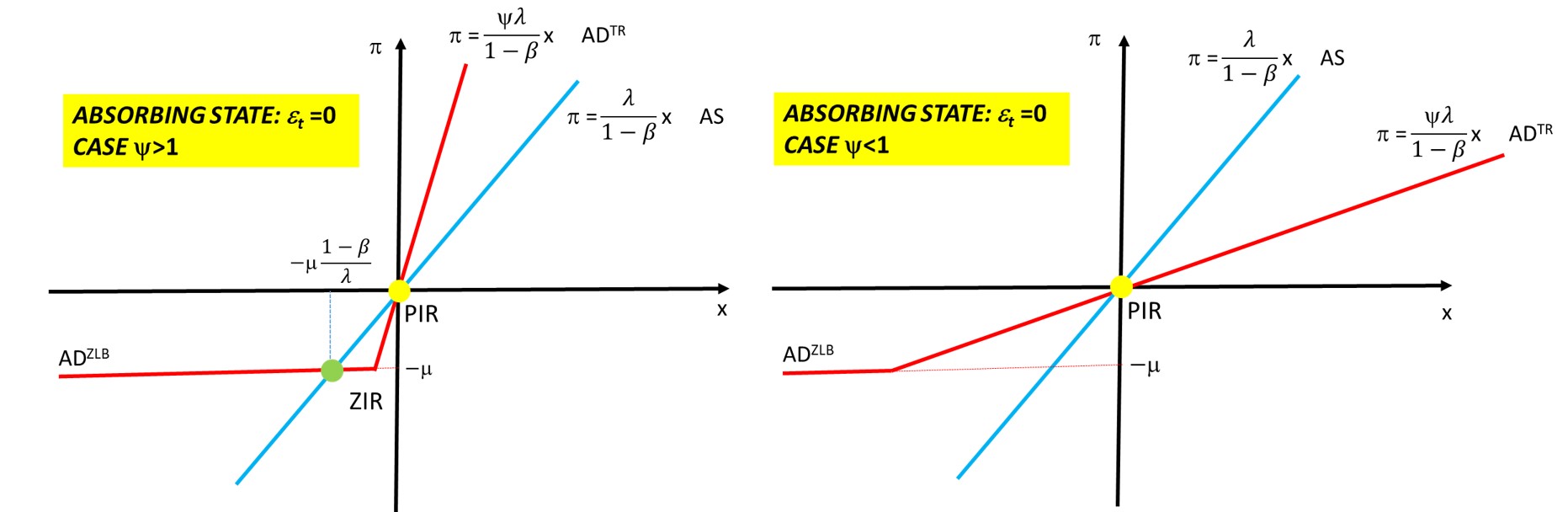}
\caption{The absorbing state in the NK-TR model} 
\label{app: fig: nk_abs}
\end{figure}

Next, turn to the transitory state. Here, there are four possibilities
depending on the value of $\theta$, and the equilibrium in the absorbing
state. These are depicted in Figure \ref{app: fig: nk_psilargerthanp}. The
derivations of those cases is as follows. 

The temporary state lasts for a random time $T$, after which the economy
jumps to the absorbing state, because the model is completely
forward-looking with no endogenous persistence. In the transitory state $%
\epsilon _{t}=-\sigma \hat{M}_{t+1|t}=\sigma pr^{L}<0$, the equilibrium will
be ($\hat{\pi}^{L},\hat{x}^{L}$) and with probability $(1-p)$ we are back in
the absorbing state. The latter can be a PIR one or a ZIR one. 
\begin{figure}
\centering\includegraphics[scale = 0.45]{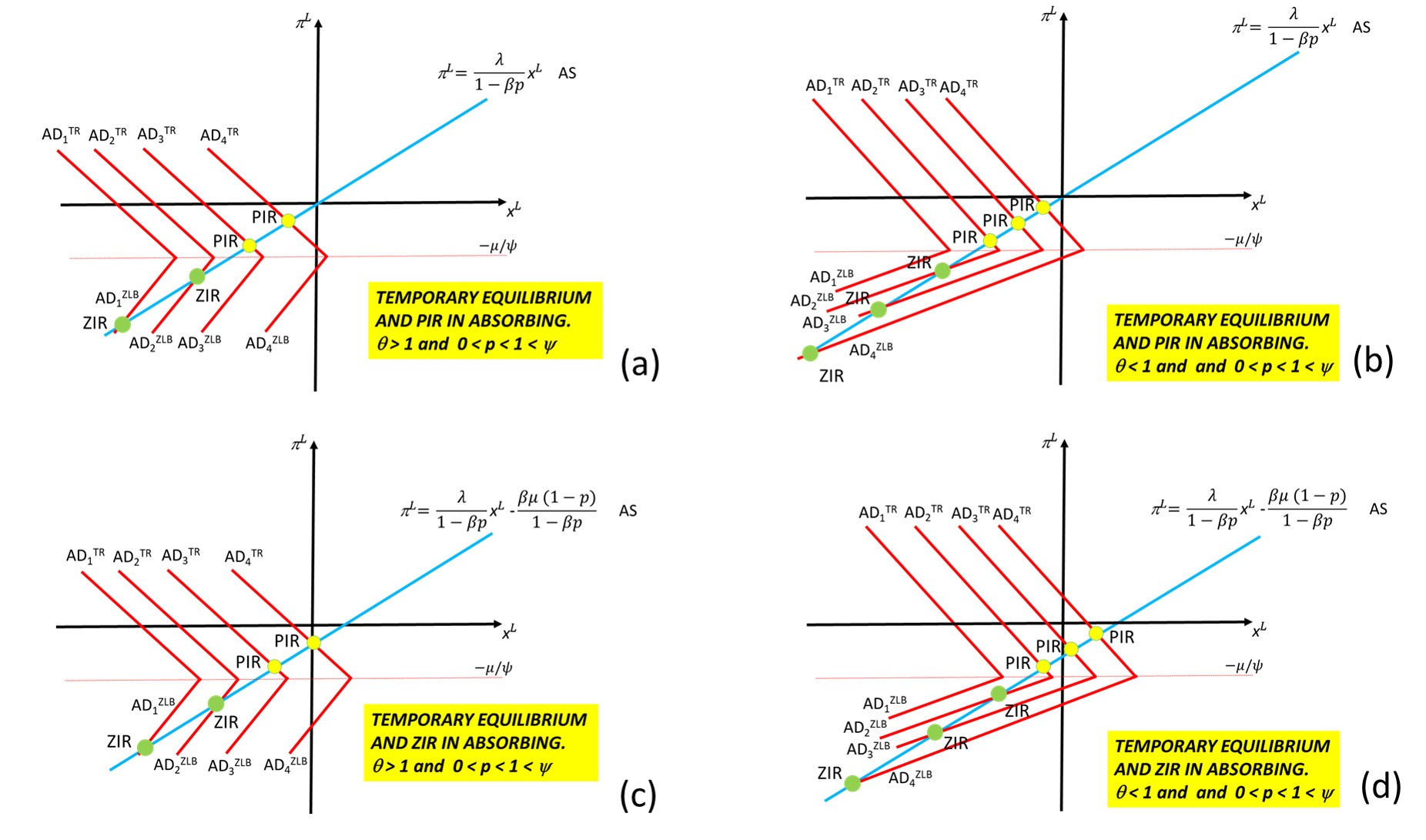}
\caption{The temporary state in the NK model when $\psi>1$.} 
\label{app: fig: nk_psilargerthanp}
\end{figure}

When the absorbing state is PIR, the system becomes 
\begin{align}
\hat{\pi}^{L} & =\frac{\lambda}{1-\beta p}\hat{x}^{L} \hspace{7em} AS \label{app: eq: AS-absPIR} \\
% \end{equation}
% \begin{equation}
\hat{\pi}^{L} & =\left\{
\begin{array}
[c]{c}
\frac{1-p}{\sigma\left(  p-\psi\right)  }\hat{x}^{L}+\frac{ p(-r^{L}%
)}{\left(  p-\psi\right)  }\qquad AD^{TR} \hspace{2em} \text{ for }\pi
\geq-\frac{\mu}{\psi}\\
\frac{1-p}{\sigma p}\hat{x}^{L}-\frac{\mu}{p}+(-r^{L})\qquad
AD^{ZLB} \hspace{1em} \text{ for }\pi\leq-\frac{\mu}{\psi}
\end{array}
\right.  \label{app: eq: AD-absPIR}
\end{align}
These curves are plotted in the top row of Figure \ref{app: fig:
nk_psilargerthanp} for the cases $\theta >1$ on the left, i.e., panel (a), where $AS$ is flatter than $AD^{ZLB}$, and $\theta <1$ on the right, i.e., panel (b), where $AS$ is steeper than $AD^{ZLB}.$ 

When the absorbing state is ZIR, instead, expectations in the temporary
equilibrium are different, so the system to solve for becomes
\begin{align}
\hat{\pi}^{L}  &  =\frac{\lambda}{1-\beta p}\hat{x}^{L}-\frac{\beta\mu(1-p)}{1-\beta p} \hspace{9em} AS
\label{app: eq: AS-absZIR}
\\
\hat{\pi}^{L}  &  =\left\{
\begin{array}
[c]{c}
\hat{x}^{L}\frac{(1-p)}{\sigma\left(  p-\psi\right)  }+\frac{p(-r^{L})}%
{p-\psi}+{\frac{\mu(1-p)}{p-\psi
}\left[  \frac{(1-\beta)}{\lambda\sigma}+1\right]} \quad AD^{TR} \hspace{1em} \text{ for }\pi\geq-\frac{\mu}{\psi}\\
\hat{x}^{L}\frac{1-p}{\sigma p}-\frac{\mu}{p}-r^{L}+ \frac{\mu(1-p)}{p}\left[  \frac{(1-\beta)}{\lambda\sigma}+1\right]   \qquad AD^{ZLB}  \text{ for }\pi\leq -\frac{\mu}{\psi}
\end{array}
\right. \label{app: eq: AD-absZIR}
\end{align}
These curves are plotted in the bottom row of Figure \ref{app: fig:
nk_psilargerthanp} for the cases $\theta >1$ on the left, i.e., panel (c), where $AS$ is flatter than $AD^{ZLB}$, and $\theta <1$ on the right, i.e., panel (d), where $AS$ is steeper than $AD^{ZLB}.$

Inspection of the graphs on the left of Figure \ref{app: fig:
nk_psilargerthanp}, where $\theta >1$ for PIR absorbing (panel (a)) and ZIR
absorbing (panel (c)) shows there is always a solution in both cases. We
therefore conclude that when $\theta >1,$ the only necessary support
restriction is $\left( r\pi _{\ast }\right) ^{-1}\leq 1$ for existence of an
equilibrium in the absorbing state. This proves (\ref{eq: supp restr NK E}).

Next, turn to the case $\theta <1.$ Now it is clear that a further support
restriction is needed on the value of the shock in the transitory state. The
cutoff can be computed by finding the point where the $AD$ and $AS$ curves
intersect at the kink of $AD$. There are two different points for the cases in Figure \ref{app: fig: nk_psilargerthanp}: panel (b), PIR absorbing and panel (d), ZIR absorbing. From inspection, it is clear that the former is the least stringent condition, so it suffices to focus on that. Specifically, we equate (\ref{app: eq: AS-absPIR}) with (\ref{app: eq: AD-absPIR}) at $\hat{\pi}^{L}=-\frac{\mu }{\psi }$ to find the value of the
shock $r^{L}=\bar{r}^{L}$ such that the equations have a solution for all $-r^{L}\leq -\bar{r}^{L}.$ Hence, the cutoff can be found by solving:
\begin{equation*}
-\frac{\mu }{\psi }\frac{1-\beta p}{\lambda }=\sigma \frac{-\left( p-\psi
\right) \frac{\mu }{\psi }+p\bar{r}^{L}}{1-p},
\end{equation*}%
which yields%
\begin{equation*}
-\bar{r}^{L}=\frac{\mu }{\psi }\frac{\left( 1-\beta p\right) \left(
1-p\right) }{p\lambda \sigma }-\frac{\left( p-\psi \right) }{p}\frac{\mu }{%
\psi }=\mu \left( \frac{\psi -p}{\psi p}+\frac{\theta }{\psi}\right) ,
\end{equation*}%
which proves (\ref{eq: supp restr NK B}).\hfill\openbox

%\end{proof}

\subsubsection{Proof of Proposition \ref{prop: NK-OP sup res}} \label{app: s: prop NK-OP sup res}
%\begin{proof}[Proof of Proposition \ref{prop: NK-OP sup res}]\label{app: s: prop NK-OP sup res}
We first look at the absorbing (or steady) state, where $\epsilon _{t}=0.$
Then, the system to solve is 
\begin{equation}\label{eq: AS/AD abs NK_OP}
\ \hat{\pi}=\frac{\lambda }{1-\beta }\hat{x}\quad AS\quad ;\quad \hat{\pi}
=\max \left\{
\begin{array}
[c]{c}
-\frac{\gamma}{\lambda}\hat{x}\quad AD^{OP}\\
-\mu \qquad AD^{ZLB}
\end{array}.
\right.
\end{equation}
This is depicted in Figure \ref{app: fig: nk-OP_abs}. In contrast with the NK-TR case, there are two inequalities to satisfy: the ZLB and the slackness condition on optimal policy, i.e., (\ref{eq: NK OP}). 
In the NK-TR case, there is only the former inequality, while the
Taylor rule is expressed as equality, thus graphically a feasible point above the ZLB needs to be on the $AD^{TR}$ line. Here instead, a feasible point can be below the first order conditions for optimal policy.\footnote{An alternative way to say the same thing is to note that the graph now shows that the $AD$ is
a correspondence and not a function, as in the case in the Taylor
rule case.} In Figure \ref{app: fig: nk-OP_abs} both the PIR and the ZIR are feasible steady states.
The PIR equilibrium is feasible because it satisfies the ZLB constraint,
i.e., is above the horizontal $AD^{ZLB}$ ZLB line.
The ZIR equilibrium is feasible because it satisfies the slackness condition on the first order conditions on optimal policy constraint, i.e., is below the $AD^{OP}$ line.\footnote{Note that there is an upper bound for the output gap defined jointly by optimal policy and the ZLB constraint. This value is given by the intersection of $AD^{OP}$ and $AD^{ZLB}$\ hence: $\hat{x}^{UB}=\frac{\lambda\mu}{\gamma}.$ If monetary authority tries to increase output further along the $AD^{OP}$ then\ eventually it hits the ZLB constraint.}
It is immediately obvious that the necessary support restriction for existence of a solution is $\mu \geq 0,$ i.e., $r^{-1}\leq \pi _{\ast }$. When this holds, there are two possible solutions: 1) PIR: $(\hat{\pi},\hat{x},\hat{R}%
)=(0,0,0)$; and 2) ZIR: $(\hat{\pi},\hat{x},\hat{R})=(-\mu ,-\mu \frac{%
(1-\beta )}{\lambda },-\mu )$. 
\begin{figure}
\centering
\includegraphics[scale = 0.5]{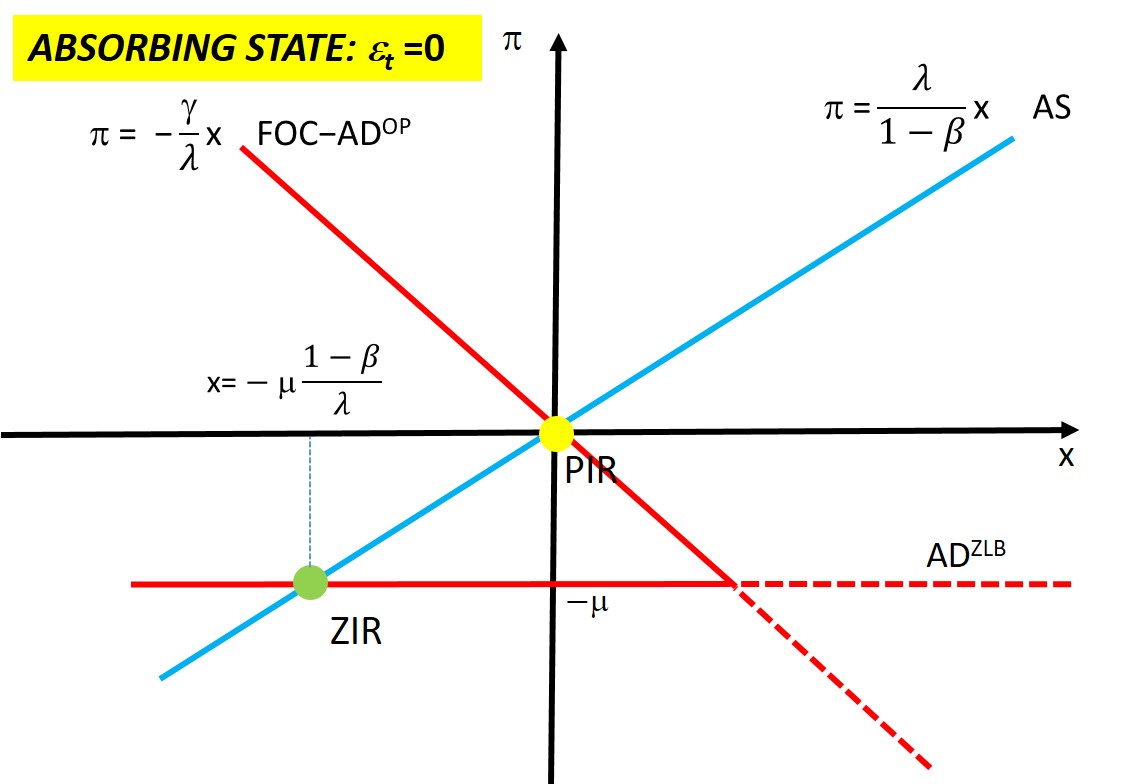}
\caption{The absorbing state in the NK-OP model.} 
\label{app: fig: nk-OP_abs}
\end{figure}

Next, turn to the transitory state. Here, there are four possibilities
depending on the value of $\theta ,$ and the equilibrium in the absorbing
state. These are depicted in Figure \ref{app: fig: nk-OP_temp}. The
derivations of those cases is as follows. 

As before, the temporary state lasts for a random time $T$, after which the economy
jumps to the absorbing state, because the model is completely
forward-looking with no endogenous persistence. In the transitory state $%
\epsilon _{t}=-\sigma \hat{M}_{t+1|t}=\sigma pr^{L}<0$, the equilibrium will
be ($\hat{\pi}^{L},\hat{x}^{L}$) and with probability $(1-p)$ we are back in
the absorbing state. The latter can be a PIR one or a ZIR one. 
\begin{figure}
\centering\includegraphics[scale = 0.5]{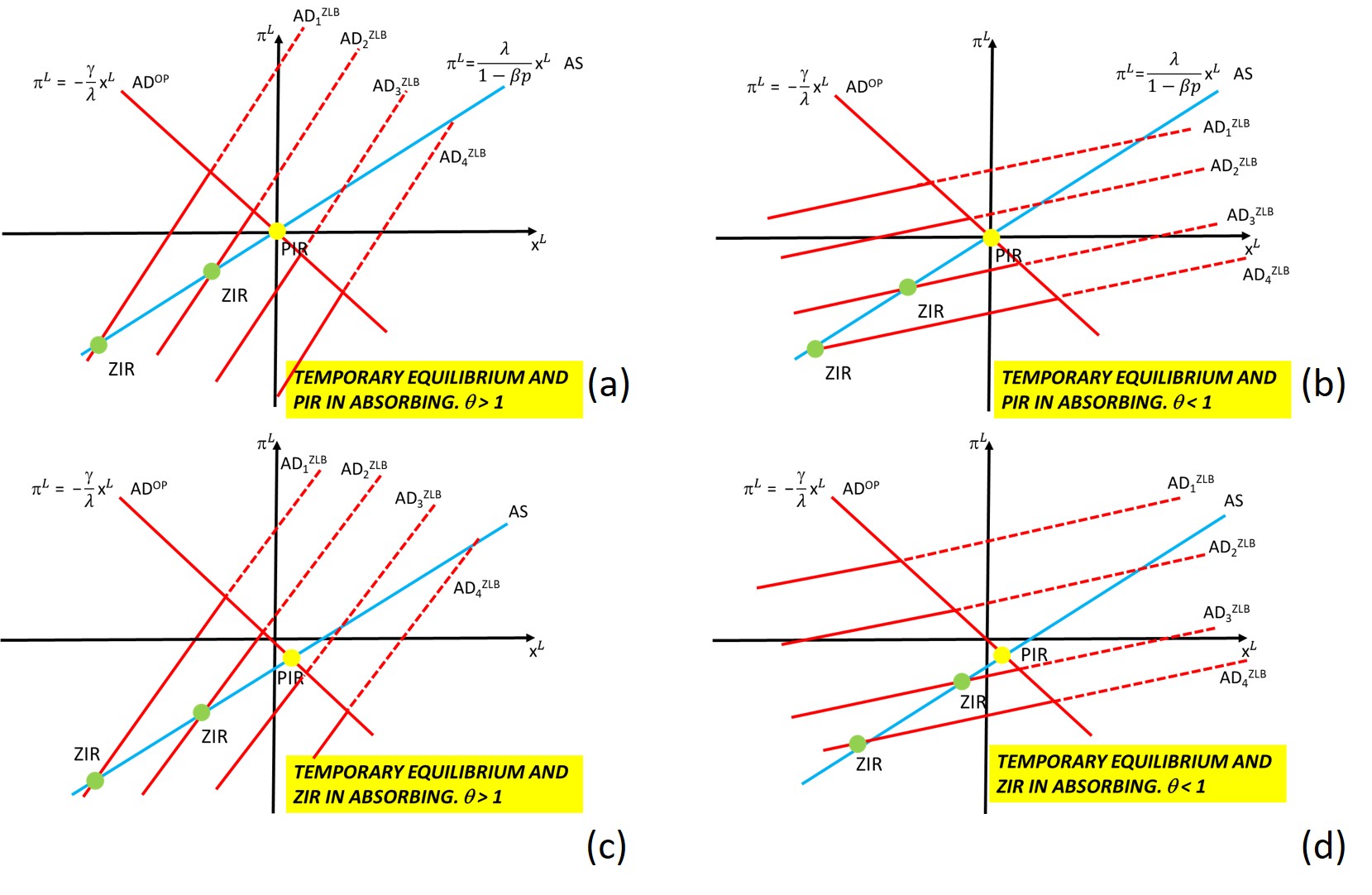}
\caption{The temporary state in the NK-OP model.} 
\label{app: fig: nk-OP_temp}
\end{figure}

When the absorbing state is PIR and the ZLB does not bind, the system becomes
\begin{align}
\hat{\pi}^{L}  &  =\frac{\lambda}{1-\beta p}\hat{x}^{L}\qquad AS\nonumber\\
\hat{\pi}^{L}  &  =-\frac{\gamma}{\lambda}\hat{x}^{L}\qquad
AD^{OP}\nonumber\\
\hat{\pi}^{L}  &  >\hat{\pi}^{L,ZLB}=\hat{x}^{L}\frac{1-p}{\sigma p}-\frac
{\mu}{p}-r^{L}\qquad AD^{ZLB}\label{sys1}
\end{align}
When the absorbing state is PIR and the ZLB binds, then  $\hat{\pi}^{L}=\hat{\pi}^{L,ZLB}$ and $\hat{\pi}^{L,ZLB}$ needs to be smaller than the one the central bank would have chosen to satisfy the first order conditions: $\hat{\pi}^{L} \leq-\frac{\gamma}{\lambda}\hat{x}^{L}.$ The system becomes
\begin{align}
\hat{\pi}^{L}  &  =\frac{\lambda}{1-\beta p}\hat{x}^{L}\qquad
AS\nonumber\\
\hat{\pi}^{L}  &  \leq-\frac{\gamma}{\lambda}\hat{x}^{L}\qquad
AD^{OP}\nonumber\\
\hat{\pi}^{L}  &  =\hat{\pi}^{L,ZLB}=\hat{x}^{L}\frac{1-p}{\sigma p}-\frac
{\mu}{p}-r^{L}\qquad AD^{ZLB}\label{sys2}
\end{align}
The inequality in (\ref{sys1}) states that the equilibrium is above the $AD^{ZLB};$ the inequality (\ref{sys2}) states you that the equilibrium is below the $AD^{OP}.$
These curves are plotted in the top row of Figure \ref{app: fig: nk-OP_temp} for the cases $\theta >1$ on the left, i.e., panel (a), where $AS$ is flatter than $AD^{ZLB}$, and $\theta <1$ on the right, i.e., panel (b), where $AS$ is steeper than $AD^{ZLB}.$ An increase in $-r^{L}$, i.e., an increase in the absolute value of the negative discount factor shock, shifts the $AD^{ZLB}$ upwards. 
In both cases, there exists a threshold level of $(-r^{L}) =  \frac{\mu}{p}$ such that the PIR coincides with the ZIR, that is, such that the intersection between $AS$ and $AD^{OP}$ coincides with the intersection between $AS$ and $AD^{ZLB}$. Hence:

(i) when $\theta>1$, there is unique equilibrium that is a ZIR if $-pr^{L}>\mu$ and a PIR if $-pr^{L}<\mu$;

(ii) when $\theta<1$, there is no equilibrium if $-pr^{L}>\mu$ and 2 equilibria (both a ZIR and a PIR) if $-pr^{L}<\mu$.

When the absorbing state is ZIR, instead, expectations in the temporary
equilibrium are different, and given by\footnote{Note that this exactly as in the Taylor rule case, because the absorbing ZIR is not affected by the policy rule.}
\begin{align*}
E_{t}\left(  \hat{\pi}_{t+1}\right)   &  =p\hat{\pi}^{L}-\mu(1-p),\\
E_{t}\left(  \hat{x}_{t+1}\right)   &  =p\ast(\hat{x}^{L})+(1-p)\ast\left(
-\mu\frac{(1-\beta)}{\lambda}\right)  =p\hat{x}^{L}-\mu\frac{(1-\beta
)(1-p)}{\lambda}.
\end{align*}
When the absorbing state is ZIR and the ZLB does not bind, the system becomes
\begin{align*}
\hat{\pi}^{L}  &  =\frac{\lambda}{1-\beta p}\hat{x}^{L}-\frac{\beta\mu
(1-p)}{1-\beta p}\qquad AS\\
\hat{\pi}^{L}  &  =-\frac{\gamma}{\lambda}\hat{x}^{L}\qquad
AD^{OP}\nonumber\\
\hat{\pi}^{L}  &  >\hat{\pi}^{L,ZLB}=\hat{x}^{L}\frac{1-p}{\sigma p}-r^{L}%
+\mu\left(  \frac{1-p}{p}\frac{1-\beta}{\sigma\lambda}-1\right)  \qquad AD^{ZLB}\nonumber
\end{align*}
In the ZLB, instead
\begin{align*}
\hat{\pi}^{L}  &  =\frac{\lambda}{1-\beta p}\hat{x}^{L}-\frac{\beta\mu
(1-p)}{1-\beta p}\text{ \ \ \ \ \ }AS\\
\hat{\pi}^{L}  &  \leq-\frac{\gamma}{\lambda}\hat{x}^{L}\text{ \ \ \ \ \ \ \ }%
AD^{OP}\nonumber\\
\hat{\pi}^{L}  &  =\hat{\pi}^{L,ZLB}=\hat{x}^{L}\frac{1-p}{\sigma p}-r^{L}%
+\mu\left(  \frac{1-p}{p}\frac{1-\beta}{\sigma\lambda}-1\right)  \text{
\ \ \ \ \ \ \ }AD^{ZLB}\nonumber
\end{align*}
These curves are plotted in the bottom row of Figure \ref{app: fig: nk-OP_temp} for the cases $\theta >1$ on the left, i.e., panel (c), where $AS$ is flatter than $AD^{ZLB}$, and $\theta <1$ on the right, i.e., panel (d), where $AS$ is steeper than $AD^{ZLB}.$
In both cases, $r^{L}$ shifts the $AD^{ZLB}$ and there exists a threshold level of $(-r^{L}) = \overline{\left(  -r^{L}\right)} >  \frac{\mu}{p}$ such that the PIR coincides with the ZIR, that is, such that the intersection between $AS$ and $AD^{OP}$ coincides with the intersection between $AS$ and $AD^{ZLB}$. Hence:

(i) when $\theta>1$, there is unique equilibrium that is a ZIR if $(-r^{L}) > \overline{\left(  -r^{L}\right)}$ and a PIR if $(-r^{L}) < \overline{\left(  -r^{L}\right)}$;

(ii) when $\theta<1$, there is no equilibrium if $(-r^{L}) > \overline{\left(  -r^{L}\right)}$ and 2 equilibria (both a ZIR and a PIR) if $(-r^{L}) < \overline{\left(  -r^{L}\right)}$.

When $\theta >1$, thus, for PIR absorbing (panel (a)) and ZIR
absorbing (panel (c)) there is always a solution in both cases. We
therefore conclude that when $\theta >1,$ the only necessary support
restriction is $r^{-1}\leq \pi _{\ast }$ for existence of an equilibrium in the absorbing state. This proves (\ref{eq: supp restr NK-OP E}). When $\theta <1$, as evident from the graph and easy to prove, $ \overline{\left(  -r^{L}\right)} < \frac{\mu}{p}$. Thus, the relevant support restriction for coherency is given by $-r^{L}<\mu/p$, which is (\ref{eq: supp restr NK-OP B}).\hfill\openbox
%\end{proof}

\subsubsection{Existence of sunspot equilibria in NK-TR model}\label{app: s: sunspots}
Consider the NK-TR model in Proposition \ref{prop: NK-TR sup res} with the
additional restriction $\epsilon _{t}=0$ and suppose there is a sunspot
shock $\varsigma _{t}\in \left\{ 0,1\right\} $ with transition matrix $K$.
In this case, the vector of exogenous state variables in the canonical
representation (\ref{eq: canon}) can be written as $X_{t}=\left( 1,\varsigma
_{t}\right) ^{\prime }$. The model can be written as a piecewise linear
system of equations $F\left( \mathbf{Y}\right) =\kappa ,$ where $F\left(
\cdot \right) $ is given by (\ref{eq: F}) with $\mathcal{A}_{J}$ given by
Proposition \ref{prop: NK-TR canonical} as before, since the sunspot shock
affects the expectations in exactly the same way as a real shock would have.
The RHS terms $\kappa $ can be obtained from (\ref{eq: NK single eq}) with $%
\psi _{x}=0$ and $\epsilon =0$, that is,%
\begin{equation*}
\begin{tabular}{lll}
$\kappa _{J_{1}}=0_{2\times 1},$ & $J_{1}=\left\{ 1,2\right\} $ & $\text{%
(PIR,PIR)}$ \\ 
$\kappa _{J_{2}}=\lambda \sigma \mu e_{1},$ & $J_{2}=\left\{ 2\right\} $ & $%
\text{(ZIR,PIR)}$ \\ 
$\kappa _{J_{3}}=\lambda \sigma \mu e_{2},$ & $J_{2}=\left\{ 1\right\} $ & $%
\text{(PIR,ZIR)}$ \\ 
$\kappa _{J_{4}}=\lambda \sigma \mu \iota _{2},$ & $J_{4}=\varnothing $ & 
(ZIR,ZIR).%
\end{tabular}%
\end{equation*}%
The four potential equilibria (solutions) are given by $\hat{\pi}_{J}:=%
\mathcal{A}_{J}^{-1}\kappa _{J},$ i.e., 
\begin{equation}
\begin{tabular}{ll}
$\hat{\pi}_{J_{1}}=0_{2\times 2},$ & $\text{(PIR,PIR)}$ \\ 
$\hat{\pi}_{J_{2}}=\mu 
\begin{pmatrix}
\frac{a_{q}+\sigma \lambda -\sigma \lambda \psi }{\psi a_{p}+\sigma \lambda
\left( \psi -\psi _{p,q,\beta ,\sigma \lambda }\right) } \\ 
\frac{a_{q}}{\psi a_{p}+\sigma \lambda \left( \psi -\psi _{p,q,\beta ,\sigma
\lambda }\right) }%
\end{pmatrix}%
,$ & $\text{(ZIR,PIR)}$ \\ 
$\hat{\pi}_{J_{3}}=\mu 
\begin{pmatrix}
\frac{a_{p}}{\psi a_{q}+\sigma \lambda \left( \psi -\psi _{p,q,\beta ,\sigma
\lambda }\right) } \\ 
\frac{a_{p}+\sigma \lambda -\sigma \lambda \psi }{\psi a_{q}+\sigma \lambda
\left( \psi -\psi _{p,q,\beta ,\sigma \lambda }\right) }%
\end{pmatrix}%
,$ & $\text{(PIR,ZIR)}$ \\ 
$\hat{\pi}_{J_{4}}=-\mu \iota _{2},$ & (ZIR,ZIR),%
\end{tabular}
\label{eq: sunspot sols}
\end{equation}%
where we used the definitions%
\begin{eqnarray*}
a_{q} &:=\left( q-1\right) \left( \beta \left( 1-p-q\right) +\sigma \lambda
+1\right) \leq 0 \\
a_{p} &:=\left( p-1\right) \left( \beta \left( 1-p-q\right) +\sigma \lambda
+1\right) \leq 0,
\end{eqnarray*}%
for compactness, and the fact that 
\begin{eqnarray*}
a_{p}+a_{q}+\sigma \lambda  &=&\left( q-1\right) \left( \beta \left(
1-p-q\right) +\sigma \lambda +1\right) + \\
&&\left( p-1\right) \left( \beta \left( 1-p-q\right) +\sigma \lambda
+1\right) +\sigma \lambda  \\
&=&\left( p+q-2\right) \left( \beta \left( 1-p-q\right) +1\right) +\left(
p+q-1\right) \sigma \lambda  \\
&=&\sigma \lambda \psi _{p,q,\beta ,\sigma \lambda }.
\end{eqnarray*}

Note that the PIR,PIR and ZIR,ZIR equilibria are actually sunspotless in the
sense that they are completely independent of the sunspot process. (They
don't depend on $K,$ i.e., $p,q$). This is perfectly intuitive, because the
sunspot would be effectively choosing over two identical outcomes in each
state. For existence of any of those equlibria, the support restirction is $%
\mu \geq 0$. So, it remains to show that there is no weaker condition that
can support any of the other two equilibria ZIR,PIR or PIR,ZIR. That is, we
need to check if any of the two sets of inequalities:%
\begin{equation*}
\begin{pmatrix}
\frac{\left( a_{q}+\sigma \lambda -\sigma \lambda \psi \right) \mu }{\psi
a_{p}+\sigma \lambda \left( \psi -\psi _{p,q,\beta ,\sigma \lambda }\right) }%
\leq -\frac{\mu }{\psi } \\ 
\frac{a_{q}\mu }{\psi a_{p}+\sigma \lambda \left( \psi -\psi _{p,q,\beta
,\sigma \lambda }\right) }>-\frac{\mu }{\psi }%
\end{pmatrix}%
\text{\quad or\quad }%
\begin{pmatrix}
\frac{a_{p}\mu }{\psi a_{q}+\sigma \lambda \left( \psi -\psi _{p,q,\beta
,\sigma \lambda }\right) }>-\frac{\mu }{\psi } \\ 
\frac{\left( a_{p}+\sigma \lambda -\sigma \lambda \psi \right) \mu }{\psi
a_{q}+\sigma \lambda \left( \psi -\psi _{p,q,\beta ,\sigma \lambda }\right) }%
\leq -\frac{\mu }{\psi }%
\end{pmatrix}%
\end{equation*}%
can be satisfied when $\mu <0$. 

Assuming $\mu <0$ and cancelling out $\mu ,$ we have 
\begin{equation}
\begin{pmatrix}
\frac{\left( a_{q}+\sigma \lambda -\sigma \lambda \psi \right) }{\psi
a_{p}+\sigma \lambda \left( \psi -\psi _{p,q,\beta ,\sigma \lambda }\right) }%
\geq -\frac{1}{\psi } \\ 
\frac{a_{q}}{\psi a_{p}+\sigma \lambda \left( \psi -\psi _{p,q,\beta ,\sigma
\lambda }\right) }<-\frac{1}{\psi }%
\end{pmatrix}%
\text{\quad or\quad }%
\begin{pmatrix}
\frac{a_{p}}{\psi a_{q}+\sigma \lambda \left( \psi -\psi _{p,q,\beta ,\sigma
\lambda }\right) }<-\frac{1}{\psi } \\ 
\frac{\left( a_{p}+\sigma \lambda -\sigma \lambda \psi \right) }{\psi
a_{q}+\sigma \lambda \left( \psi -\psi _{p,q,\beta ,\sigma \lambda }\right) }%
\geq -\frac{1}{\psi }%
\end{pmatrix}%
.  \label{eq: inequalities}
\end{equation}%
Given that $a_{p}\leq 0$ and $a_{q}\leq 0,$ the bottom inequality on the LHS
and the top inequality on the RHS both imply that $\psi a_{p}+\sigma \lambda
\left( \psi -\psi _{p,q,\beta ,\sigma \lambda }\right) >0$ and $\psi
a_{q}+\sigma \lambda \left( \psi -\psi _{p,q,\beta ,\sigma \lambda }\right)
>0$, respectively. Then, in a ZIR,PIR equilibrium, the top inequality on the
left of (\ref{eq: inequalities}) implies%
\begin{eqnarray}
\psi \left( a_{q}+a_{p}+\sigma \lambda -\sigma \lambda \psi \right) & \geq  &
-\sigma \lambda \left( \psi -\psi _{p,q,\beta ,\sigma \lambda }\right), \quad \text{or}\\ \nonumber
-\psi \sigma \lambda \left( \psi -\psi _{p,q,\beta ,\sigma \lambda }\right)
& \geq & -\sigma \lambda \left( \psi -\psi _{p,q,\beta ,\sigma \lambda }\right) ,
 \label{eq: inequality2}
\end{eqnarray}
which cannot hold, since $\psi >1$ and $\psi _{p,q,\beta ,\sigma \lambda
}\leq 1$. An entirely symmetric argument can be
used to rule out a PIR,ZIR -- the top inequality on the RHS of (\ref{eq:
inequalities}) also leads to (\ref{eq: inequality2}). 

\subsubsection{Relationship to \citet[Proposition 1]{Nakataschmidt2019JME}}\label{app: NS}

The model in \cite{Nakataschmidt2019JME} (henceforth NS) corresponds to
(\ref{eq: NK NKPC}) with $u_{t}=0,$ (\ref{eq: NK EE}) with $\epsilon
_{t}=-\sigma\hat{M}_{t+1|t}$ and (\ref{eq: NK OP}). They denote their AD shock
as $r_{t}^{n}:=\mu-\hat{M}_{t+1|t},$ in our notation, and assume that it
follows a two-state Markov process with support $\left\{  r_{L}^{n},r_{H}%
^{n}\right\}  ,$ where $r_{L}^{n}<0<r_{H}^{n},$ and transition probabilities
$\Pr\left(  r_{t+1}^{n}=r_{L}^{n}|r_{t}^{n}=r_{j}^{n}\right)  =p_{j}$ for
$j\in\left\{  L,H\right\}  .$ This translates in our notation to $0>r_{L}%
^{n}=\mu+pr^{L},$ i.e., $-r^{L}p>\mu,$ and $0<r_{H}^{n}=\mu,$ i.e.,
$r^{-1}<\pi_{\ast}.$ The transition probabilities are in our notation
$p_{L}=p$ and $p_{H}=1-q.$ When the `high' state is absorbing ($q=1$),
we have $p_{H}=0$ in their notation.

Specializing to the case $p_{H}=0,$ NS Proposition 1 states that an
equilibrium exists if and only if the following condition holds%
\begin{equation}
p\leq p_{L}^{\ast}\text{ \ \ and \ \ }0\leq p_{H}^{\ast}%
,\label{eq: NS conditions}%
\end{equation}
where
\begin{align*}
p_{L}^{\ast}  & =\frac{-q_{1}+\sqrt{q_{1}^{2}-4q_{2}q_{0}}}{2q_{2}},\\
q_{0}  & =-\left(  \lambda^{2}+\gamma\left(  1-\beta\right)  \right)  \frac
{1}{\sigma\lambda}<0,\\
q_{1}  & =\left(  \lambda^{2}+\gamma\left(  1-\beta\right)  \right)  \left(
\frac{1+\beta}{\sigma\lambda}+1\right)  =-q_{0}\left(  1+\beta+\sigma
\lambda\right)  >0,\\
q_{2}  & =-\left(  \lambda^{2}+\gamma\left(  1-\beta\right)  \right)
\frac{\beta}{\sigma\lambda}=\beta q_{0}<0,
\end{align*}
so that
\begin{align}
p_{L}^{\ast} &  =\frac{-q_{1}+\sqrt{q_{1}^{2}-4\beta q_{0}^{2}}}{2\beta q_{0}%
}=\frac{q_{0}\left(  1+\beta+\sigma\lambda\right)  -q_{0}\sqrt{\left(
1+\beta+\sigma\lambda\right)  ^{2}-4\beta}}{2\beta q_{0}}\nonumber\\
&  =\frac{1+\beta+\sigma\lambda -\sqrt{\left(  1+\beta
+\sigma\lambda\right)  ^{2}-4\beta}}{2\beta},\label{eq: p_L^*}%
\end{align}
and%
\begin{align}
p_{H}^{\ast} &  =\frac{-\phi_{1}-\sqrt{\phi_{1}^{2}-4\phi_{2}\phi_{0}}}%
{2\phi_{2}},\label{eq: p_H^*}\\
\phi_{0} &  =-\left(  \frac{1-p}{\sigma\lambda}\left(  1-\beta p\right)
-p\right)  \frac{\mu}{\mu+pr^{L}}>0,\nonumber\\
\phi_{1} &  =-\frac{1-\beta p+\left(  1-p\right)  \beta\frac{\mu}{\mu+pr^{L}}%
}{\sigma\lambda}-\frac{\lambda^{2}+\left(  1-\beta\frac{\mu}{\mu+pr^{L}%
}\right)  \gamma}{\lambda^{2}+\gamma\left(  1-\beta\right)  },\nonumber\\
\phi_{2} &  =-\frac{\beta}{\sigma\lambda}<0.\label{eq: phi2}%
\end{align}

Substituting for $p_{L}^{\ast}$ in the first inequality in
(\ref{eq: NS conditions}) using (\ref{eq: p_L^*}), we obtain%
\begin{equation}
p<\frac{ 1+\beta+\sigma\lambda -\sqrt{\left(  1+\beta
+\sigma\lambda\right)  ^{2}-4\beta}}{2\beta}.\label{eq: NS cond 1}%
\end{equation}
This is equivalent to the condition $\theta>1$ in
(\ref{eq: supp restr NK-OP E}). Specifically, note that $\theta=\frac{\left(
1-p\right)  \left(  1-\beta p\right)  }{\sigma\lambda p}>1$ is equivalent to%
\begin{equation}
\left(  1-p\right)  \left(  1-\beta p\right)  -\sigma\lambda
p>0.\label{eq: theta quadratic}%
\end{equation}
The discriminant of the quadratic equation $\left(  1-p\right)  \left(
1-\beta p\right)  -\sigma\lambda p=0$ is $\left(  1+\beta+\sigma
\lambda\right)  ^{2}-4\beta=\left(  1-\beta\right)  ^{2}+2\sigma\lambda
+\sigma^{2}\lambda^{2}+2\sigma\beta\lambda>0,$ so the equation has real roots
$p_{1}\leq p_{2}$ given by%
\[
p_{1}=\frac{1+\beta+\sigma\lambda  - \sqrt{\left(
1+\beta+\sigma\lambda\right)  ^{2}-4\beta}}{2\beta},\ p_{2}=\frac{
1+\beta+\sigma\lambda  +\sqrt{\left(  1+\beta+\sigma\lambda\right)
^{2}-4\beta}}{2\beta}.
\]
Thus, $\theta>1$ is equivalent to $p<p_{1}=p_{L}^{\ast},$ which is NS's condition (\ref{eq: NS cond 1})$.$

Next, turn to the second inequality ($p_{H}^{\ast}\geq0$) in
(\ref{eq: NS conditions}). From (\ref{eq: p_H^*}) and (\ref{eq: phi2}), this
is equivalent to%
\[
-\phi_{1}\leq\sqrt{\phi_{1}^{2}-4\phi_{2}\phi_{0}}.
\]
The inequality is obviously satisfied for $\phi_{1}>0,$ and therefore, it is
only a restriction on how negative $\phi_{1}$ can be. In particular, it cannot
fall below $-\sqrt{\phi_{1}^{2}-4\phi_{2}\phi_{0}},$ so, equivalently, when
$\phi_{1}<0,$ we must have $\left\vert \phi_{1}\right\vert \leq\sqrt{\phi
_{1}^{2}-4\phi_{2}\phi_{0}}$, which is clearly equivalent to $\phi_{2}\phi
_{0}\leq0$. Hence, the second condition of NS is equivalent to%
\[
\phi_{2}\phi_{0}=\frac{\beta}{\sigma\lambda}\left(  \frac{\left(  1-p\right)
\left(  1-\beta p\right)  }{\sigma\lambda}-p\right)  \frac{\mu}{\mu+pr^{L}%
}=\frac{\beta p\left(  \theta-1\right)  }{\sigma\lambda}\frac{\mu}{\mu+pr^{L}%
}\leq0.
\]
Since NS assumed $\mu+pr^{L}<0$ and $\mu>0,$ it must be that $\theta>1$. So,
under NS's restrictions on the support $r_{L}^{n}<0<r_{H}^{n},$ the condition
(\ref{eq: NS conditions}) in NS Proposition 1 is equivalent to $\theta>1$ in
our Proposition \ref{prop: NK-OP sup res}.

% \subsection{Proof of Proposition \ref{prop: NK cutoff}}\label{app: Proof NK cutoff}
\subsection{Derivation of results in Subsection \ref{s: cc conditions k}}

\begin{proof}[Proof of Proposition \ref{prop: NK cutoff}]
Proposition \ref{prop: NK-TR canonical} expresses the model in the form (\ref{eq: F}) and gives $\det \mathcal{A}_{J_{i}},$ $%
i=1,...,4$. We need to find the range of parameters for which all $\det \mathcal{A}_{J_{i}}$ are of the same sign. Inspection of (\ref{eq: determinants}) shows we need to consider the following two cases. 

\paragraph{Case $\protect\psi_{p,q,\protect\beta,\protect\sigma\protect%
\lambda} > 0$.}

For CC we need all determinants to be positive. First, observe that $\psi
_{p,q,\beta ,\sigma \lambda }\allowbreak =p+q-1-\allowbreak \frac{\left(
1-\left( p+q-1\right) \beta \right) \left( 2-p-q\right) }{\sigma \lambda }%
\leq 1,$ because $p+q-1\leq 1$ and  $(1-\allowbreak (p+q-1)\beta
)\allowbreak (2-p-q)\allowbreak \geq 0$. Thus, $\det \mathcal{A}_{J_{1}}>0$
implies%
\begin{equation}
\psi <\psi _{p,q,\beta ,\sigma \lambda }\quad \text{or\quad }\psi >1.
\label{eq: detA1>0}
\end{equation}%
For $\det \mathcal{A}_{J_{2}}>0$ we need%
\begin{equation*}
\sigma \lambda \psi \left( \left( 1-p\right) \left( 1-\left( p+q-1\right)
\beta \right) -p\sigma \lambda \right) +\sigma ^{2}\lambda ^{2}\psi
_{p,q,\beta ,\sigma \lambda }>0.
\end{equation*}%
Now, observe that $\psi _{p,q,\beta ,\sigma \lambda }>0$ implies $%
(p+q-1)\lambda \sigma \allowbreak >(1-\allowbreak (p+q-1)\beta )(2-p-q),$
which, in turn, implies 
\begin{equation*}
\left( 1-p\right) \left( 1-\left( p+q-1\right) \beta \right) -p\sigma
\lambda <-\left( 1-q\right) \left( \lambda \sigma +\left( 1-\left(
p+q-1\right) \beta \right) \right) <0,
\end{equation*}%
Therefore, $\det \mathcal{A}_{J_{2}}>0$ implies 
\begin{eqnarray}
\psi  &<&\frac{\sigma \lambda \psi _{p,q,\beta ,\sigma \lambda }}{p\sigma
\lambda -\left( 1-p\right) \left( 1-\left( p+q-1\right) \beta \right) } 
\notag \\
&=&\frac{\left( p+q-1\right) \sigma \lambda -\left( 1-\left( p+q-1\right)
\beta \right) \left( 2-p-q\right) }{p\sigma \lambda -\left( 1-p\right)
\left( 1-\left( p+q-1\right) \beta \right) }<1,  \label{eq: psi bound}
\end{eqnarray}%
the last inequality following from 
\begin{align}
& \left( p+q-1\right) \sigma \lambda -\left( 1-\left( p+q-1\right) \beta
\right) \left( 2-p-q\right) -p\sigma \lambda +\left( 1-p\right) \left(
1-\left( p+q-1\right) \beta \right)   \notag \\
& =-\left( 1-q\right) \sigma \lambda -\left( 1-\left( p+q-1\right) \beta
\right) \left( 1-q\right) <0.  \label{eq: aux ineq}
\end{align}%
An entirely symmetric argument applies for $\det \mathcal{A}_{J_{3}}$.
Hence, combining (\ref{eq: psi bound}) and (\ref{eq: detA1>0}), we obtain $%
\psi <\psi _{p,q,\beta ,\sigma \lambda }$, which is (\ref{eq: CC-NK-2s B}).

\paragraph{Case $\protect\psi_{p,q,\protect\beta,\protect\sigma\protect%
\lambda}<0$.}

The CC now requires $\det \mathcal{A}_{J_{i}}<0$ for all $i.$ For $\det 
\mathcal{A}_{J_{1}}<0,$ we need $\psi_{p,q,\beta,\sigma\lambda}<\psi <1.$
Next, we turn to $\det \mathcal{A}_{J_{2}}<0$%
\begin{equation*}
\sigma^2\lambda^2\psi_{p,q,\beta,\sigma\lambda}+\sigma \lambda \psi \left(
\left( 1-p\right) \left( 1-\left( p+q-1\right) \beta \right) -p\sigma
\lambda \right) <0.
\end{equation*}%
If $\left( 1-p\right) \left( 1-\left( p+q-1\right) \beta \right) \allowbreak
-\allowbreak p\sigma \lambda <0,$ then 
\begin{equation*}
\psi >\frac{\sigma\lambda\psi_{p,q,\beta,\sigma\lambda}}{\left( p\sigma
\lambda -\left( 1-p\right) \left( 1-\left( p+q-1\right) \beta \right)
\right) }=\frac{\psi_{p,q,\beta,\sigma\lambda}}{\left( p-\frac{\left(
1-p\right) \left( 1-\left( p+q-1\right) \beta \right) }{\sigma \lambda }%
\right) }<\psi_{p,q,\beta,\sigma\lambda}.
\end{equation*}%
So, this condition is satisfied for all $\psi >\psi_{p,q,\beta,\sigma\lambda}
$. Next, if $\left( 1-p\right) \left( 1-\left( p+q-1\right) \beta \right)
\allowbreak -\allowbreak p\sigma \lambda >0$, then

\begin{align*}
\frac{1}{\sigma^2 \lambda^2 }\det \mathcal{A}_{J_{2}}&
=\psi_{p,q,\beta,\sigma\lambda}+\psi\frac{ \left( 1-p\right) \left( 1-\left(
p+q-1\right) \beta \right)}{\sigma\lambda} \\
\text{ }& <\psi_{p,q,\beta,\sigma\lambda}+\frac{\left( 1-p\right) \left(
1-\left( p+q-1\right) \beta \right)}{\sigma \lambda}-p <0,
\end{align*}%
where the first inequality follows from $\psi <1$ and the second inequality follows from $\psi_{p,q,\beta,\sigma\lambda}<0$
and (\ref{eq: aux ineq}). An entirely symmetric argument applies for $\det 
\mathcal{A}_{J_{3}}<0$. Hence, we have established that the CC condition in
this case is $\psi_{p,q,\beta,\sigma\lambda}<\psi <1$, which is (\ref{eq: CC-NK-2s E}).
\end{proof}

% \subsection{Derivation of equation (\ref{eq: NK EE UMP})\label{s: QE}}
\subsection{Derivation of results in Subsection \ref{s: UMP}}\label{s: QE}

\begin{proof}[Derivation of equation (\ref{eq: NK EE UMP})]

This is a simplified version of the New Keynesian model of bond market
segmentation that appears in \cite{IkedaLiMavroeidisZanetti2020} and \cite{Mavroeidis2019}, and is based
on \cite{ChenCurdiaFerrero2012}. The economy consists of two types of
households. A fraction $\omega_{r}$ of type `r' households can only trade
long-term government bonds. The remaining $1-\omega_{r}$ households of type
`u' can purchase both short-term and long-term government bonds, the latter
subject to a trading cost $\zeta_{t}$. This trading cost gives rise to a term
premium, i.e., a spread between long-term and short-term yields, that the
central bank can manipulate by purchasing long-term bonds. The term premium
affects aggregate demand through the consumption decisions of constrained
households. This generates an UMP channel.

Households choose consumption to maximize an isoelastic utility function and firms set prices
subject to Calvo frictions. These give rise to an Euler equation for output
and a Phillips curve, respectively. Equation (\ref{eq: NK EE UMP})
can be derived from these Euler equations and an assumption about the policy rule for long-term asset purchases. For simplicity, we omit the AD shock $\epsilon_t$ from this derivation, as it is straightforward to add. 

Up to a loglinear approximation, the relevant first-order conditions of the
households' optimization problem can be written as%
\begin{align}
0  &  =E_{t}\left[  -\frac{1}{\sigma}\left(  \hat{c}_{t+1}^{u}-\hat{c}_{t}%
^{u}\right)  +\hat{R}_{t}-\hat{\pi}_{t+1}\right]  ,\label{l_HHu2}\\
\frac{\zeta}{1+\zeta}\hat{\zeta}_{t}  &  =E_{t}\left[  -\frac{1}{\sigma
}\left(  \hat{c}_{t+1}^{u}-\hat{c}_{t}^{u}\right)  +\hat{R}_{L,t+1}-\hat{\pi}
_{t+1}\right]  ,\label{l_HHu3}\\
0  &  =E_{t}\left[  -\frac{1}{\sigma}\left(  \hat{c}_{t+1}^{r}-\hat{c}_{t}%
^{r}\right)  +\hat{R}_{L,t+1}-\hat{\pi}_{t+1}\right]  , \label{l_HHr2}%
\end{align}
where $\sigma$ is the elasticity of intertemporal substitution, $\zeta$ is the
steady state value of $\zeta_{t},$ hatted variables denote log-deviations from
steady state, $c_{t}^{j}$ is consumption of household $j\in\left\{
u,r\right\}  ,$ $R_{t}$ is the short-term nominal interest rate, and $R_{L,t}$
is the gross yield on long-term government bonds from period $t-1$ to
$t$. Goods market clearing yields
\begin{equation}
\hat{x}_{t}=\omega_{r}\hat{c}_{t}^{r}+\left(  1-\omega_{r}\right)  \hat{c}%
_{t}^{u}, \label{eq: good mkt clearing}%
\end{equation}
where $x_{t}$ is output, and we have assumed, for simplicity, that in steady
state $c^{u}=c^{r}$, which implies $c^{u}=c^{r}=x.$ Multiplying (\ref{l_HHu2})
and (\ref{l_HHr2}) by $\left(  1-\omega_{r}\right)  $ and $\omega_{r},$
respectively, and adding them yields%
\begin{equation}
\hat{x}_{t}=E_{t}\hat{x}_{t+1}-\sigma E_{t}\left[  \left(  1-\omega
_{r}\right)  \hat{R}_{t}+\omega_{r}\hat{R}_{L,t+1}-\hat{\pi}_{t+1}\right]  .
\label{eq: EE with RL}%
\end{equation}
Subtracting (\ref{l_HHu2}) from (\ref{l_HHu3}) yields%
\begin{equation}
E_{t}\left(  \hat{R}_{L,t+1}\right)  =\hat{R}_{t}+\frac{\zeta}{1+\zeta}%
\hat{\zeta}_{t}, \label{eq: yield spread}%
\end{equation}
which establishes that the term premium between long and short yields is
proportional to $\hat{\zeta}_{t}.$ Substituting for $E_{t}\left(  \hat
{R}_{L,t+1}\right)  $ in (\ref{eq: EE with RL}) using (\ref{eq: yield spread})
yields%
\begin{equation}
\hat{x}_{t}=E_{t}\hat{x}_{t+1}-\sigma\left(  \hat{R}_{t}+\omega_{r}\frac
{\zeta}{1+\zeta}\hat{\zeta}_{t}\right)  +\sigma E_{t}\left(  \hat{\pi}_{t+1}\right)
\label{eq: EE with zeta}%
\end{equation}
Next, assume that the cost of trading long-term bonds depends on their supply,
$b_{L,t},$ i.e.,%
\[
\hat{\zeta}_{t}=\rho_{\zeta}\hat{b}_{L,t},\quad\rho_{\zeta}\geq0.
\]
Substituting for $\hat{\zeta}_{t}$ in (\ref{eq: EE with zeta}) yields the
Euler equation%
\begin{equation}
\hat{x}_{t}=E_{t}\hat{x}_{t+1}-\sigma\left(  \hat{R}_{t}+\omega_{r}\frac
{\zeta}{1+\zeta}\rho_{\zeta}\hat{b}_{L,t}\right)  +\sigma E_{t}\left(
\hat{\pi}_{t+1}\right)  . \label{eq: EE UMP}%
\end{equation}
Suppose that UMP follows
the policy rule
\begin{equation}\label{eq: UMP rule}
\hat{b}_{L,t}=\alpha \min\left\{\hat{R}_{t}^{\ast}+\mu,0\right\},
\end{equation}
where $\hat{R}_{t}^{\ast}$ is the shadow rate prescribed by the Taylor rule (\ref{eq: NK TR UMP}), and $\alpha>0$ is a factor of proportionality that can be interpreted as varying the intensity of UMP -- a bigger $\alpha$ corresponds to a larger intervention for any given deviation of inflation and output from target. Substituting for $\hat{b}_{L,t}$ in (\ref{eq: EE UMP}) using (\ref{eq: UMP rule}), and using the fact that $\min\left\{\hat{R}_{t}^{\ast}+\mu,0\right\} =\allowbreak \hat{R}_{t}^{\ast} - \max\left\{\hat{R}_{t}^{\ast},-\mu\right\} =\allowbreak \hat{R}_{t}^{\ast} - \hat{R}_t$ yields (\ref{eq: NK EE UMP}) with $\xi:=\alpha\omega_{r}\frac{\zeta}{1+\zeta}\rho_{\zeta}$.
\end{proof}

% \subsection{Proof of Proposition \protect\ref{prop: NK cutoff UMP}}
\begin{proof}[Proof of Proposition \protect\ref{prop: NK cutoff UMP}]

The proof can follow the same steps as the proof of Proposition \ref{prop:
NK cutoff}, but because of the absorbing state assumption, it is easier to
proceed graphically. First, we look at the absorbing (or steady) state. The $%
AS$ curve is the same as (\ref{eq: AS/AD abs NK}), but the $AD$ curve is
different: 
\begin{equation*}
\hat{\pi}=\frac{\lambda }{1-\beta }\hat{x}\quad AS\quad ;
\qquad \hat{\pi}=\left( 1-\xi \right) \max \left\{ -\mu ,\psi \hat{\pi}\right\} +\xi \psi \hat{\pi}\qquad AD
\end{equation*}%
If the AS curve is everywhere steeper or everywhere flatter than the AD\
curve, then there will always be a unique steady state for any value of $\mu
.$ This holds if and only if: 
\begin{equation*}
\xi \psi >1,\quad \text{and \quad }\psi >1,\quad \text{OR\quad }\xi \psi
<1,\quad \text{and \quad }\psi <1.
\end{equation*}%
The steady state is a PIR, and it is given by $\hat{\pi}=\hat{x}=\hat{R}=0$
(because the value of the shock is zero at the absorbing state).

Suppose that in the transitory state $\epsilon _{t}=-\sigma \hat{M}_{t+1|t}=\sigma
pr^{L}<0$, for comparability with the standard NK model (this does not matter for the argument, since we only need to look at the slope of the AD curve). The MSV solution, if it exists, will be constant ($\hat{\pi}^{L},\hat{x}^{L}$)
and with probability $(1-p)$ we are back in the absorbing state. The $AS$
curve is given by (\ref{app: eq: AS-absPIR}), but the $AD$ curve (\ref{app:
eq: AD-absPIR}) now becomes

\begin{equation}
\hat{\pi}^{L}=\left\{ 
\begin{array}{c}
\frac{1-p}{\sigma\left(  p-\psi\right)  }\hat{x}^{L}-\frac{ pr^{L}}{\left(  p-\psi\right)  }\qquad \hspace{2em} AD^{TR} \hspace{0.5em}
\text{ for }\pi >-\frac{\mu }{\psi } \\ 
\frac{1-p}{\sigma \left( p- \xi\psi \right)}\hat{x}^{L}  - \frac{\left( 1-\xi \right)\mu+pr^{L}}{\left( p-\xi \psi \right)} \qquad AD^{ZLB}\text{ for }\pi \leq -\frac{\mu }{\psi }.
\end{array}
\right.   \label{app: eq: AD-UMP}
\end{equation}

% \begin{equation}
% \hat{x}^{L}=\left\{ 
% \begin{array}{c}
% \sigma \frac{\left( p-\psi \right) \hat{\pi}^{L}+pr^{L}}{1-p}\qquad AD^{TR}%
% \text{ for }\pi >-\frac{\mu }{\psi } \\ 
% \sigma \frac{\left( p-\xi \psi \right) \hat{\pi}^{L}+\left( 1-\xi \right)
% \mu +pr^{L}}{1-p}\qquad AD^{ZLB}\text{ for }\pi \leq -\frac{\mu }{\psi }.%
% \end{array}%
% \right.   \label{app: eq: AD-UMP}
% \end{equation}%
Again, coherency requires that $AD^{TR}$ and $AD^{ZLB}$ be either both
flatter or both steeper than $AS.$ For $AD^{TR},AD^{ZLB}$ both to be flatter
than $AS$ we need 
\begin{equation*}
\psi <p-\frac{\left( 1-p\right) \left( 1-\beta p\right) }{\sigma \lambda }%
=\psi _{p,1,\beta ,\sigma \lambda },\text{ \ and \ }\xi \psi <\psi
_{p,1,\beta ,\sigma \lambda }.
\end{equation*}%
Alternatively, $AD^{TR}$ and $AD^{ZLB}$ must be both steeper than $AS,$
which requires 
\begin{equation*}
\psi >\psi _{p,1,\beta ,\sigma \lambda }\quad \text{and\quad }\xi \psi >\psi
_{p,1,\beta ,\sigma \lambda }.
\end{equation*}

Combining with the inequalities in the absorbing state, and using the fact
that $\psi _{p,1,\beta ,\sigma \lambda }\leq 0$ and $\xi >0,$ we obtain (\ref%
{eq: CC-NK-2s-UMP}).
\end{proof}

\subsection{Derivation of results in Subsection \ref{s: endog}} \label{app: s: inertia}

\subsubsection{Coefficients in \nameref{ex: NK inert}} \label{app: s: coeffNKITR}

The coefficients in the canonical representation of the model are:
\begin{align*}
A_{0}  &  =%
\begin{pmatrix}
1 & -\lambda & 0\\
0 & 1 & 0\\
0 & 0 & 1
\end{pmatrix}
,\quad A_{1}=%
\begin{pmatrix}
1 & -\lambda & 0\\
0 & 1 & \sigma\\
-\psi & -\psi_{x} & 1
\end{pmatrix}
,\quad B_{0}=B_{1}=%
\begin{pmatrix}
-\beta & 0 & 0\\
-\sigma & -1 & 0\\
0 & 0 & 0
\end{pmatrix}
,\\
C_{0}  &  =%
\begin{pmatrix}
-1 & 0 & 0 & 0\\
0 & -1 & 0 & -\sigma\mu\\
0 & 0 & 0 & -\mu
\end{pmatrix}
,\quad C_{1}=%
\begin{pmatrix}
-1 & 0 & 0 & 0\\
0 & -1 & 0 & 0\\
0 & 0 & -1 & 0
\end{pmatrix}
,\quad D_{0}=D_{1}=0_{3\times4},
\end{align*}
$H_{0}=0_{3\times3}$, $H_{1}=-\phi a a'$, $a=\left(  0,0,1\right)
^{\prime}$, $b=0_{3\times1}$, $c=\left(  0,0,0,\mu\right)  ^{\prime}$,
$d=0_{4\times1}$ and $h=0_{3\times1}$. %Unfortunately, this model does not satisfy Assumption \ref{ass: ST}.

\subsubsection{Brute force method for checking coherency}\label{app: s: bruteforcealg}

To derive (\ref{eq: equations for G and Z}), first note that $E(Y_{t+1}|Y_{t}\allowbreak=\mathbf{Y}_{t}%
 e_{i},X_{t}\allowbreak=\mathbf{X}e_{i})\allowbreak=\left( \mathbf{G}g'\mathbf{Y}_{t}e_i+\mathbf{Z}\right)
 K^{\prime}e_{i}$, because the support of $Y_{t+1}$ conditional on $Y_t=\mathbf{Y}_{t} e_{i}$ is $\mathbf{G}g'\mathbf{Y}_{t}e_i+\mathbf{Z}$, recalling the definition $y_t:=g'Y_t$. Substituting this and $Y_t = (\mathbf{G}y_{t-1}+\mathbf{Z})e_i$ into (\ref{eq: canon endog}) yields (\ref{eq: equations for G and Z}). 
 
We can solve the model backwards from some date $T$ at which it is known that
$\mathbf{Y}_{T}=\mathbf{G}_{J_{0}}y_{T-1}+\mathbf{Z}_{J_{0}}$, where $J_{0}%
\in\mathcal{J}$ denotes the regime configuration across the exogenous states
at $T$, and the set $\mathcal{J}$ has $2^{k}$ elements. We will treat
$\mathbf{G}_{J_{0}},$ $\mathbf{Z}_{J_{0}}$ as known for the ensuing
discussion. For example, if $J_{0}$ is PIR-only, i.e., the constraint never
binds, $\mathbf{G}_{J_{0}},$ $\mathbf{Z}_{J_{0}}$ can be obtained using the
\cite{Blan80} method. More generally, $\mathbf{G}_{J_{0}},$ $\mathbf{Z}_{J_{0}}$ can be solved from the identities implied by (\ref{eq: equations for G and Z}), i.e.,
\begin{align}
0 &=A_{s_{t,i}}\mathbf{G}e_{i}+h_{s_{t,i}}+B\mathbf{G}K^{\prime}e_{i}g^{\prime
}\mathbf{G}e_{i},\quad \text{\ and}\label{eq: equation for G}\\
0 &=\left(  A_{s_{t,i}}\mathbf{Z}+B_{s_{t,i}}\mathbf{G}K^{\prime}e_{i}g^{\prime
}\mathbf{Z}+B_{s_{t,i}}\mathbf{Z}K^{\prime}+C_{s_{t,i}}\mathbf{X}+D_{s_{t,i}%
}\mathbf{X}K^{\prime}\right)  e_{i},\label{eq: equation for Z}%
\end{align}
for all $i=1,\ldots,k.$

Given $\mathbf{Y}_{T}=\mathbf{G}_{J_{0}}y_{T-1}+\mathbf{Z}_{J_{0}},$ we solve
for $\mathbf{Y}_{T-1}$ as a function of $y_{T-2}$ from%
\begin{align*}
0 &  =\left(  A_{s_{T-1,i}}\mathbf{Y}_{T-1}+B_{s_{T-1,i}}\mathbf{Y}_{T}%
^{i}K^{\prime}+C_{s_{T-1,i}}\mathbf{X}+D_{s_{T-1,i}}\mathbf{X}K^{\prime
}\right)  e_{i}+h_{s_{T-1,i}}y_{T-2}\\
&  =\left(  A_{s_{T-1,i}}\mathbf{Y}_{T-1}+B_{s_{T-1,i}}\left(  \mathbf{G}%
_{J_{0}}g^{\prime}\mathbf{Y}_{T-1}e_{i}+\mathbf{Z}_{J_{0}}\right)  K^{\prime
}+C_{s_{T-1,i}}\mathbf{X}+D_{s_{T-1,i}}\mathbf{X}K^{\prime}\right)
e_{i}\\ & +h_{s_{T-1,i}}y_{T-2}\\
&  =\left(  A_{s_{T-1,i}}+B_{s_{T-1,i}}\mathbf{G}_{J_{0}}K^{\prime}%
e_{i}g^{\prime}\right)  \mathbf{Y}_{T-1}e_{i}\\
&  +\left(  B_{s_{T-1,i}}\mathbf{Z}_{J_{0}}K^{\prime}+C_{s_{T-1,i}}%
\mathbf{X}+D_{s_{T-1,i}}\mathbf{X}K^{\prime}\right)  e_{i}+h_{s_{T-1,i}%
}y_{T-2}.
\end{align*}
Since we can now treat $\mathbf{G}_{J_{0}},\mathbf{Z}_{J_{0}}$ as fixed for
solving backwards, given $J_{0}\in\mathcal{J}$, the CC condition is that all
of the $2^{k}$ determinants%
\begin{equation}
\det\mathcal{A}_{J_{0}J_{1}}=\prod_{i=1}^{k}\det\left(  A_{s_{T-1,i}%
}+B_{s_{T-1,i}}\mathbf{G}_{J_{0}}K^{\prime}e_{i}g^{\prime}\right)  ,\quad
J_{1}\in\mathcal{J}\label{eq: detAJ0J1}%
\end{equation}
should have the same sign:%
\begin{equation}
\det\mathcal{A}_{J_{0}J_{1}}\text{ has the same sign }\forall J_{1}%
\in\mathcal{J}.\label{eq: CC at T-1}%
\end{equation}
For example, if $k=2,$ then the determinants can be written as%
\[%
\begin{tabular}
[c]{lll}%
$\det\mathcal{A}_{J_{0}\left\{  1,2\right\}  }$ & $=\det\left(  A_{1}%
+B_{1}\mathbf{G}_{J_{0}}K^{\prime}e_{1}g^{\prime}\right)  \det\left(
A_{1}+B_{1}\mathbf{G}_{J_{0}}K^{\prime}e_{2}g^{\prime}\right)  $ & (P,P)\\
$\det\mathcal{A}_{J_{0}\left\{  2\right\}  }$ & $=\det\left(  A_{0}%
+B_{0}\mathbf{G}_{J_{0}}K^{\prime}e_{1}g^{\prime}\right)  \det\left(
A_{1}+B_{1}\mathbf{G}_{J_{0}}K^{\prime}e_{2}g^{\prime}\right)  $ & (Z,P)\\
$\det\mathcal{A}_{J_{0}\left\{  1\right\}  }$ & $=\det\left(  A_{1}%
+B_{1}\mathbf{G}_{J_{0}}K^{\prime}e_{1}g^{\prime}\right)  \det\left(
A_{0}+B_{0}\mathbf{G}_{J_{0}}K^{\prime}e_{2}g^{\prime}\right)  $ & (P,Z)\\
$\det\mathcal{A}_{J_{0}\varnothing}$ & $=\det\left(  A_{0}+B_{0}%
\mathbf{G}_{J_{0}}K^{\prime}e_{1}g^{\prime}\right)  \det\left(  A_{0}%
+B_{0}\mathbf{G}_{J_{0}}K^{\prime}e_{2}g^{\prime}\right)  $ & (Z,Z).
\end{tabular}
\ \
\]
If the CC condition (\ref{eq: CC at T-1}) is violated, we need support
restrictions. Otherwise, the solution will be given by%
\begin{multline}\label{eq: sol Y at T-1}
\mathbf{Y}_{T-1}e_{i}=-\left(  A_{s_{T-1,i}}+B_{s_{T-1,i}}\mathbf{G}_{J_{0}%
}K^{\prime}e_{i}g^{\prime}\right)  ^{-1}\\
\left[  \left(  B_{s_{T-1,i}}\mathbf{Z}_{J_{0}}K^{\prime}+C_{s_{T-1,i}%
}\mathbf{X}+D_{s_{T-1,i}}\mathbf{X}K^{\prime}\right)  e_{i}+h_{s_{T-1,i}%
}y_{T-2}\right]
\end{multline}%
for all $i=1,\ldots,k$, depending on which of the above satisfies the
inequality implied by the regime configuration $J_{1}$. Collecting all the
states, the solutions (\ref{eq: sol Y at T-1}) can be written as
$\mathbf{Y}_{T-1}=\mathbf{G}_{J_{0}J_{1}}y_{T-2}+\mathbf{Z}_{J_{0}J_{1}},$
with
\begin{equation}%
\begin{tabular}
[c]{l}%
$\mathbf{G}_{J_{0}J_{1},i}:=-\left(  A_{s_{T-1,i}}+B_{s_{T-1,i}}%
\mathbf{G}_{J_{0}}K^{\prime}e_{i}g^{\prime}\right)  ^{-1}h_{s_{T-1,i}}%
,\quad\text{and}$\\
$\mathbf{Z}_{J_{0}J_{1},i}:=-\left(  A_{s_{T-1,i}}+B_{s_{T-1,i}}%
\mathbf{G}_{J_{0}}K^{\prime}e_{i}g^{\prime}\right)  ^{-1}$\\ $\qquad\qquad\qquad \left(  B_{s_{T-1,i}%
}\mathbf{Z}_{J_{0}}K^{\prime}+C_{s_{T-1,i}}\mathbf{X}+D_{s_{T-1,i}}%
\mathbf{X}K^{\prime}\right)  e_{i},$%
\end{tabular}
\ \ \ \ \label{eq: sol G,Z at T-1}%
\end{equation}
for all $i=1,\ldots,k$. Note that the double subscript in $\mathbf{G}%
_{J_{0}J_{1}}$ and $\mathbf{Z}_{J_{0}J_{1}}$ shows that there will
be $2^{k}$ different solutions $J_{1}\in\mathcal{J}$ at $T-1$ corresponding to
each regime configuration $J_{0}\in\mathcal{J}$ at $T$. So, there will be
$2^{2k}$ different cases.

Substituting backwards to any date $t<T$, it is clear that the CC condition
would be%
\[
\det\mathcal{A}_{J_{0}\ldots J_{T-t}}\text{ has the same sign }\forall
J_{T-t}\in\mathcal{J},
\]
where%
\[
\det\mathcal{A}_{J_{0}\ldots J_{T-t}}=\prod_{i=1}^{k}\det\left(  A_{s_{t,i}%
}+B_{s_{t,i}}\mathbf{G}_{J_{0}\ldots J_{T-t-1}}K^{\prime}e_{i}q^{\prime
}\right)  ,\quad J_{T-t}\in\mathcal{J}%
\]
and the solution will be given by $\mathbf{Y}_{t}=\mathbf{G}_{J_{0}\ldots J_{T-t}}y_{t-1}+\mathbf{Z}%
_{_{J_{0}\ldots J_{T-t}}},
$
where $\mathbf{G}_{J_{0}\ldots J_{T-t}},$ $\mathbf{Z}_{_{J_{0}\ldots J_{T-t}}%
}$ are computed recursively by%
\begin{equation}%
\begin{tabular}
[c]{l}%
$\mathbf{G}_{J_{0}\ldots J_{T-t},i}:=-\left(  A_{s_{t,i}}+B_{s_{t,i}%
}\mathbf{G}_{J_{0}\ldots J_{T-t-1}}K^{\prime}e_{i}g^{\prime}\right)
^{-1}h_{s_{t,i}},\quad\text{and}$\\
$\mathbf{Z}_{J_{0}\ldots J_{T-t},i}:=-\left(  A_{s_{t,i}}+B_{s_{t,i}%
}\mathbf{G}_{J_{0}\ldots J_{T-t-1}}K^{\prime}e_{i}g^{\prime}\right)
^{-1}$\\ $\qquad\qquad\qquad\quad\left(  B_{s_{t,i}}\mathbf{Z}_{J_{0}\ldots J_{T-t-1}}K^{\prime
}+C_{s_{t,i}}\mathbf{X}+D_{s_{t,i}}\mathbf{X}K^{\prime}\right)  e_{i}.$%
\end{tabular}
\ \ \ \label{eq: sol G,Z at t}%
\end{equation}

At the end of this recursion at $t=1$ we will have $2^{\left(  T-1\right)  k}$
paths. The initial condition $y_{0}$ will then pick the path(s) that satisfy
the inequalities at all $t$. If the CC condition is satisfied at all $t$, then
there will be a unique solution path for that particular $y_{0}$. Otherwise,
there may be 0 (incoherency) or multiple (incompleteness) solutions.

This suggests the following algorithm for checking the coherency of the model.

\begin{algorithm}
[Coherency in model with endogenous states]Set a date $T>1.$

\begin{enumerate}
\item For each possible regime configuration $J_{0}\in\mathcal{J}$ ($2^{k}$ elements):

\begin{enumerate}
\item Solve (\ref{eq: equation for G}) and (\ref{eq: equation for Z}) to
obtain $\mathbf{G}_{J_{0}}$ and $\mathbf{Z}_{J_{0}}$.

\item For each $J_{1}\in\mathcal{J}$ ($2^{k}$ elements):

\begin{enumerate}
\item Compute $\det\mathcal{A}_{J_{0}J_{1}}$ from (\ref{eq: detAJ0J1}).

\item If $sign\left(  \det\mathcal{A}_{J_{0}J_{1}}\right)  $ is different from
previous $J_{1}$, break the loop and go to next $J_{0}$.

\item Otherwise compute $\mathbf{G}_{J_{0}J_{1}}$ and $\mathbf{Z}_{J_{0}J_{1}%
}$ using (\ref{eq: sol G,Z at T-1})

\item Continue with a list of nested loops for each $J_{T-t},$ for $t=T-2$
till $t=1$.
\end{enumerate}
\end{enumerate}

\item If there is no $J_{0}\in\mathcal{J}$ for which you reach $t=1,$ conclude
that there is no equilibrium without support restrictions.

\item Otherwise, there will be a unique solution. The solution along any
sequence $i_{t},$ $t=1,\ldots,T$ of exogenous shocks can be determined as follows:

\begin{enumerate}
\item Pick a $\hat{J}_{0},\ldots,\hat{J}_{T-2}\in\mathcal{J}^{T-1}$.

\item Find the (unique) $\hat{J}_{T-1}\in\mathcal{J}$ that ensures
$\mathbf{G}_{J_{0}\ldots\hat{J}_{T-1}}y_{0}+\mathbf{Z}_{J_{0}\ldots\hat
{J}_{T-1}}$ satisfies the inequalities determined by regime $\hat{J}_{T-1}$.

\item For $t=2$ to $T-1$,

\begin{enumerate}
\item Compute $y_{t-1}=g^{\prime}\left(  \mathbf{G}_{\hat{J}_{0}\ldots\hat
{J}_{T-t+1}}y_{t-2}+\mathbf{Z}_{\hat{J}_{0}\ldots\hat{J}_{T-t+1}}\right)
e_{i_{t}}$.

\item If $\mathbf{G}_{\hat{J}_{0}\ldots\hat{J}_{T-t}}y_{0}+\mathbf{Z}_{\hat
{J}_{0}\ldots\hat{J}_{T-t}}$ satisfies the inequalities determined by regime
$\hat{J}_{T-t}$, you have found the unique solution with regime configuration
$\hat{J}_{0},\ldots,\hat{J}_{T-1}$. 

\item Otherwise, exit the loop, go back to 3.(a) and pick the next element in
$\mathcal{J}^{T-2}$.
\end{enumerate}
\end{enumerate}
\end{enumerate}
\end{algorithm}

% An illustration of the above algorithm at work is given for a NK model with a
% forward-guidance type policy rule of \cite{DebortoliGaliGambetti2018}. The model is given by equations
% (\ref{eq: NK NKPC}) with $u_{t}=0$, (\ref{eq: NK EE}), and $\hat{R}_{t}%
% =\max\left\{  -\mu,\hat{R}_{t}^{\ast}\right\}  ,$ with $\hat{R}_{t}^{\ast
% }=\rho_{R}\hat{R}_{t-1}^{\ast}+\left(  1-\rho_{R}\right)  \psi\hat{\pi}_{t}$.
% When $\rho_{R}=0,$ this model reduces to the one studied in Proposition
% \ref{prop: NK-TR CC}. 

% [What should I show here? Evaluate numerically the support
% restriction as $\rho_{R}$ increases for the second MR
% calibration in Table 1 (confidence driven)? Or plot a solution path for a particular $\tau,$ i.e.,
% $i_{t}=1$ (transitory state) for $t<\tau,$ and $i_{t}=2$ (absorbing state) for
% $t\geq\tau$? What would that achieve?\ Perhaps show multiple solution paths
% for different initial conditions $\hat{R}_{0}^{\ast}$, e.g., $\hat{R}%
% _{0}^{\ast}=0$ (PIR steady state) and $\hat{R}_{0}^{\ast}=-\psi\mu$ (ZIR
% steady state); or for the same $\hat{R}_{0}^{\ast}$ but different regime
% configurations at $T,$ such as ZIR,PIR, or PIR,PIR?]

\subsubsection{Derivation of the analytical results in \nameref{ex: NK inert}} \label{app: s: NKITR_analytics}

We proceed as for the proof of Proposition \ref{prop: NK-TR sup res}, but now the Taylor rule is given by $\hat{R}_{t} =\max\{-\mu,\allowbreak\phi\hat{R}_{t-1}+\psi\hat{\pi}_{t}\}$. 
%\label{TR_ACS}%
First look at the steady state, where $\epsilon _{t}=0.$  Then, we need to solve a system as (\ref{eq: AS/AD abs NK}), where the only difference now is the $AD^{TR}$ equation given by $\hat{\pi} = (\phi +\psi)\frac{\lambda}{1-\beta}\hat{x}.$
The graphical representation would be the same as Figure \ref{app: fig: nk_abs}, but with a steeper $AD^{TR}$. Whenever $ (\phi +\psi) >1$, the necessary support restriction for existence of a solution is $\mu\geq 0,$ i.e., $\left( r\pi _{\ast }\right) ^{-1}\leq 1.$ When this holds, there are two possible solutions: 1) PIR: $(\hat{\pi},\hat{x},\hat{R})=(0,0,0)$; and 2) ZIR: $(\hat{\pi},\hat{x},\hat{R})=(-\mu ,-\mu \frac{(1-\beta )}{\lambda },-\mu )$. 
However, in this case, the absorbing state admits endogenous dynamics because of the presence of the endogenous state variable $\hat{R}.$ Outside the two steady states then the economy will travel along a stable trajectory that leads to one of the 2 steady states. Let's see under which condition the following solution exists: (i) when the shock disappears the economy will converge to the PIR along the stable manifold; (ii) the solution is MSV in the sense that it depends just on state variables; (iii) in the transitory state where $\epsilon_{t}
=-\sigma\hat{M}_{t+1\mid t}=\sigma pr^{L}<0,$ the economy will be in a ZIR. Under these assumptions, once the shock disappears then we must be on the unique stable manifold that leads to the PIR, i.e., the `intended steady state'.
%This simplifies things because it allows us to pin down where the system will end up eventually, i.e., PIR in the absorbing state.
Assumption (i) hence is key because it pins down the expectations in the absorbing state. This is similar to the proof of Proposition \ref{prop: NK-TR sup res}. However, rather than jump to the intended steady state as when the model is forward-looking, we will arrive there inertially along the unique stable manifold.
To find the MSV solution of the PIR system, we use undetermined coefficients and assume a solution of this form:
\begin{equation}
    \hat{\pi}_{t}=\gamma_{\pi}\hat{R}_{t-1};\qquad \hat{x}_{t}=\gamma
_{x}\hat{R}_{t-1};\qquad \hat{R}_{t}=\gamma_{R}\hat{R}_{t-1}.
\label{app: eq: UCPIR}
\end{equation}
Substituting in the \nameref{ex: NK inert} system yields the following cubic equation in $\gamma_{R}$
\begin{equation}
    \beta\gamma_{R}^{3}+\gamma_{R}^{2}\left(  \psi\sigma-1-\beta-\beta\phi
-\lambda\sigma\right)  +\gamma_{R}\left(  1+\phi+\beta\phi+\lambda\sigma
\phi\right)  -\phi=0
\label{app: eq: cubicPIR}
\end{equation}
Let us now assume that it exist a unique solution within the unit circle, i.e., $\left\vert \gamma_{R}\right\vert <1$, as it would be in most applications.\footnote{Just as an example, using: $\psi=1.5,\sigma =1,\beta=0.99,\phi=0.8,\lambda=0.02,$ then the unique stable solution would be $\gamma_{R}=0.35206.$} Then, the dynamics along the stable trajectory is given by the recursion:
\begin{equation}
    \hat{\pi}_{t+j}=\gamma_{\pi}\gamma_{R}^{j}\hat{R}_{t-1};\text{ \ \ \ }\hat
{x}_{t+j}=\gamma_{x}\gamma_{R}^{j}\hat{R}_{t-1};\text{ \ \ \ }\hat{R}%
_{t+j}=\gamma_{R}^{j+1}\hat{R}_{t-1}.
\end{equation}
Note that if $\hat{R}_{t-1}=-\mu,$ then simply
\[
\hat{\pi}_{t}=-\gamma_{\pi}\mu;\qquad \hat{x}_{t}=-\gamma_{x}\mu;\qquad \hat{R}_{t}=-\gamma_{R}\mu,
\]
hence, the system will never be in a ZIR when the shock vanishes, because $\hat{R}_{t}=-\gamma_{R}\mu > - \mu$ if $\left\vert \gamma_{R}\right\vert <1$.

Next, turn to the transitory state. Here, we just want to study the situation in which the system is ZIR in the transitory state. In this case the system becomes
\begin{equation*}
    \hat{\pi}_{t} =\beta\hat{\pi}_{t+1|t}+\lambda\hat{x}_{t}, \quad
\hat{x}_{t}   =\hat{x}_{t+1|t}-\sigma\left(  -\mu-\hat{\pi}
_{t+1|t}\right)  +\sigma pr^{L}.
\end{equation*}
Note that this system is completely forward looking and not inertial because it does not have endogenous state variables, since by assumption $\hat{R}_{t-1}=-\mu$. Hence we can follow the same steps we did for Proposition \ref{prop: NK-TR sup res}, because the MSV solution, if it exists, will be constant $ \left (\hat{\pi}^{L}_{t},\hat{x}^{L}_{t} \right) $ and with probability $(1-p)$ we are back on the manifold of the PIR absorbing state.
The expectations thus are: $\hat{\pi}_{t+1|t}=p\hat{\pi}^{L}+(1-p)\left(  \gamma_{\pi}(-\mu)\right)  ; \quad
\hat{x}_{t+1|t}=p\hat{x}^{L}+(1-p)\left( \gamma_{x}(-\mu)\right).$ 
Substitute into the above ZIR system to get

\begin{equation}
\hat{\pi}^{L}=\frac{\lambda}{1-\beta p}\hat{x}^{L}-\frac{\beta(1-p)\gamma
_{\pi}}{1-\beta p}\mu \qquad AS
\end{equation}
\begin{equation}
\hat{\pi}^{L}=\frac{1-p}{\sigma p}\hat{x}^{L}-\frac{\mu}{p}\left(
1-\frac{(1-p)\left(  \gamma_{x}+\sigma\gamma_{\pi}\right)  }{\sigma}\right)
-r^{L}\qquad AD^{ZLB}%
\end{equation}

Note that if $\gamma_{x}=\gamma_{\pi}=0,$ we are back to equation (\ref{app: eq: AS-absPIR}) and (\ref{app: eq: AD-absPIR}) and figure \ref{app: fig: nk_psilargerthanp}. In this case, a graph would be very similar, since the intercepts are different, but the slopes are not affected. The same reasoning therefore applies. For the solution to hold it must be that: $-\mu>\phi\left(  -\mu\right)
+\psi\hat{\pi}=>\hat{\pi}\leq-\frac{\mu\left(  1-\phi\right)  }{\psi}.$ To find the cutoff equates the two equations when $\hat{\pi}=-\frac{\mu\left(  1-\phi\right)  }{\psi},$ to get
\begin{equation}
    -\bar{r}^{L} =\mu\left(  \frac{\psi
-p}{\psi p}+\frac{\theta}{\psi}+\frac{\phi}{\psi}\left(  1-\theta\right)  -(1-p)\frac
{\lambda\gamma_{x}+\gamma_{\pi}\left[  \beta(1-p)+\lambda\sigma\right]
}{\lambda\sigma p}\right),
\end{equation}
which is (\ref{eq: rl NK-ITR}).

\subsubsection{Quasi differencing derivations} \label{app: s: simutri}

Premultiplying (\ref{eq: canon}) by
the $n_{2}\times n$ matrix $\left(  Q^{-1}\right)  _{22}^{-1}\left(
Q^{-1}\right)  _{2\cdot},$ where $\left(  Q^{-1}\right)  _{22}$ is the bottom
right $n_{2}\times n_{2}$\ submatrix of $Q^{-1}$ and $\left(  Q^{-1}\right)
_{2\cdot}$ consists of the bottom $n_{2}$ rows of $Q^{-1},$ we get
\begin{align}
0  & =\widetilde{A}_{s}\widetilde{Y}_{t}+\widetilde{Y}_{t+1|t}+\widetilde
{C}_{s}X_{t}+\widetilde{D}_{s}X_{t+1|t}\label{eq: canon reduced}\\
s  & =1_{\left\{  \widetilde{a}^{\prime}\widetilde{Y}_{t}+\widetilde
{b}^{\prime}\widetilde{Y}_{t+1|t}+c^{\prime}X_{t}+d^{\prime}X_{t+1|t}%
>0\right\}  },\nonumber
\end{align}
where $\widetilde{Y}_{t}=\left(  Q^{-1}\right)  _{22}^{-1}\left(
Q^{-1}\right)  _{2\cdot}Y_{t}=Y_{2t}+\left(  Q^{-1}\right)  _{22}^{-1}\left(
Q^{-1}\right)  _{21}Y_{1t},$ $\widetilde{A}_{s}=\left(  Q^{-1}\right)
_{22}^{-1}\allowbreak \Lambda_{s,22}\allowbreak\left(  Q^{-1}\right)  _{22}$, $\widetilde{C}%
_{s}=\left(  Q^{-1}\right)  _{22}^{-1}\left(  Q^{-1}\right)  _{2\cdot}C_{s}$
and $\widetilde{D}_{s}=\left(  Q^{-1}\right)  _{22}^{-1}\left(  Q^{-1}\right)
_{2\cdot}D_{s}$, $\widetilde{a}^{\prime}=a^{\prime}Q_{2}\left(  Q^{-1}\right)
_{22}$ and $\widetilde{b}^{\prime}=b^{\prime}Q_{2}\left(  Q^{-1}\right)
_{22}$.

$\widetilde{Y}_{t}$ and $\widetilde{A}_{s}$ can be derived from:
\begin{align*}
\widetilde{A}_{s}\widetilde{Y}_{t}  &  =\left(  Q^{-1}\right)  _{22}%
^{-1}\left(  Q^{-1}\right)  _{2\cdot}%
\begin{pmatrix}
Q_{1} & Q_{2}%
\end{pmatrix}
\begin{pmatrix}
\Lambda_{s,11} & \Lambda_{s,12}\\
0 & \Lambda_{s,22}%
\end{pmatrix}%
\begin{pmatrix}
\left(  Q^{-1}\right)  _{1\cdot}Y_{t}\\
\left(  Q^{-1}\right)  _{2\cdot}Y_{t}%
\end{pmatrix}
\\
&  =%
\begin{pmatrix}
0 & \left(  Q^{-1}\right)  _{22}^{-1}\Lambda_{s,22}%
\end{pmatrix}%
\begin{pmatrix}
\left(  Q^{-1}\right)  _{1\cdot}Y_{t}\\
\left(  Q^{-1}\right)  _{2\cdot}Y_{t}%
\end{pmatrix}
\\
&  =\left(  Q^{-1}\right)  _{22}^{-1}\Lambda_{s,22}\left(  Q^{-1}\right)
_{2\cdot}Y_{t}=\left[  \left(  Q^{-1}\right)  _{22}^{-1}\Lambda_{s,22}\left(
Q^{-1}\right)  _{22}\right]  \left[  \left(  Q^{-1}\right)  _{22}^{-1}\left(
Q^{-1}\right)  _{2\cdot}Y_{t}\right]  .
\end{align*}
$\widetilde{a}$ (and similarly $\widetilde{b}$) follows from
\begin{align*}
a^{\prime}Y_{t}  &  =a^{\prime}%
\begin{pmatrix}
Q_{1} & Q_{2}%
\end{pmatrix}%
\begin{pmatrix}
\left(  Q^{-1}\right)  _{1\cdot}Y_{t}\\
\left(  Q^{-1}\right)  _{2\cdot}Y_{t}%
\end{pmatrix}
=%
\begin{pmatrix}
0 & a^{\prime}Q_{2}%
\end{pmatrix}%
\begin{pmatrix}
\left(  Q^{-1}\right)  _{1\cdot}Y_{t}\\
\left(  Q^{-1}\right)  _{2\cdot}Y_{t}%
\end{pmatrix}
\\
&  =a^{\prime}Q_{2}\left(  Q^{-1}\right)  _{2\cdot}Y_{t}=a^{\prime}%
Q_{2}\left(  Q^{-1}\right)  _{22}\left[  \left(  Q^{-1}\right)  _{22}%
^{-1}\left(  Q^{-1}\right)  _{2\cdot}Y_{t}\right]  ,
\end{align*}
where $a^{\prime}Q_{1}=0$ follows by Assumption \ref{ass: ST}.

\subsubsection{Proof of claims in \nameref{ex: ACS inert}} \label{app: s: ACS-STR}

The model is:
\begin{align*}
\hat{R}_{t}  &  =\max\left(  -\mu,\phi\hat{R}_{t-1}+\psi\hat{\pi}_{t}\right)
\\
\hat{\pi}_{t+1|t}  &  =\hat{R}_{t}+\hat{M}_{t+1|t}.
\end{align*}
Let $Y_{t}=\left(  \hat{\pi}_{t},\hat{R}_{t-1}\right)  ^{\prime}.$ At a PIR we
have%
\begin{equation}
\underbrace{%
\begin{pmatrix}
\phi & \psi\\
0 & 0
\end{pmatrix}
}_{A_{1}}%
\begin{pmatrix}
\hat{R}_{t-1}\\
\hat{\pi}_{t}%
\end{pmatrix}
+\underbrace{%
\begin{pmatrix}
-1 & 0\\
1 & -1
\end{pmatrix}
}_{B_{1}}%
\begin{pmatrix}
\hat{R}_{t}\\
\hat{\pi}_{t+1|t}%
\end{pmatrix}
+\underbrace{%
\begin{pmatrix}
0 & 0\\
1 & 0
\end{pmatrix}
}_{D_{1}}%
\begin{pmatrix}
\hat{M}_{t+1|t}\\
1
\end{pmatrix}
=0 \label{eq: pir}%
\end{equation}
while at a ZIR we have%
\[
\underbrace{%
\begin{pmatrix}
0 & 0\\
0 & 0
\end{pmatrix}
}_{A_{0}}%
\begin{pmatrix}
\hat{R}_{t-1}\\
\hat{\pi}_{t}%
\end{pmatrix}
+\underbrace{%
\begin{pmatrix}
-1 & 0\\
1 & -1
\end{pmatrix}
}_{B_{0}}%
\begin{pmatrix}
\hat{R}_{t}\\
\hat{\pi}_{t+1|t}%
\end{pmatrix}
+\underbrace{%
\begin{pmatrix}
0 & -\mu\\
1 & 0
\end{pmatrix}
}_{D_{0}}%
\begin{pmatrix}
\hat{M}_{t+1|t}\\
1
\end{pmatrix}
=0.
\]
Since $B_{0}=B_{1}$ is clearly invertible and $A_{0}=0,$ the matrices
$B_{0}^{-1}A_{0}$ and $B_{1}^{-1}A_{1}$ clearly commute, satisfying the first
part of Assumption \ref{ass: ST}. Because $B_{0}^{-1}A_{0}=0,$ we may choose
$Q\Lambda_{1}Q^{-1}$ as the Jordan decomposition of $B_{1}^{-1}A_{1},$ where%
\begin{equation}
Q=%
\begin{pmatrix}
1 & 1\\
-\frac{\phi}{\psi} & 1
\end{pmatrix}
,\quad\Lambda_{1}=%
\begin{pmatrix}
0 & 0\\
0 & -\psi-\phi
\end{pmatrix}
. \label{eq: Q and Lambda1}%
\end{equation}
The occasionally binding constraint is $\phi\hat{R}_{t-1}+\psi\hat{\pi}%
_{t}+\mu>0$, so $a=\left(  \phi,\psi\right)  ^{\prime},$ $b=0,$ $c=\left(
0,\mu\right)  $ in $s_{t}=1_{\left\{  a^{\prime}Y_{t}+b^{\prime}Y_{t+1|t}+c^{\prime}%
X_{t}+d^{\prime}X_{t+1|t}>0\right\}}$. From (\ref{eq: Q and Lambda1}), we see that $Q_{1}=\left(  1,-\frac{\phi}%
{\psi}\right)^{\prime},$ so $a^{\prime}Q_{1}=0$, thus verifying the second
part of Assumption \ref{ass: ST}.

The model can be written in the form (\ref{eq: canon reduced}) with
$\widetilde{Y}_{t}=\hat{\pi}_{t}+\frac{\phi}{\psi}\hat{R}_{t-1},$
$\widetilde{a}=\left(  \psi+\phi\right)  \frac{\psi}{\phi+\psi}=\psi$ and
\begin{align}
-\left(  \psi+\phi\right)  \widetilde{Y}_{t}+\widetilde{Y}_{t+1|t}-\hat
{M}_{t+1|t}  &  =0,\quad\text{if }\psi\widetilde{Y}_{t}>-\mu,\text{
}\label{eq: tilde pir}\\
\widetilde{Y}_{t+1|t}-\hat{M}_{t+1|t}+\mu\frac{\phi+\psi}{\psi}  &
=0,\quad\text{if }\psi\widetilde{Y}_{t}\leq-\mu. \label{eq: tilde zir}%
\end{align}
This is a piecewise linear model. If $\hat{M}_{t}$ follows a $2$-state Markov
Chain, it can be put in GLM\ form (\ref{eq: F}) with%
\begin{equation}%
\begin{array}
[c]{ll}%
\mathcal{A}_{1}=K-\left(  \phi+\psi\right)  I_{2}, & J_{1}=\left\{
1,2\right\}  \text{ (PIR,PIR)}\\
\mathcal{A}_{2}=-\left(  \phi+\psi\right)  e_{2}e_{2}^{\prime}+K,\quad &
J_{2}=\left\{  2\right\}  \text{ (ZIR,PIR)}\\
\mathcal{A}_{3}=-\left(  \phi+\psi\right)  e_{1}e_{1}^{\prime}+K, &
J_{2}=\left\{  1\right\}  \text{ (PIR,ZIR)}\\
\mathcal{A}_{4}=K, & J_{4}=\varnothing\text{ (ZIR,ZIR).}%
\end{array}
\label{eq: F in ACS TTRS}%
\end{equation}
The algebra to analyse its coherency properties is exactly the same as for the
noninertial case with $\phi=0$. Specifically, under Assumption
\ref{ass: M nonlinear absorbing}, the CC condition of the \nameref{th: GLM}
Theorem holds if and only if $\psi+\phi<p,$ which nests the noninertial case
$\psi<p$.\footnote{The latter can be derived from Proposition
\ref{prop: NK cutoff} with $\sigma=\infty$ and $q=1$.} This means that if
$\phi>p,$ then for all $\psi>0$ this model will not be generically coherent,
meaning that we will require support restrictions for existence of an equilibrium.

Finally, it is fairly straightforward to infer the support restriction
$-r^{L}\leq\mu\frac{\psi+\phi-p}{\psi p}$ by following the steps in the proof
of Proposition \ref{prop: NK-TR sup res}, i.e., by solving the model under all
four regime configurations. For brevity, it suffices to give the solutions for
the cases PIR,PIR and ZIR,PIR. For PIR,PIR, we have
\[
\hat{\pi}_{t}=\left\{
\begin{array}
[c]{ll}%
-\frac{\phi}{\psi}\hat{R}_{t-1}+\frac{p}{\psi+\phi-p}r^{L}, & \text{if }%
\hat{M}_{t}=-r^{L}\\
-\frac{\phi}{\psi}\hat{R}_{t-1} & \text{if }\hat{M}_{t}=0,
\end{array}
\right.
\]
which requires the support restrictions%
\[
\psi\hat{\pi}_{t}+\phi\hat{R}_{t-1}=\left\{
\begin{array}
[c]{ll}%
\frac{\psi p}{\psi+\phi-p}r^{L}\geq-\mu, & \text{if }\hat{M}_{t}=-r^{L}\\
0\geq-\mu, & \text{if }\hat{M}_{t}=0,
\end{array}
\right.
\]
i.e.,
\begin{equation}
-r^{L}\leq\mu\frac{\psi+\phi-p}{\psi p}. \label{eq: supp ACS inert}%
\end{equation}
For ZIR,PIR, the solution is%
\[
\hat{\pi}_{t}=\left\{
\begin{array}
[c]{ll}%
\frac{\psi+(1-p)\phi}{\psi p}\hat{R}_{t-1}-r^{L}, & \hat{M}_{t}=-r^{L}\\
-\frac{\phi}{\psi}\hat{R}_{t-1} & \hat{M}_{t}=0.
\end{array}
\right.
\]
which requires the support restrictions%
\[
\psi\hat{\pi}_{t}+\phi\hat{R}_{t-1}=\left\{
\begin{array}
[c]{ll}%
-\frac{\psi+(1-p)\phi}{p}\mu-\psi r^{L}-\phi\mu\leq-\mu, & \text{if }\hat
{M}_{t}=-r^{L}\\
0\geq-\mu, & \text{if }\hat{M}_{t}=0,
\end{array}
\right.
\]
which is also (\ref{eq: supp ACS inert}). \hfill\openbox

\subsection{Derivation of the equilibria in Table \ref{tab: equilibria simple}}\label{app:simple ex}

Here we derive the analytical expressions for the equilibria in Table \ref{tab: equilibria simple}. Assume to be in a period $t$, where the negative shock hits the economy, i.e., $\hat{M}_{t}=-r^{L}>0$. To solve for the possible equilibria of 
\begin{equation}
\hat{\pi}_{t+1|t}-\hat{M}_{t+1|t}-\max\left\{  -\mu,\psi\hat{\pi}_{t}\right\}
=0,\label{eq: univariate contTR}%
\end{equation}
one needs to solve for the expectations terms, that takes into account the possibility of ending up in the absorbing steady state. As we saw in the main text (see panel A in Figure \ref{fig: 2states}), when $\psi>1,$ there are two possible steady state outcomes in the absorbing state: PIR where the economy is at the intended steady state inflation target, i.e., $(\hat{M},\hat{\pi},\hat{R})=(0,0,0);$ ZIR where the economy steady state hits the ZLB constraint, i.e., $(\hat{M},\hat{\pi},\hat{R})=(0,\hat{\pi}^{ZIR}=-\mu,-\mu).$ Hence, in the temporary state in $t$, agents might expect to end up in PIR or in ZIR. If the agents expect to end up in PIR in the absorbing state, then the expectations terms will be
\begin{align}
E_{t}\left(  \hat{\pi}_{t+1}\right)   &  =p\hat{\pi}+(1-p)0=p\hat{\pi
},\label{app: eq: exppai}\\
E_{t}\left(  \hat{M}_{t+1}\right)   &  =p(-r^{L})+(1-p)0=-pr^{L},
\label{app: eq: expm}
\end{align}
and thus (\ref{eq: univariate contTR}) becomes
\begin{equation}
p\hat{\pi}=\max\left\{  -\mu,\psi\hat{\pi}\right\}  -pr^{L}.
\label{app: eq: panelc}
\end{equation}
Panel B in Figure \ref{fig: 2states} displays this equation in a graph.
There are two changes with respect to Panel A that shows the absorbing state given by the equation $\hat{\pi}=\max\left\{  -\mu,\psi\hat{\pi}\right\}.$ First the blue line is flatter, because the slope is $p$ rather than 1. Second, the negative $\hat{r}_{t}$ (i.e., positive $\hat{M}_t$) shifts the red curve upwards. The two equilibria in Panel B survive only if the real interest rate is not too low, in which case the red line shifts above the blue line and there is no possible equilibrium (incoherency). It is easy to show that the two equilibria in Panel B are given by
\begin{align}
\hat{\pi}_{t}  &  =\left\{
\begin{array}
[c]{ll}%
r^{L}\frac{p}{\psi-p}, & \text{if }\hat{M}_{t}=-r^{L}\in\left(  0,\mu
\frac{\psi-p}{\psi p}\right) \\
0, & \text{if }\hat{M}_{t}=0,
\label{app: eq: PIRPIR}
\end{array}
\right. \\
\hat{\pi}_{t}  &  =\left\{
\begin{array}
[c]{ll}
-r^{L}-\frac{\mu}{p}, & \text{if }\hat{M}_{t}=-r^{L}\in\left(  0,\mu\frac
{\psi-p}{\psi p}\right) \\
0, & \text{if }\hat{M}_{t}=0.
\label{app: eq: ZIRPIR}
\end{array}
\right.
\end{align}
These are the (PIR, PIR) and (ZIR, PIR) equilibria in Table \ref{tab: equilibria simple}. The second one implies a liquidity trap equilibrium in the temporary state. If $r^{L}<-\frac{\psi-p}{\psi p}\mu,$ there is no equilibrium.

If the agents expect to end up in ZIR in the absorbing state, instead,
then the expectations terms will be
\begin{align}
E_{t}\left(  \hat{\pi}_{t+1}\right)   &  =p\hat{\pi}+(1-p)(-\mu
),\\
E_{t}\left(  \hat{M}_{t+1}\right)   &  =p(-r^{L})+(1-p)0=-pr^{L},
\end{align}
and thus (\ref{eq: univariate contTR}) becomes%
\begin{equation}
p\hat{\pi}-\mu(1-p)=\max\left\{  -\mu,\psi\hat{\pi}\right\}  -pr^{L}.
\end{equation}
Panel C shows this case. With respect to Panel B, the blue line (LHS) now shifts down, because of the expectation of the possibility of a (permanent) liquidity trap equilibrium in the future (i.e., $(1-p)(-\mu)$). The two possible equilibria are
\begin{align}
\hat{\pi}_{t}  &  =\left\{
\begin{array}
[c]{ll}%
\frac{pr^{L}-\left(  1-p\right)  \mu}{\psi-p}, & \text{if }\hat{M}_{t}%
=-r^{L}\in\left(  0,\mu\frac{\psi-1}{\psi}\right) \\
-\mu, & \text{if }\hat{M}_{t}=0,
\label{app: eq: PIRZIR}
\end{array}
\right. \\
\hat{\pi}_{t}  &  =\left\{
\begin{array}
[c]{ll}%
-r^{L}-\mu, & \text{if }\hat{M}_{t}=-r^{L}\in\left(  0,\mu\frac{\psi-1}{\psi
}\right) \\
-\mu, & \text{if }\hat{M}_{t}=0.
\label{app: eq: ZIRZIR}
\end{array}
\right.
\end{align}
These are the (PIR, ZIR) and (ZIR, ZIR) equilibria in Table \ref{tab: equilibria simple}. Again, the second one implies a liquidity trap in the temporary state, and if $r^{L}<-\mu\frac{\psi-1}{\psi}$ there is no equilibrium.

\subsection{Further numerical results on multiple equilibria}\label{app: s: kfigures}
Figures \ref{fig: mult sol k=4} and \ref{fig: mult sol k=5} give solutions to the model of Section \ref{s: incompleteness} with $k=4$ and $k=5$ states.

\begin{figure}
\centering\includegraphics{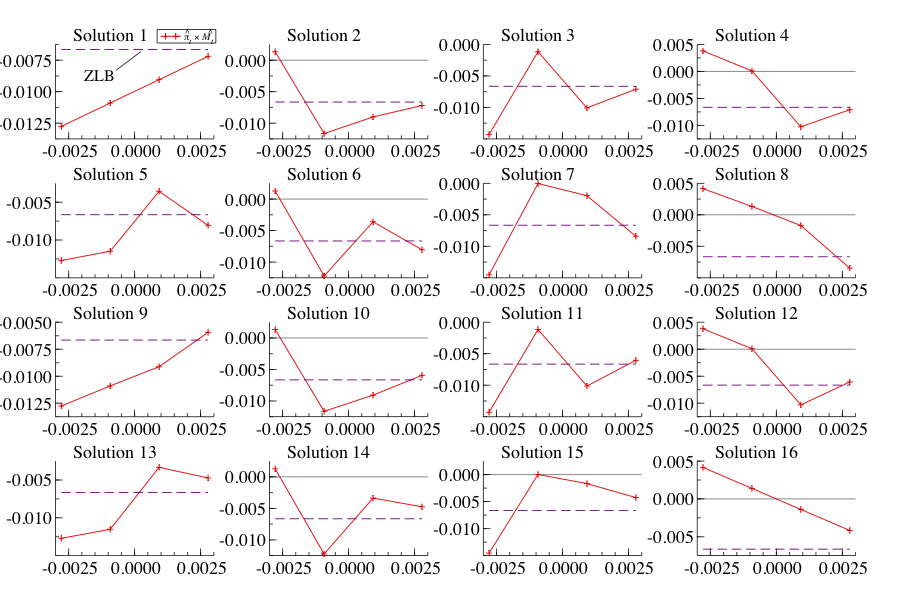}
\caption{The equilibria of model  $\hat{\pi}_{t|t+1} = max(-\mu,\psi \hat{\pi}_t)+\hat{M}_{t+1|t}$, when $\mu=0.01$, $\psi=1.5$ and $\hat{M}_t$ follows a 4-state Markov Chain with mean 0, conditional st.~dev.~$\sigma = 0.0007$, and autocorrelation $\rho=0.9$.\label{fig: mult sol k=4}} 
\end{figure}

\begin{figure}
\centering\includegraphics{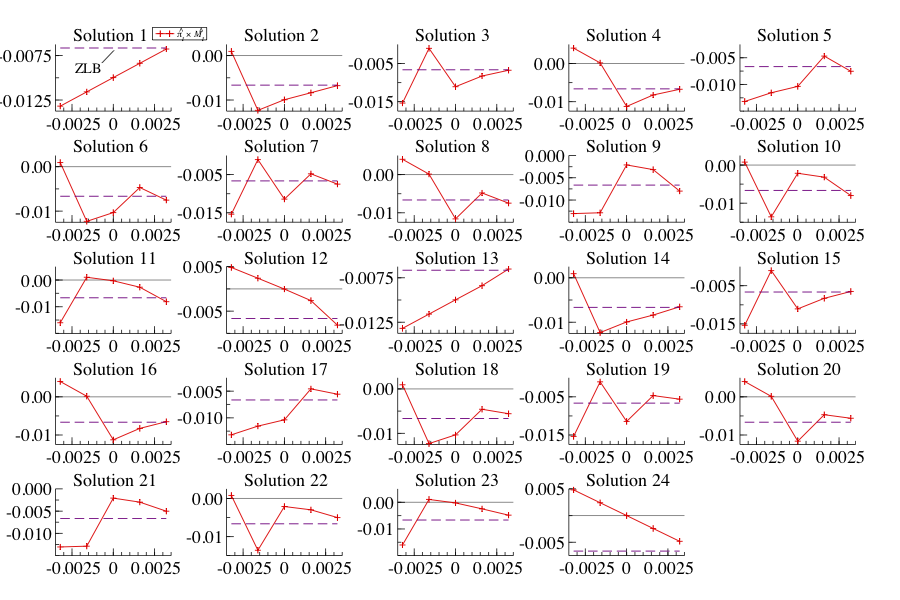}
\caption{The equilibria of model  $\hat{\pi}_{t|t+1} = max(-\mu,\psi \hat{\pi}_t)+\hat{M}_{t+1|t}$, when $\mu=0.01$, $\psi=1.5$ and $\hat{M}_t$ follows a 5-state Markov Chain with mean 0, conditional st.~dev.~$\sigma = 0.007$, and autocorrelation $\rho=0.9$.\label{fig: mult sol k=5}} 
\end{figure}

\subsection{A model with ZLB on inflation expectations}\label{app: s: ZLB expectations}
In this section we exemplify how the coherency of a model with a second inequality constraint can be analysed using the methodology of this paper. We consider \nameref{ex: ACS} with an addition ZLB on inflation expectation, motivated by \cite{GorodnichenkoSergeyev21}. The model is given by
\begin{equation}\label{eq: univariate}
\max\left\{  \hat{\pi}_{t+1|t},0\right\}  =\max\left\{  -\mu,\psi\hat{\pi}%
_{t}\right\}  +\hat{M}_{t+1|t}.%
\end{equation}
Suppose $\hat{M}_{t}$ follows a $k$-state Markov chain, with states $m$ and
transition probability kernel given by the $k\times k$ matrix $K$ with
$K_{ij}=\Pr\left(  \hat{M}_{t+1}=m_{j}|\hat{M}_{t}=m_{i}\right)  $.
\[
E_{t}\left(  \hat{\pi}_{t+1}\right)  =K\pi,\quad E_{t}\left(  \hat{M}%
_{t+1}\right)  =Km,
\]
and so the equation to be solved, (\ref{eq: univariate}), can be written as%
\[
\max\left(  K\pi,0\right)  =Km+\max\left(  -\mu\iota_{k},\psi\pi\right)  .
\]
This is a system of piecewise linear equations with two inequality constraints
$K\pi\geq0$ and $\pi\geq-\mu/\psi.$ This defines at most $4^{k}$ cones because
some inequality combinations may be impossible. The following graph
illustrates for the case $k=2$ where the first state is transitory and
persists with probability $p$ and the second state is absorbing. The ZLB on
the interest rate generates the inequalities we defined previously, i.e.,
$\pi_{t}>-\mu/\psi$ is a PIR and $\pi_{t}\leq-\mu$ is a ZIR. The ZLB on
expectations is%
\[
E_{t}\left(  \pi_{t+1}\right)  \geq0\Rightarrow\left\{
\begin{array}
[c]{ll}%
p\pi_{1}+\left(  1-p\right)  \pi_{2}\geq0, & \text{if }\hat{M}_{t}\text{
transitory}\\
\pi_{2}\geq0, & \text{if }\hat{M}_{t}\text{ absorbing.}%
\end{array}
\right.
\]
The possible regime configurations induced by combining the two inequalities are depicted in Figure \ref{fig: ACS with ZLB on exp}. We see that
these inequalities split $\Re^{2}$ into 10 cones. Let Z,P denote the interest
rate regimes ZIR and PIR, respectively, and Z$_{e}$,P$_{e}$ the expectations
regime. So, (P,P;P$_{e}$,P$_{e}$) denotes positive interest rate and positive
expectations in both states. The model can then be written in canonical form (\ref{eq: canon})
with $\mathcal{A}_{J}$ for each cone defined as shown in Table \ref{t: ZLB expectations}.
\begin{figure}
    \centering
    \includegraphics[trim={6cm 2cm 4cm 1.8cm},clip]{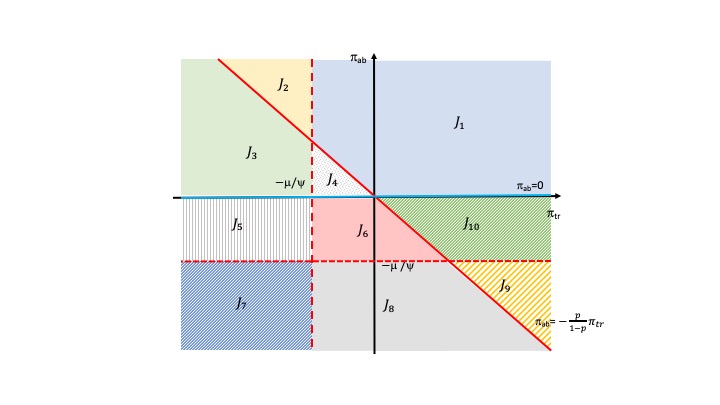}
    \caption{Combination of regimes in \nameref{ex: ACS} with additional ZLB on inflation expectations. Dotted lines delineate interest rate regimes. Solid lines delineate expectations regimes in transitory (red) and absorbing (blue) states. Regimes are denoted by $J_i$.}
    \label{fig: ACS with ZLB on exp}
\end{figure}
\begin{table}[tbp] \centering
\begin{tabular}
[c]{lcc}\hline\hline
Coefficient matrix ($2\times2$) & \multicolumn{1}{|c}{Int. rate regime} &
\multicolumn{1}{|c}{Infl. exp. regime }\\\hline
$\mathcal{A}_{J_{1}}=K-\psi I$ & \multicolumn{1}{|c}{P,P} &
\multicolumn{1}{|c}{P$_{e}$,P$_{e}$}\\
$\mathcal{A}_{J_{2}}=K-\psi e_{2}e_{2}^{T}$ & \multicolumn{1}{|c}{Z,P} &
\multicolumn{1}{|c}{P$_{e}$,P$_{e}$}\\
$\mathcal{A}_{J_{3}}=e_{2}e_{2}^{T}K-\psi e_{2}e_{2}^{T}$ &
\multicolumn{1}{|c}{Z,P} & \multicolumn{1}{|c}{Z$_{e}$,P$_{e}$}\\
$\mathcal{A}_{J_{4}}=e_{2}e_{2}^{T}K-\psi I$ & \multicolumn{1}{|c}{P,P} &
\multicolumn{1}{|c}{Z$_{e}$,P$_{e}$}\\
$\mathcal{A}_{J_{5}}=-\psi e_{2}e_{2}^{T}$ & \multicolumn{1}{|c}{Z,P} &
\multicolumn{1}{|c}{Z$_{e}$,Z$_{e}$}\\
$\mathcal{A}_{J_{6}}=-\psi I$ & \multicolumn{1}{|c}{P,P} &
\multicolumn{1}{|c}{Z$_{e}$,Z$_{e}$}\\
$\mathcal{A}_{J_{7}}=0$ & \multicolumn{1}{|c}{Z,Z} &
\multicolumn{1}{|c}{Z$_{e}$,Z$_{e}$}\\
$\mathcal{A}_{J_{8}}=-\psi e_{1}e_{1}^{T}$ & \multicolumn{1}{|c}{P,Z} &
\multicolumn{1}{|c}{Z$_{e}$,Z$_{e}$}\\
$\mathcal{A}_{J_{9}}=e_{1}e_{1}^{T}K-\psi e_{1}e_{1}^{T}$ &
\multicolumn{1}{|c}{P,Z} & \multicolumn{1}{|c}{P$_{e}$,Z$_{e}$}\\
$\mathcal{A}_{J_{10}}=e_{1}e_{1}^{T}K-\psi I$ & \multicolumn{1}{|c}{P,P} & \multicolumn{1}{|c}{P$_{e}$,Z$_{e}%
$}\\\hline\hline
\end{tabular}
\caption{Coefficients of canonical representation (\ref{eq: canon}) of \nameref{ex: ACS} with an additional ZLB on inflation expectations.}\label{t: ZLB expectations}%
\end{table}%
%EndExpansion

We see by inspection that the CC\ condition in the GLM Theorem is violated,
since some of the $\mathcal{A}_{J}$ are evidently singular. Even if we
restrict attention to $\pi_{2}\geq0,$ i.e., regimes $J_{1}$ to $J_{4}$, the
determinants are (for $\psi>1$): $\det\mathcal{A}_{J_{1}}=\left(
\psi-1\right)  \left(  \psi-p\right)  >0$, $\det\mathcal{A}_{J_{2}}=p\left(
1-\psi\right)  <0$, $\det\mathcal{A}_{J_{3}}=0$ and $\det\mathcal{A}_{J_{4}%
}=\psi\left(  \psi-1\right)  >0$. So, this model is not generically coherent.

\end{document}